\documentclass[twoside,10pt]{amsart}

\usepackage{lmodern}
\usepackage[T1]{fontenc}
\usepackage[utf8]{inputenc}
\usepackage[english]{babel} %
\usepackage{array}		%
\usepackage{graphicx}   %
\usepackage{enumerate}  %
\usepackage[hidelinks]{hyperref}  %
\usepackage{amsmath}
\usepackage{amssymb}
\usepackage{mathtools}	%
\mathtoolsset{centercolon}	%
\usepackage{tikz-cd}    %
\usetikzlibrary{babel}  %

\usepackage{tikz}				%
\usepackage{pgfplots}
\usepackage[linesnumbered, lined, boxed, norelsize]{algorithm2e}

\newtheorem{thm}{Theorem}[section]
\newtheorem{prop}[thm]{Proposition}
\newtheorem{lemma}[thm]{Lemma}
\newtheorem{cor}[thm]{Corollary}

\newtheorem*{conjecture*}{Conjecture}

\newtheorem{question}[thm]{Question}
\newtheorem*{question*}{Question}
\newtheorem{definition}[thm]{Definition}

\theoremstyle{definition}
\newtheorem{remark}[thm]{Remark}

\newtheorem{ex}[thm]{Example}
\newtheorem{notation}[thm]{Notation}

\numberwithin{equation}{section}
\numberwithin{figure}{section}
\numberwithin{subsection}{section}

\newcommand{\N}{\mathbb{N}}

\newcommand{\R}{\mathbb{R}}
\newcommand{\C}{\mathbb{C}}

\newcommand{\Z}{\mathbb{Z}}

\renewcommand{\H}{\mathcal{H}}

\newcommand\restr[2]{{
  \left.\kern-\nulldelimiterspace 
  #1
  \right|_{#2} 
  }}
\newcommand\supind[1]{{\smash[t]{(#1)}}}

\let\originalleft\left
\let\originalright\right
\renewcommand{\left}{\mathopen{}\mathclose\bgroup\originalleft}
\renewcommand{\right}{\aftergroup\egroup\originalright}

\DeclareMathOperator{\sgn}{sgn}
\DeclareMathOperator{\Sym}{Sym}

\DeclareMathOperator{\SL}{SL}
\DeclareMathOperator{\GL}{GL}
\DeclareMathOperator{\id}{id}

\DeclareMathOperator{\diag}{diag}

\hypersetup{bookmarksdepth=3}
\usepackage{fancyhdr}
\usepackage{fullpage}
\usepackage{subfigure}
\newcommand{\defstyle}{\textbf}  %

\newcommand{\sphere}[1][2]{\mathrm{S}^{#1}} %
\newcommand{\surf}[1][]{\mathcal{S}_{#1}}   %

\newcommand{\VV}{V} 											%
\newcommand{\VQ}{\ensuremath{\mathrm{q}}} %
\newcommand{\Vs}[1]{s_{#1}}   						%
\newcommand{\Vr}[1]{r_{#1}}   						%
\newcommand{\Vt}[1]{t_{#1}}   						%

\newcommand{\Vf}[1]{\eta_{#1}}						%
\newcommand{\Vg}[1]{\theta_{#1}}          %

\newcommand{\Vc}{c}                       %
\newcommand{\Vtr}{t}                      %

\newcommand{\Vp}[1][]{\ensuremath{p_{#1}}}								%
\newcommand{\Vpt}[1][]{\ensuremath{\widetilde{p}_{#1}}}   %
\newcommand{\setP}[1][2d]{\ensuremath{\mathcal{I}_{#1}}}  %
\newcommand{\setPT}[1][2d]{\ensuremath{\mathcal{J}_{#1}}} %
\newcommand{\vBasis}[1]{\ensuremath{\mathbb{B}_{#1}}}     %
\newcommand{\hBasis}[1]{\ensuremath{\mathcal{B}_{#1}}}    %
\newcommand{\bEl}[1][]{\ensuremath{b_{#1}}}               %
\newcommand{\hbEl}[1][]{\ensuremath{b_{#1}}}              %

\newcommand{\xx}{\ensuremath{x}}
\newcommand{\yy}{\ensuremath{y}}
\newcommand{\zz}{\ensuremath{z}}
\newcommand{\xyz}{\ensuremath{\xx,\yy,\zz}}

\newcommand{\G}[1][3]{\ensuremath{\mathrm{O}_{#1}}}
\newcommand{\GO}[1][3]{\ensuremath{\mathrm{O}_{#1}}}     %
\newcommand{\B}[1][3]{\ensuremath{\mathrm{B}_{#1}}}
\newcommand{\GB}[1][3]{\ensuremath{\mathrm{B}_{#1}}}    %
\newcommand{\Ggen}{\ensuremath{\mathrm{G}}}
\newcommand{\Bgen}{\ensuremath{\mathrm{B}}}
\newcommand{\GS}[1][3]{\ensuremath{\mathfrak{S}_{#1}}}  %
\newcommand{\GA}[1][3]{\ensuremath{\left\{1,-1\right\}^3}} %

\newcommand{\Cl}[1][]{\ensuremath{\lambda_{#1}}}
\newcommand{\Cr}[1][]{\ensuremath{\alpha_{#1}}}
\newcommand{\Cs}[1][]{\ensuremath{\beta_{#1}}}
\newcommand{\Ct}[1][]{\ensuremath{\gamma_{#1}}}

\newcommand{\Ilw}[1][]{\ensuremath{\tilde{\mathrm{\ell}}_{#1}}}
\newcommand{\Irw}[1][]{\ensuremath{\tilde{\mathrm{a}}_{#1}}}
\newcommand{\Isw}[1][]{\ensuremath{\tilde{\mathrm{b}}_{#1}}}
\newcommand{\Itw}[1][]{\ensuremath{\tilde{\mathrm{c}}_{#1}}}

\newcommand{\Il}[1][]{\ensuremath{{\mathrm{\ell}}_{#1}}}
\newcommand{\Ir}[1][]{\ensuremath{{\mathrm{a}}_{#1}}}
\newcommand{\Is}[1][]{\ensuremath{{\mathrm{b}}_{#1}}}
\newcommand{\It}[1][]{\ensuremath{{\mathrm{c}}_{#1}}}

\newcommand{\Eqm}[1][]{\ensuremath{{\rm W}_{#1}}}

\title{Rational invariants of even ternary forms\\ under the orthogonal group}
\author{Paul Görlach}
\address{Paul Görlach - INRIA Méditerranée, France; Max Planck Institute for Mathematics in the Sciences, Germany}
\email{goerlach@mis.mpg.de}

\author{Evelyne Hubert}
\address{Evelyne Hubert - INRIA Méditerranée, France}
\email{evelyne.hubert@inria.fr}

\author{Théo Papadopoulo}
\address{Théo Papadopoulo - INRIA Méditerranée, France}
\email{theodore.papadopoulo@inria.fr}

\begin{document}

\begin{abstract}
In this article we determine a generating set of rational invariants of minimal 
cardinality for the action of the orthogonal group $\G$ on the space 
$\R[x,y,z]_{2d}$ of ternary forms of even degree $2d$. The construction relies 
on two key ingredients: On one hand, the Slice Lemma allows us to reduce the 
problem to determining the invariants for the action on a subspace of the 
finite subgroup $\GB$ of signed permutations. 
On the other hand, our construction relies in a fundamental way on specific bases of harmonic polynomials. These bases provide maps with prescribed $\GB$-equivariance properties. Our explicit construction of these bases should be relevant well beyond the scope of this paper. 
The expression of the $\GB$-invariants can then be given in a compact form as the composition of two equivariant maps. Instead of providing (cumbersome) explicit expressions for the $\G$-invariants, we provide efficient algorithms for their evaluation and rewriting. We also use the constructed $\GB$-invariants to determine the $\G$-orbit locus and provide an algorithm for the inverse problem of finding an element in $\R[x,y,z]_{2d}$ with prescribed values for its invariants. These computational issues are relevant in brain imaging.\\[5pt]


\end{abstract}

\maketitle

\setcounter{tocdepth}{1} %
{\setlength{\parskip}{0pt}\tableofcontents}


\section{Introduction}
Invariants are useful for classifying 
objects up to the action of a group of 
transformations.
In this article we determine a set of generating rational invariants of minimal cardinality for the action of the orthogonal group $\G$ on the space $\R[x,y,z]_{2d}$ of ternary forms of even degree $2d$. 
We do not give the explicit expression for these invariants, but provide an algorithmic way of evaluating them for any ternary form.

Classical Invariant Theory \cite{GY03} is centered around the action of the general linear group on homogeneous polynomials, with an emphasis on binary forms. Yet, the orthogonal group arises in applications as the relevant group of transformations, especially in three-dimensional space. 
Its relevance to brain imaging is the original motivation for the present article. 

Computational Invariant Theory \cite{DK15,GY03,Stu08} has long focused on polynomial invariants. In the case of the group $\G$, any two real orbits are separated by a generating set of polynomial invariants. The generating set can nonetheless be very large. For instance a generating set of $64$ polynomial invariants for the action of $\G$ on $\R[x,y,z]_4$ is determined as a subset of a  minimal generating set of polynomial invariants of the action of $\G$ on the \emph{elasticity tensor} in \cite{AKO17}. There, the problem is mapped to the joint action of $\SL_2(\C)$ on binary forms of different degrees and resolved by Gordan's algorithm \cite{GY03,Oli16} so that the invariants are given as transvectants.

A generating set of rational invariants separates general orbits \cite{PV94,Rosenlicht56} -- this 
remains true for any group, even for non-reductive groups. Rational invariants 
can thus prove to be sufficient, and sometimes more relevant, in applications 
\cite{Hubert12Labahn,Hubert13Labahn,Hubert16Labahn} and in connection with 
other mathematical disciplines \cite{Hubert07s,Hubert12focm}. 
A practical and very general algorithm to compute a generating set of these first appeared in \cite{HK08}; see also \cite{DK15}.
The case of the action of $\G$ on $\R[x,y,z]_4$, a $15$-dimensional space, is nonetheless not easily tractable by this algorithm. 
In the case of $\R[x,y,z]_4$, the $12$ generating invariants we construct in 
this article are seen as being uniquely determined by their restrictions to a 
\emph{slice} $\Lambda_4$, which is here a $12$-dimensional subspace. The 
knowledge of these restrictions is proved to be sufficient to evaluate the 
invariants at any point in the space $\R[x,y,z]_4$.
The underlying \emph{slice method} is a technique used to show rationality of 
invariant fields \cite{CTS07}. We demonstrate here its power for the 
computational aspects of Invariant Theory.

More generally, by virtue of the so-called \emph{Slice Lemma}, the field of 
rational invariants of the action of $\G$ on $\R[x,y,z]_{2d}$ is isomorphic to 
the field of rational invariants for the action of the finite subgroup $\GB$ of 
signed permutations on a subspace  $\Lambda_{2d}$ of $\R[x,y,z]_{2d}$. Working 
with this isomorphism, we provide efficient algorithms for the evaluation of a 
(minimal) generating set of rational $\G$-invariants based on a (minimal) 
generating set of rational $\GB$-invariants. Finding an element of 
$\R[x,y,z]_{2d}$ with prescribed values of $\G$-invariants and rewriting any 
$\G$-invariant in terms of the generating set are also made possible through 
the specific generating sets of rational $\GB$-invariants we construct.

$\R[x,y,z]_{2d}$ decomposes into a direct sum of $\G$-invariant vector spaces determined by the harmonic polynomials 
of degree $2d, 2d-2,\ldots, 2, 0$.
The construction of the $\GB$-invariants relies in a fundamental way on specific bases of harmonic polynomials. These bases provide maps with prescribed $\GB$-equivariance properties. Our explicit construction of these bases should be relevant well beyond the scope of this paper. 
The explicit expression of the $\GB$-invariants can then be given in a compact form as the composition of two equivariant maps. Both the rewriting and the inverse problem can be made explicit for this well-structured set of invariants.

The whole construction is first made explicit for $\R[x,y,z]_4$ and then extended to $\R[x,y,z]_{2d}$. It should then be clear how one can obtain the rational invariants of the joint action of $\G$ on similar spaces, as for instance the elasticity tensor.

In the rest of this section we first introduce the geometrical motivation and then the context of application for our constructions. 
Notations, preliminary material and a formal statement of the problems appear then in Section~\ref{sec:Preliminaries}. In Section~\ref{sec:SliceMethod}, we describe the technique for reducing the construction of invariants for the action of $\G$ on $\R[x,y,z]_{2d}$ to the one for the action of a finite subgroup, the signed symmetric group, on a subspace. Based on this, we construct a generating set of rational invariants for the case of ternary quartic forms in Section~\ref{sec:Quartics}. We extend this approach in Section~\ref{sec:B3Harmonics} to ternary forms of arbitrary even degree; central there is the construction of bases of vector spaces of harmonic polynomials with specific equivariant properties with respect to the signed symmetric group. Finally, we solve the principal algorithmic problems associated with rational invariants in Section~\ref{sec:AlgorithmicSolutions}.

\paragraph{\textbf{Acknowledgments}} Paul Görlach was partly funded by INRIA 
Mediterranée \emph{Action Transverse}. 
Evelyne Hubert wishes to thank  Rachid Deriche, Frank Grosshans, Boris Kolev for discussions and valuable pointers.
Théo Papadopoulo receives funding 
from 
the ERC Advanced Grant No 694665 : CoBCoM - Computational Brain 
Connectivity Mapping.
 
  \subsection{Motivation: spherical functions up to rotation}  \label{ssec:motivation}
 
In geometric terms, the results in this paper allow one to determine when two 
general centrally symmetric closed surfaces in $\R^3$ differ only by a rotation.
\begin{figure}
  \centering
  \includegraphics[width=0.3\textwidth]{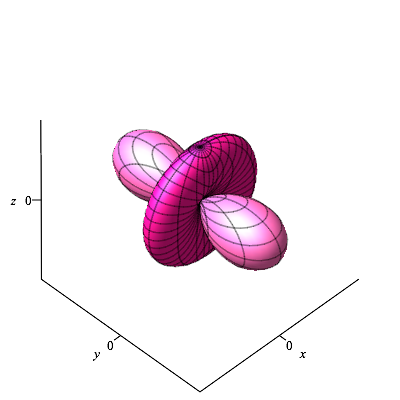}
  \includegraphics[width=0.3\textwidth]{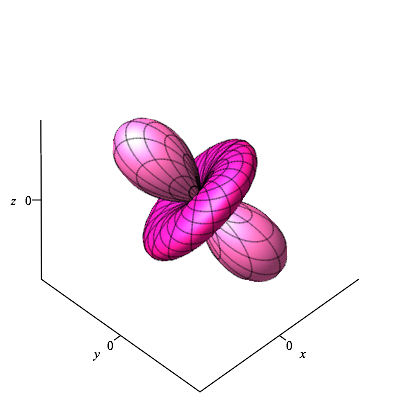}
  \caption{Two centrally symmetric surfaces only differing by a rotation.}
  \label{fig:rotateSurfaces}
\end{figure}

The surfaces we consider are given by a continuous function $f \colon \sphere 
\to \R$ defined on the unit sphere $\sphere := \{(x,y,z) \in \R^3 \mid 
x^2+y^2+z^2=1\}$. The surface is then defined as the set of points
  \[\surf := \{\left(x\,f(x,y,z), y\,f(x,y,z),z\,f(x,y,z)\right) \in \R^3 \mid (x,y,z) \in \sphere\} \subset \R^3,\] 
i.e.\ for each point on the unit sphere we rescale its distance to the origin according to the function $f$. For example, the surface described by a constant function $f \equiv r$ (for some $r \in \R$) is
the sphere with radius $|r|$ centered at the origin. With more general 
functions 
$f$, a large variety of different surfaces can be described. For the surface to 
be symmetric with respect to the origin, one needs the function $f$ to satisfy 
the property
\begin{equation} \label{eq:SurfaceSymmetry}
  f(p) = f(-p) \quad \forall p \in \sphere.
\end{equation}

Since the unit sphere $\sphere \subset \R^3$  is a compact set, the Stone--Weierstraß Theorem implies that a given continuous function $\sphere \to \R$ can be approximated arbitrarily well by polynomial functions, i.e.\ by functions $f \colon S^2 \to \R$ of the form 
\begin{equation}\label{eq:polyd}
  f(x,y,z) = \sum_{i+j+k \leq d} a_{ijk} \: x^i y^j z^k.
\end{equation}
The symmetry property (\ref{eq:SurfaceSymmetry}) is then fulfilled if and only $a_{ijk} = 0$ whenever $i+j+k$ is odd, in other words, $f$ must only consist of even-degree terms. We can then rewrite $f$ in such a way that all its monomials are of the same degree, by multiplying monomials of small degree with a suitable power of $x^2+y^2+z^2$ (which does not change the values of $f$ since $x^2+y^2+z^2=1$ for all points $(x,y,z)$ on the sphere).

We therefore model closed surfaces that are centrally symmetric with respect to the origin by 
$f \colon \sphere \to \R$, a polynomial function of degree $2d$:
  \[f(x,y,z) = \sum_{i+j+k=2d} a_{ijk} \: x^i y^j z^k.\] 
The modeled surface can then be encoded \emph{exactly} by simply storing the $\frac{1}{2}(2d+2)(2d+1)$ numbers $a_{ijk} \in \R$.
However, from such an encoding of a surface via the coefficients $a_{ijk}$ of its defining polynomial, it is not immediately apparent when the defined surfaces have the same geometric shape, only differing by a rotation. As an example, Figure~\ref{fig:rotateSurfaces} depicts the two surfaces $\surf[1]$ and $\surf[2]$ described by %
\begin{align*}
  f_1(x,y,z) &:= 18x^2-27y^2+18z^2 \text{ and} \\
  f_2(x,y,z) &:= 13x^2+20xy-20xz-2y^2+40yz-2z^2, 
\end{align*}
respectively, whose numerical description in terms of coefficients look very distinct, but whose shapes look very similar. Indeed $\surf[2]$ arises from $\surf[1]$ by applying the rotation determined by the matrix
\[\frac{1}{3} \begin{pmatrix} 2 & -1 & -2 \\ 2 & 2 & 1 \\ 1 & -2 & 2 \end{pmatrix},\]
so that 
  \[f_2(x,y,z) = f_1\left(\frac{2x+2y+z}{3}, \frac{-x+2y-2z}{3}, \frac{-2x+y-2z}{3}\right).\]

In general, we want to consider surfaces to be of the same shape if they differ by a rotation or, equivalently, by an orthogonal transformation: Since we only consider surfaces symmetric with respect to the origin, two surfaces differing by an orthogonal transformation also differ by a rotation.

The question which arises is: \emph{How can we (algorithmically) decide whether two surfaces only differ by an orthogonal transformation, only by examining the coefficients of their defining polynomial?}

Regarding this question, we should be aware that 
if $f \colon \sphere \to \R$ describes $\surf \subset \R^3$, then its negative $-f$ describes the same set $\surf \subset \R^3$. With the intuition of describing the surface by deforming the unit sphere, this ambiguity corresponds to turning the closed surface $\surf$ inside out. We typically want to think of the surfaces defined by $f$ and by $-f$ as two distinct objects (even though they are equal as subsets of $\R^2$). For example, the surface defined by the constant function $f \equiv 1$ is the unit sphere, while the surface defined by $f \equiv -1$ should be considered as the unit sphere turned inside out.

Then the question above corresponds to the following algebraic question:

\begin{question} \label{qu:decideSimilarity}
  Given two polynomials $f = \sum_{i+j+k=2d} a_{ijk} x^i y^j z^k$ and $g = \sum_{i+j+k=2d} b_{ijk} x^i y^j z^k$, how can we decide in terms of their coefficients $a_{ijk}$ and $b_{ijk}$ whether or not there exists an orthogonal transformation $\varphi \colon \R^3 \to \R^3$ such that $f = g \circ \varphi$ as functions $\R^3 \to \R$?
\end{question}

Once we can decide whether two polynomials define equally shaped surfaces, a natural question is how to uniquely encode the shape of a surface, i.e.\ an equivalence class of surfaces up to orthogonal transformations. The corresponding question for polynomials is:

\begin{question} \label{qu:encodeEqClass} 
  How can we encode in a unique way equivalence classes of polynomials up to orthogonal transformations?
\end{question}

For the purpose of illustration we discuss next the case of homogeneous polynomials of degree two, i.e.\ \emph{quadratic forms}. In the following section, we shall describe the mathematical setup for addressing Questions~\ref{qu:decideSimilarity} and~\ref{qu:encodeEqClass} with the rational invariants of the action of the orthogonal group on ternary forms of even degree.

\subsubsection*{Illustration for the case of quadratic surfaces}

Quadratic forms are in one-to-one correspondence with symmetric $3 \times 3$-matrices, as indeed we can write \eqref{eq:polyd} as:
\begin{equation} \label{eq:GramianMatrix}
  f(\xyz) =
\begin{pmatrix} \xx & \yy & \zz \end{pmatrix} 
\underbrace{\begin{pmatrix} a_{200} & \frac{1}{2}a_{110} & \frac{1}{2}a_{101} \\[2pt]
  \frac{1}{2}a_{110} & a_{020} & \frac{1}{2} a_{011} \\[2pt]
   \frac{1}{2}a_{101} & \frac{1}{2}a_{011} & a_{002} \end{pmatrix}}_{=: A} 
\begin{pmatrix} \xx \\ \yy \\ \zz \end{pmatrix}
\end{equation}

When we compose $f$ with a rotation, given by a matrix $R$, we obtain another homogeneous polynomial of degree $2$. The defining symmetric matrix (as in \ref{eq:GramianMatrix}) of the obtained polynomial is then $RAR^T$. For two symmetric matrices $A$ and $B$, Linear Algebra then tells us that there is an orthogonal matrix $R$ such that {$B=RAR^T$} if and only if $A$ and $B$ have the same eigenvalues. Yet, eigenvalues are algebraic functions of the entries of a matrix and thus cannot be easily expressed symbolically. It is thus easier to compare the coefficients of the characteristic polynomials. Up to scalar factors, these are:
\begin{equation} \label{eq:quadrInvariants}
\begin{array}{c}
e_1:=a_{200}+a_{020}+a_{002}, \\
e_2:=4a_{020}a_{200}+4a_{002}a_{020}+4a_{002}a_{200}-a_{110}^2-a_{101}^2-a_{011}^2, \\ 
e_3:=4a_{002}a_{020}a_{200}+a_{101}a_{110}a_{011}-a_{101}^2a_{020}-a_{011}^2a_{200}-a_{002}a_{110}^2.
\end{array}
\end{equation}

The result is that two homogeneous polynomial of degree 2 are obtained from one another by an orthogonal transformation if and only if the functions $e_1, e_2, e_3$ of their coefficients take the same values. This answers Question~\ref{qu:decideSimilarity} in this case. The functions $e_1, e_2, e_3$ also provide a solution to Question~\ref{qu:encodeEqClass}. To each polynomial of degree $2$ we can associate a point in $\R^3$ whose components are the values of the functions $e_1, e_2, e_3$ for this polynomial. All equivalent polynomials will be mapped to the same point. One can check for example that this point is $(9,-2592,-34992)$ for both $f_1$ and $f_2$ defined above.

\subsection{Biomarkers in neuroimaging} \label{ssec:neuro}

The cerebral white matter is the complex wiring that enables communication between the various regions of grey matter of the brain. Its integrity plays a pivotal role in the proper functionning of the brain.  Diffusion MRI (dMRI) has the ability to measure the  diffusion of water molecules. 
The architecture of white matter, or other fibrous structures can be inferred from such measurements in-vivo and non-invasively. 
This paragraph introduces our original motivation for the problem solved in this article and thus follows 
only a single sub-thread of this research. The interested reader is referred to~\cite{JohansenBerg14,Jones11} for a more globalview of this field and  the progress  accomplished on the many challenges 
 (acquisition protocols, mathematical models, reconstruction algorithms, $\ldots$) since its inception \cite{le-bihan-breton-etal:86}. 

In the popular technique known as  diffusion tensor imaging (DTI) 
the aquired signal is modeled as $S (b, v) = S_0 \exp(-b\, v^t A v)$, 
where $b$ includes the  control parameters of the aquisition protocol while 
$v$ is an orientation vector. 
In this model, the \emph{anisotropic apparent diffusion coefficient} 
is the quadratic function 
$v \mapsto v^t A v$ given by a positive symmetric matrix $A$
 which admittedly reflects the microstructure of the tissue. 
Such a signal $S(b,v)$ is sampled, measured and reconstructed  at each voxel 
(that is, the elementary volume in which a 3D image is decomposed), 
meaning that a symmetric positive definite matrix $A$ is estimated  at each voxel.
The eigenvector associated to the largest eigenvalue is then an indicator 
of the orientation of the dominant fibre bundle.

Due to the resolution of dMRI images, i.e. the size of the voxel compared 
to the diameter of the fibers, this is 
only a gross approximation of the underlying tissue microstructure.
DTI cannot discern between complex fibre bundle configurations
such as fiber crossings or kissings. 
With the advent of higher angular resolution   diffusion imaging or Q-ball imaging,
higher order models of the diffusion signal were thus investigated \cite{Ghosh16}.
Figure~\ref{fig:dMRI} illustrates, at a sample of 
voxels, diffusion as  the graph of a function on the sphere.
As such, expansions into spherical harmonics is  a natural choice for  the  
the \emph{apparent distribution coefficient}, \emph{diffusion kurtosis tensor} 
or the \emph{fiber orientation distribution} \cite{Schultz13}.
The usual (real) spherical harmonics of order $d$
are the restrictions to the sphere of 
homogeneous harmonic ternary polynomials (called \emph{tensor} in this context) 
of degree $d$ \cite{AH12}, once expressed in polar coordinates. 
Diffusivity is mostly  \emph{antipodal}, i.e. centrally 
symmetric,  and thus only even degree  polynomials 
(equivalently even order spherical harmonics)  are considered. 

\begin{figure}
    \includegraphics[width=\textwidth]{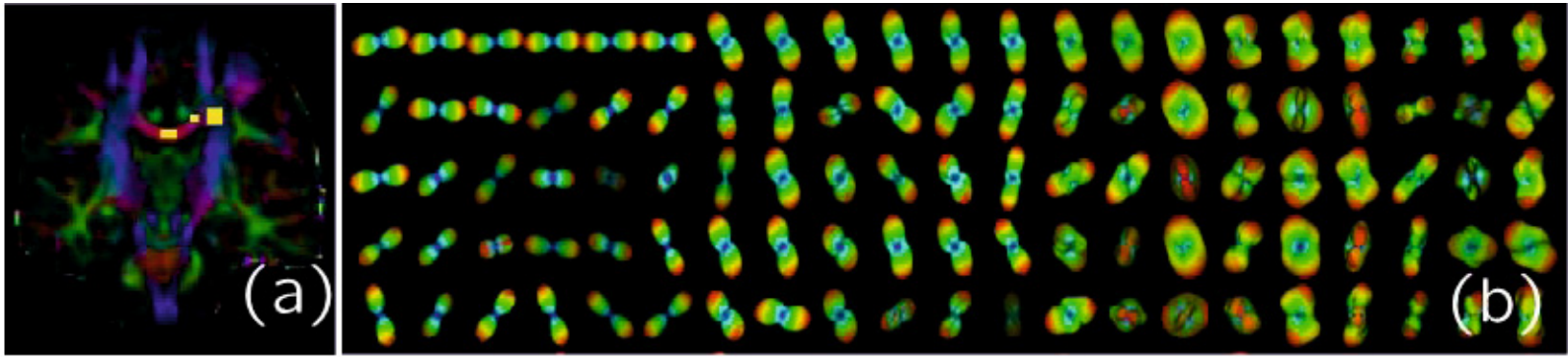}
    \caption{Slice of a brain diffusion image (a). The yellow boxes indicate the regions that are zoomed in (b) illustrating the diffusion signal at each voxel.
             From~\cite{Pap14}.}
    \label{fig:dMRI}
\end{figure}

It is desirable to extract meaningful scalars, known as \emph{biomarkers},
from the reconstructed signal. 
In the case of DTI,
mean diffusivity and fractional anisotropy, given respectively as 
$$\bar{\lambda} = \frac{\lambda_1+\lambda_2+\lambda_3}{3} = \frac{1}{3} \mathrm{tr}(A) 
\quad \hbox{ and }
\sqrt{\frac{(\lambda_1-\bar\lambda)^2+(\lambda_2-\bar\lambda)^2+(\lambda_3-\bar\lambda)^2}{\lambda_1^2+\lambda_2^2+\lambda_3^2}} = \sqrt{1 -\frac{1}{3} \frac{\mathrm{tr}(A)^2 }{\mathrm{tr}(A^2)}}$$
in terms of the eigenvalues $\lambda_1,\lambda_2,\lambda_3$ of $A$,
are relevant 
to detect various diseases such as 
Alzheimer, Parkinson, or writer's cramp \cite{delmaire-vidailhet-etal:09,li-messe-etal:10}. 
These quantities are naturally invariant under rotation. 
There thus has been efforts \cite{basser-pajevic:07,GPD12a,Pap14,CV15}, in the computational dMRI literature, to find rotation invariants
for the higher order models of diffusion, with some successes up to order $6$.  The present article provides a complete (functionally independent
and generating) set of rational invariants, for all even degrees $2d$.  It belongs to a wider project to identify the biologically relevant ones
for a given pathology. 
In this respect, it is interesting to be able to reconstruct the diffusion signal   
corresponding to a given set of invariants.
This would allow practitioners to interpret the invariants  as 
\emph{shape descriptors} and relate the diffusion abnormalities with the pathology. 
This is an inverse problem that we shall solve for the invariants constructed in this paper.

\section{Preliminaries} \label{sec:Preliminaries}

In order to set the notations, we review the definitions of the action of $\G$ on the vector space of forms $\VV_{ d} := \R[\xyz]_{d}$ and of its rational invariants. We then elaborate on the Harmonic Decomposition of $\VV_{2d}$. 
  \subsection{The action of the orthogonal group on homogeneous polynomials} \label{ssec:GroupAction}
We start out with a few notational conventions. We denote by $\R[\xyz]_d$ the set of homogeneous polynomials  
  \[f =\sum_{i+j+k =  d} a_{ijk} \,\xx^i \yy^j \zz^k\] 
of degree~$d$ in the three variables $\xyz$ with real coefficients $a_{ijk} \in \R$. By $\G \subset \R^{3\times 3}$ we denote the group of orthogonal matrices
and we consider the action of $\G$ on $\R[\xyz]$ given by
\[\begin{array}{ccl} \G \times \R[\xyz]_d & \to  & \R[\xyz]_d 
   \\ (g, f) & \mapsto & f\circ g^{-1} . \end{array}\]
By $f \circ g^{-1}$, we mean the composition of the orthogonal transformation $g^{-1} \colon \R^3 \to \R^3$ with the polynomial function $f \colon \R^3 \to \R$, resulting in a different polynomial function that we denote by $gf$. Note that $gf$ is again homogeneous of degree $d$, so the above action is well-defined. As $g^{-1}=g^T$, the polynomial $gf$ is obtained from $f$ by applying the substitutions
\[
 \xx \mapsto g_{11} \xx + g_{21} \yy + g_{31} \zz, \qquad
 \yy \mapsto g_{12} \xx + g_{22} \yy + g_{32} \zz, \qquad 
 \zz \mapsto g_{13} \xx + g_{23} \yy + g_{33} \zz,
\]
when $g = (g_{ij}) \in \G$ is an orthogonal $3\times 3$-matrix with entries $g_{ij}$. The use of the inverse $g^{-1}$ ($=g^T$) instead of $g$ in the above composition is only of notational importance, guaranteeing 
$g_1(g_2 f) = (g_1 g_2) f$ for all $g_1, g_2 \in \G$.

A counting argument shows that there are $N := \binom{d+2}{2}= \frac{1}{2}(d+2)(d+1)$ monomials $x^i y^j z^k$ of degree $d$, so formally, $\R[\xyz]_{ d}$ is an $N$-dimensional vector space over $\R$.
In order to stress this point of view, we shall  denote \[\VV_{ d} := \R[\xyz]_{ d}\] and, accordingly, we shall from now on typically denote elements of $\VV_{d}$ by letters like $v$ or $w$. %
The monomials $\{\xx^i\yy^j\zz^k \mid i+j+k= d\}$ form a basis for this vector space. But alternative bases will prove useful in this paper.

As $g(v+w) = gv+gw$ and $g(\lambda v) =\lambda(gv)$ holds for all $g \in \G$, $v,w \in \VV_{ d}$ and $\lambda \in \R$, the action of $\G$ on $\VV_{ d}$ is a \emph{linear} group action. If we fix a basis %
for $\VV_{ d}$, there is a polynomial map from $\G$ to the group of invertible 
$N\times N$ matrices that describes the action of $\G$ on $\VV_{ d}$ in this 
basis.

\begin{ex}[The action on quadratic forms]
  We show the matrix that provides the transformation by $g=\left(g_{ij}\right)_{1\leq i,j\leq 3} \in\G$ on $\VV_{2}$ in two different bases, starting with the monomial basis.

  If 
    \[f := a_{200} \, \xx^2+a_{110} \, \xx \yy+a_{020} \, \yy^2+a_{101} \, \xx \zz+a_{011} \, \yy \zz+a_{002} \, \zz^2\]
  and
    \[gf = b_{200} \, \xx^2+b_{110} \, \xx \yy+b_{020} \, \yy^2+b_{101} \, \xx \zz+b_{011} \, \yy \zz+b_{002} \, \zz^2,\]
  then the relation between the coefficients of $f$ and $gf$ is given by the matricial equality:
   \begingroup \small 
    \[ 
    \begin{pmatrix} b_{200}\\ b_{110} \\ b_{020}\\ b_{101} \\ b_{011}\\ b_{002} \end{pmatrix}
    =
    \begin{pmatrix}
      g_{11}^2 & g_{11}g_{12} & g_{12}^2 & g_{11}g_{13} & g_{12}g_{13} & g_{13}^2 \\
      2g_{21}g_{11} & g_{21}g_{12}+g_{11} g_{22} & 2g_{22}g_{12} & g_{21}g_{13}+g_{11}g_{23} & g_{22}g_{13}+g_{12}g_{23} & 2g_{23}g_{13} \\
      g_{21}^2 & g_{21}g_{22} & g_{22}^2 & g_{21}g_{23} & g_{22}g_{23} & g_{23}^2 \\
      2g_{31}g_{11} & g_{31}g_{12}+g_{11} g_{32} & 2g_{32}g_{12} & g_{31}g_{13}+g_{11}g_{33} & g_{32}g_{13}+g_{12}g_{33} & 2g_{33}g_{13} \\
      2g_{31}g_{21} & g_{31}g_{22}+g_{21} g_{32} & 2g_{32}g_{22} & g_{31}g_{23}+g_{21}g_{33} & g_{23}g_{32}+g_{22}g_{33} & 2g_{33}g_{23} \\
      g_{31}^2 & g_{31}g_{32} & g_{32}^2 & g_{31}g_{33} & g_{32}g_{33} & g_{33}^2
    \end{pmatrix}
      \begin{pmatrix} a_{200}\\ a_{110} \\ a_{020}\\ a_{101} \\ a_{011}\\ a_{002} \end{pmatrix} 
    \]
    \endgroup

    Let us denote by $R(g)$ the matrix in the above equality. 

    Alternatively the set of polynomials $\{\xx^2+\yy^2+\zz^2,\,\yy\zz,\,\zz\xx,\,\xx\yy,\,\yy^2-\zz^2,\,\zz^2-\xx^2\}$ forms a basis for $\VV_2$. Let $P$ be the matrix of change of basis, i.e.\ 
      \[
      \begin{pmatrix} \xx^2+\yy^2+\zz^2 & \yy \zz & \zz \xx & \xx \yy & \yy^2-\zz^2 & \zz^2-\xx^2 \end{pmatrix}
      =
      \begin{pmatrix} \xx^2 & \xx\yy & \yy^2 & \zz\xx & \yy\zz & \zz^2\end{pmatrix} \, P 
      \]
    
      If 
         \[f =  a_1\,(\xx^2+\yy^2+\zz^2)+a_2\,\yy\zz+a_3\,\zz\xx+a_4\,\xx\yy+a_5\,(\yy^2-\zz^2)+a_6\,(\zz^2-\xx^2)\]
      and 
        \[gf = b_1\,(\xx^2+\yy^2+\zz^2)+b_2\,\yy\zz+b_3\,\zz\xx+b_4\,\xx\yy+b_5\,(\yy^2-\zz^2)+b_6\,(\zz^2-\xx^2),\]
      then the vector of coefficients $b= \begin{pmatrix} b_1 & \ldots & b_6\end{pmatrix}^T$ is related to the vector of coefficients $a= \begin{pmatrix} a_1 & \ldots & a_6\end{pmatrix}^T$ by the matricial equality:
        \[b = \tilde{R}(g) \, a, \quad \hbox{ where } \tilde{R}(g)=P^{-1} R(g) P \]
    
 Using that $g \in \G$, one can compute that
  
    \begingroup \small
    \[
    \tilde{R}(g)= \begin{pmatrix} 1 & 0 & 0 & 0 & 0 & 0 \\ 
    0 & g_{23}g_{32}+g_{22}g_{33} & g_{31}g_{23}+g_{21}g_{33}   & g_{31}g_{22}+g_{21}g_{32} & 2g_{32}g_{22}-2g_{33}g_{23} & 2g_{33}g_{23}-2g_{31}g_{21} \\ 
    0 & g_{32}g_{13}+g_{12}g_{33} & g_{31}g_{13}+g_{11}g_{33} & g_{31}g_{12}+g_{11}g_{32} & 2g_{32}g_{12}-2g_{33}g_{13} & 2g_{33}g_{13}-2g_{31}g_{11} \\ 
    0 & g_{22}g_{13}+g_{12}g_{23} & g_{21}g_{13}+g_{11}g_{23} & g_{21}g_{12}+g_{11}g_{22} & 2g_{22}g_{12}-2g_{23}g_{13} & 2g_{23}g_{13}-2g_{21}g_{11} \\ 
    0 & g_{22}g_{23} & g_{21}g_{23} & g_{21}g_{22} & g_{22}^2-g_{23}^2 & g_{23}^2-g_{21}^2 \\ 
    0 & g_{32}g_{33}+g_{22}g_{23} & g_{31}g_{33}+g_{21}g_{23} & g_{31}g_{32}+g_{21}g_{22} & g_{22}^2+g_{32}^2-g_{23}^2-g_{33}^2 & g_{23}^2+g_{33}^2-g_{21}^2-g_{31}^2
    \end{pmatrix}.
    \]
    \endgroup
  Hence the linear spaces generated by $\xx^2+\yy^2+\zz^2$, on one hand, and by $\left\{\yy\zz,\,\zz\xx,\,\xx\yy,\,\yy^2-\zz^2,\,\zz^2-\xx^2\right\}$, on the other hand, are both invariant under the linear action of $\G$ on $\VV_2$.
\end{ex}

As we shall see in Section~\ref{ssec:HarmDecomp}, the quadratic polynomial $x^2+y^2+z^2 \in \VV_2$ also plays a special role in the action of $\G$ on $\VV_{d}$ for $d > 2$ due to the property that it is fixed by the action of $\G$.
  
  \subsection{Rational invariants and algorithmic problems} \label{ssec:AlgoFormulations}
    
We denote by $\R(\VV_{d})$ the set of rational functions $p \colon \VV_{d} \dashrightarrow \R$ on the vector space $\VV_{d}$. Explicitly, if we denote elements of $\VV_{d}$ as
  \[\sum_{i+j+k = d} a_{ijk} \, x^i y^j z^k,\]
then 
\begin{equation} \label{eq:functionFieldExplicit}
  \R(\VV_{d}) = \R(a_{ijk} \mid i,j,k \geq 0, i+j+k = d)
\end{equation}
is the field of rational expressions (i.e.\ quotients of polynomial expressions) in the variables $a_{ijk}$.

The explicit description of elements in $\R(\VV_{d})$ as expressions in the variables $a_{ijk}$ however reflects the choice of the monomial basis $\{x^i y^j z^k \mid i+j+k = d\}$ for $\VV_{d}$. Later on, we will work with a different basis for $\VV_{d}$ giving rise to an alternative presentations of elements in $\R(\VV_{d})$. We will therefore mostly avoid the description of $\R(\VV_{d})$ as in \eqref{eq:functionFieldExplicit} in the following discussions.

When working with rational functions $p = \frac{p_1}{p_0} \in \R(\VV_{d})$ 
(where $p_1, p_0 \colon \VV_{d} \to \R$ are polynomial functions and $p_0 
\not\equiv 0$), there is always an issue of division by zero: 
The function $p = \frac{p_1}{p_0}$ is only \emph{defined on a general point}, 
namely, on the set $\{v \in \VV_{d} \mid p_0(v) \neq 0\}$.
To keep this in mind, the rational function is denoted by a dashed arrow $p \colon \VV_{d} \dashrightarrow \R$. Similarly, the equalities we shall write have to be understood wherever they are defined, i.e.\ where all denominators involved are non-zero.

More generally, it is said that a statement $\mathcal P$ about a point $v \in V$
in a given $\R$-vector space $V$ (e.g. $V = \VV_{d}$) holds for a 
\defstyle{general point} if there exists a non-zero polynomial function $p_0 
\colon V \to \R$ such that $\mathcal P$ holds for all points $v \in V$ where 
$p_0(v) \neq 0$.
For a statement $\mathcal P$ about \emph{several} points $v_1, 
\ldots, v_k \in 
V$, we say that $\mathcal P$ holds for \defstyle{general points} $v_1, \ldots, 
v_k$ if there exists a non-zero polynomial function $p_0 \colon V \to \R$ such 
that $\mathcal P$ holds whenever $p_0(v_i) \neq 0$ for all $i \in \{1, \ldots, 
k\}$. 
Note that this is \emph{not} necessarily equivalent to saying that $\mathcal 
P$ holds for a general point $(v_1,\ldots,v_k)$ of the vector space $V^k$.

The set of \defstyle{rational invariants} for the action of $\G$ on $\VV_{d}$ is defined as
  \[\R(\VV_{d})^{\G} := \{p \in \R(\VV_{d}) \mid p(v) = p(gv) \ \forall v \in \VV_{d}, \; g \in \G\}.\] 
For any $d$ there exists a \emph{finite set} $\{p_1, \dots, p_m\} \subset \R(\VV_{d})^{\G}$ of rational invariants which generate $\R(\VV_{d})^{\G}$ as a field extension of $\R$. This means that any other rational invariant $q \in \R(\VV_{d})^{\G}$ can be written as a rational expression in terms of $p_1, \ldots, p_m$. We call such a finite set $\{p_1, \ldots, p_m\}$ a \defstyle{generating set of rational invariants}.

In contrast to the ring of polynomial invariants, the fact that the field of rational invariants is finitely generated is just an instance of the elementary algebraic fact: Any subfield of a finitely generated field is again finitely generated (see for example \cite[Theorem~24.9]{Isa09}). A lower bound for the cardinality of a generating set of rational invariants is given by the following result, implied by \cite[Corollary of Theorem~2.3]{PV94}:

\begin{thm} \label{thm:NumberInvariants}
  For $d \geq 2$, any generating set of rational invariants for the action of $\G$ on $\VV_{d}$ consists of at least $\dim \VV_{d} - \dim \G = \binom{d+2}{2} - 3$ elements.
\end{thm}

An important characterization of rational invariants is given by the following theorem, for which we refer to \cite[Lemma~2.1, Theorem~2.3]{PV94} or \cite{Rosenlicht56}. %

\begin{thm} \label{thm:rationalOrbitSeparation}
  Rational invariants $p_1, \dots, p_m \in \R(\VV_{d})^{\G}$ form a generating 
  set of $\R(\VV_{d})^{\G}$
   if and only if for general points $v, w \in \VV_{d}$ the following holds:
  \[w = gv \text{ for some } g \in \G \ \Leftrightarrow \ p_i(v) = p_i(w) \ \forall i \in \{1, \dots, m\}.\]
\end{thm}

In other words, a set of rational invariants generates the invariant field 
$\R(\VV_{d})^{\G}$ if and only if they separate $\G$-orbits in general position.

In this article, we focus on forms of even degree and explicitly construct a 
generating set of rational 
invariants for the action of $\G$ on $\VV_{2d}$,
whose cardinality attains the lower bound from Theorem 
\ref{thm:NumberInvariants}. 
The results appear in Corollary~\ref{cor:QuarticFullInvariants} for $\VV_4$ and in Corollary~\ref{cor:HigherFullInvariants} for $\VV_{2d}$, $d\geq 2$.

Theorem~\ref{thm:rationalOrbitSeparation} gives a crucial justification for approaching Question~\ref{qu:decideSimilarity} with rational invariants. It further addresses Question~\ref{qu:encodeEqClass} of how to encode even degree polynomials up to orthogonal transformations: We may encode a general point $v \in \VV_{2d}$ as the $m$-tuple $(p_1(v), \dots, p_m(v)) \in \R^m$. Then the $m$-tuples of $v \in \VV_{2d}$ and $w \in \VV_{2d}$ are equal if and only if $v$ and $w$ are  equivalent under the action of $\G$.

Associated with this approach are the following main Algorithmic Problems:

\begin{enumerate}[1.]
  \item \textbf{Characterization}: Determine  a set of generating rational invariants $p_1, \ldots, p_m \in \R(\VV_{2d})^{\G}$.
  \item \textbf{Evaluation}: Evaluate $p_1(v), \ldots, p_m(v)$ for a given (general) point $v \in \VV_{2d}$ in an efficient and robust way.
  \item \textbf{Rewriting}: Given a rational invariant $q \in \R(\VV_{2d})^{\G}$, express $q$ as a rational expression in terms of $p_1, \ldots, p_m$.
  \item \textbf{Reconstruction}: Which $m$-tuples $\mu = (\mu_1, \ldots, \mu_m) \in \R^m$ lie in the image of the map
  \[\pi \colon \VV_{2d} \dashrightarrow \R^m, \quad v \mapsto (p_1(v), \ldots, p_m(v))?\]
  If $\mu$ lies in the image of $\pi$, find a representative $v \in \VV_{2d}$ such that $\pi(v) = \mu$.

\end{enumerate}

General algorithms for computing a generating set of rational invariants based on Gröbner basis algorithms exist \cite{HK08}, \cite[Section 4.10]{DK15}, but the complexity increases drastically with the dimension of $\VV_{2d}$, which in turn grows quadratically in $d$. Already for $2d = 4$, these general methods are far from a feasible computation. Furthermore, these algorithms typically do not produce a minimal generating set. In this article we shall demonstrate the efficiency of a more structural approach for describing a generating set of rational invariants with minimal cardinality. How to address the algorithmic challenges 2--4 will become more apparent from our construction of the generating rational invariants, and we will examine them in detail in Section~\ref{sec:AlgorithmicSolutions}.

The field of rational invariants $\R(\VV_{2d})^{\G}$ is the quotient field of 
the ring of polynomial invariants $\R[\VV_{2d}]^{\G}$ \cite[Theorem 3.3]{PV94}; 
Any rational invariant can be written as the quotient of two polynomial 
invariants. Yet determining a generating set of polynomial invariants is a 
somewhat more arduous task. In \cite{AKO17}, a minimal set of generating 
polynomial invariants was identified for the action of $\G$ on $\H_{4}\oplus 2 
\H_{2} \oplus 2\H_{0}$, the space of the elasticity tensor. We can extract from 
this basis a set of $64$ polynomial invariants that generate 
$\R[\VV_{4}]^{\G}$. These invariants are computed thanks to Gordan's algorithm 
\cite{GY03,Oli16}, 
after the problem is reduced to the action of $\SL(2,\C)$ on binary forms through Cartan's map. One has to observe though that polynomial invariants separate the real orbits of $\G$ \cite[Proposition 2.3]{Sch01},
while rational invariants will only separate \emph{general} orbits 
(Theorem~\ref{thm:rationalOrbitSeparation}).

  \subsection{Harmonic Decomposition} \label{ssec:HarmDecomp}
 
In the study of the action of $\G$ on $\VV_{d}$, the \emph{Harmonic Decomposition} plays a central role. We start out by collecting some basic facts about the apolar inner product and harmonic polynomials. %

\begin{definition}
  The \textbf{apolar inner product} $\langle \cdot , \cdot \rangle_d \colon 
  \VV_{d} \times \VV_{d} \to \R$ is defined as follows: If $v = \sum_{ijk} 
  a_{ijk} x^i y^j z^k \in \VV_{d}$ and $w = \sum_{ijk} b_{ijk} x^i y^j z^k \in 
  \VV_{d}$, we define
    \[\langle v, w \rangle_d := \sum_{i,j,k} i! j! k! \, a_{ijk} b_{ijk},\]
  where the sums range over $\{i,j,k \geq 0 \mid i+j+k = d\}$.
\end{definition}

The apolar inner product arises naturally as follows: Identifying $\VV_d = \R[\xyz]_d$ with the space of symmetric tensors $\Sym^d(\R^3)^* = \Sym^d(\VV_1)$, it is (up to a rescaling constant) the inner product inherited from $\VV_1^{\otimes d}$ via the embedding $\Sym^d(\VV_1) \hookrightarrow \VV_1^{\otimes d}$. Here, the inner product on $\VV_1^{\otimes d}$ is induced by the standard inner product on $\VV_1 = (\R^3)^*$.

Since the group $\G$ preserves the standard inner product on $\VV_1$, this viewpoint leads to the following fact.
\begin{prop} \label{prop:InnerProduct}
  The apolar inner product is preserved by the group action of $\G$, i.e.\ if $v, w \in \VV_{d}$  and $g \in \G$, then $\langle gv, gw \rangle_d = \langle v, w \rangle_d$.
\end{prop}

Another intrinsic formulation of the apolar product \cite[Section 2.1]{AH12} is given as follows: For a polynomial function $f(\xyz)\in \R[\xyz]_{d}$, let $f(\partial)$ be the differential operator obtained from $f$ by replacing $\xyz$ respectively by $\frac{\partial\hbox{ }}{\partial \xx}$, $\frac{\partial\hbox{ }}{\partial \yy}$, $\frac{\partial\hbox{ }}{\partial \zz}$. One then checks that for $f_1, f_2 \in \R[\xyz]_d$ the apolar inner product is given by 
  \[\langle f_1, f_2 \rangle_{d} = f_1(\partial)(f_2) = f_2(\partial)(f_1).\]
With this viewpoint, Proposition~\ref{prop:InnerProduct} follows by observing that $\langle gf_1, f_2 \rangle_d = \langle f_1, g^T f_2 \rangle_d$ holds for all group elements $g$.

From now on, we denote 
   \[\VQ := x^2+y^2+z^2 \in \VV_2.\]
This $\VQ \in \VV_2$ plays a special role as it is fixed by the action of $\G$: $g\VQ=\VQ$ for all $g \in \G$.

\begin{definition} \label{def:Harmonics}
  For any $d \geq 2$, we consider the inclusion of vector spaces
    \(\VV_{d-2} \hookrightarrow \VV_{d}, \quad v \mapsto \VQ\cdot v\)
  and its image $\VQ\VV_{d-2} \subset \VV_{d}$, which is given by those polynomials in $\VV_{d} = \R[\xyz]_d$ that are divisible by $\VQ$. 
	
	We define the subspace $\H_{d} \subset \VV_{d}$ of \textbf{harmonic polynomials} of degree $d$ to be the orthogonal complement of $\VQ\VV_{d-2} \subset \VV_{d}$ with respect to the apolar inner product on $\VV_{d}$.
\end{definition}

An immediate consequence of the invariance of $\VQ$ and Proposition~\ref{prop:InnerProduct} is the following observation.
\begin{prop} \label{prop:InvariantSubspaces}
  Let $g \in \G$ and $v \in \VV_{d}$. Then the following holds:
  \begin{enumerate}[(i)]
    \item If $v \in \VQ\VV_{d-2} \subset \VV_{d}$, then also $gv \in \VQ\VV_{d-2}$.
    \item If $v \in \H_{d}$, then also $gv \in \H_{d}$.
  \end{enumerate}
\end{prop}

Harmonic functions are typically introduced as the functions $f$ such that $\Delta(f)=0$, 
where $\Delta$ is the Laplacian operator, i.e.\ 
$\Delta=\frac{\partial^2\hbox{ }}{\partial\xx^2}+\frac{\partial^2\hbox{ }}{\partial\yy^2}+\frac{\partial^2\hbox{ }}{\partial\zz^2}$,
\cite{ABR01}. 
We can see that this is equivalent to Definition~\ref{def:Harmonics} as follows: Understanding the apolar inner product via differential operators as described above, we have 
  \[\langle f_1 f_3, f_2\rangle_{d} =  \langle f_1 , f_3(\partial)(f_2) \rangle_{k} \qquad \text{for all } f_1 \in \R[\xyz]_k,  f_2 \in \R[\xyz]_{d}, f_3 \in \R[\xyz]_{d-k}.\]

In particular, $\langle f_1, \VQ f_2\rangle_{d} = \langle \Delta(f_1) , f_2 \rangle_{d-2}$ holds for all $f_1 \in \R[\xyz]_d, f_2 \in \R[\xyz]_{d-2}$, so  that $f_1 \in \H_d$ if and only if $\Delta(f_1)=0$.

So far, all formulations have been made for arbitrary degree $d$. However, we are ultimately interested in the case of even degree only. From now on and for the remainder of the article we  will therefore only consider the case of degree $2d$.

By Definition~\ref{def:Harmonics}, there is an orthogonal decomposition $\VV_{2d} = \H_{2d} \oplus \VQ\VV_{2d-2}$. For $d \geq 2$ we can also decompose $\VV_{2d-2}$ in this manner, so we may iterate this decomposition which leads to the following observation, \cite[Theorem~5.7]{ABR01}.

\begin{thm}[Harmonic Decomposition] \label{thm:HarmonicDecomp}
  For $d \geq 2$ there is a decomposition
    \[\VV_{2d} = \H_{2d} \oplus \VQ\H_{2d-2} \oplus \VQ^2\H_{2d-4}\dots \oplus \VQ^{d-2}\H_4 \oplus \VQ^{d-1}\VV_2,\]
  i.e.\ each $v \in \VV_{2d}$ can uniquely be written as a sum 
    \[v = h_{2d} + \VQ h_{2d-2} + \VQ^2h_{2d-4} + \ldots + \VQ^{d-2}h_4 + \VQ^{d-1}v'\]
  where $h_{2k} \in \H_{2k}$ and $v' \in \VV_2$.
\end{thm}

We mention at this point that it would be possible to refine Theorem~\ref{thm:HarmonicDecomp} by further decomposing $\VV_2 = \H_2 \oplus \R\VQ$, but this is not beneficial for our purpose.

Different bases for the vector spaces $\H_{2d} \subset \VV_{2d}$ of harmonic polynomials are used in applications. A frequent choice in practice is the basis of \emph{spherical harmonics} \cite{AH12} which are usually given as functions in spherical coordinates.  In Section~\ref{sec:B3Harmonics}, we shall construct another basis for $\H_{2d}$ that exhibits certain symmetries with respect to the group of signed permutations. %

\section{The Slice Method} \label{sec:SliceMethod}
  
Our aim is to determine a generating set of rational invariants for the linear action of the orthogonal group $\G$ on the vector spaces $\VV_{2d}$ of even degree ternary forms. The group $\G$ is infinite and of dimension~3 as an algebraic group. We reduce the problem to the simpler question of determining rational invariants for the linear action of a \emph{finite} group $\B$ (contained in $\G$ as a subgroup) on a subspace $\Lambda_{2d}$ of $\VV_{2d}$.

\subsection{The Slice Lemma} \label{ssec:SliceLemma}
We introduce the general technique for the reduction mentioned above, 
called the \emph{slice method}. For the formulation, we abstract from our 
specific setting: We consider a linear action of a real algebraic group $\Ggen$ 
on a finite-dimensional $\R$-vector space $\VV$, denoted
  \[\Ggen \times \VV \to \VV, \quad (g,v) \mapsto gv.\]
We denote by $\R(\VV)$ the field of 
rational functions and by $\R(\VV)^{\Ggen}$ the finitely generated subfield of 
rational invariants, as introduced already for $\VV = \VV_{2d}$ and $\Ggen=\G$. 
Let $\VV_\C$ and $\Ggen_\C$ denote the complexifications of $\VV$ and $\Ggen$, 
respectively (and analougously for other real vector spaces and real algebraic 
groups). Note that there is an induced action of $\Ggen_\C$ on $\VV_\C$.
  
In our particular case, we have $\Ggen = \G$ and $\VV = \VV_{2d}$, and the 
action of $G$ on $\VV$ is defined in Section~\ref{ssec:GroupAction}. Here, 
$\Ggen_\C = \G(\C)$ is the group of complex orthogonal matrices and $\VV_\C = 
\VV_{2d} \otimes_\R \C = \C[x,y,z]_{2d}$.

The main technique for the announced reduction is known as the \emph{Slice Method} \cite[Section~3.1]{CTS07}, \cite{Popov94s}. It is based on the following definition.

\begin{definition} \label{def:Slice}
  Consider a linear group action of an algebraic group $\Ggen$ on a finite-dimensional $\R$-vector space $\VV$. A subspace $\Lambda \subset \VV$ is called a \defstyle{slice} for the group action, and the subgroup 
	\[\Bgen := \{g \in \Ggen \mid gs \in \Lambda \ \forall s \in \Lambda\} \subset \Ggen\]
	is called its \defstyle{stabilizer}, if the following two properties hold:
	\begin{enumerate}[(i)]
		\item For a general point $v \in \VV$ there exists $g \in \Ggen$ such that $gv \in \Lambda$.
		\item For a general point $s \in \Lambda_\C$ the following holds: If $g \in 
		\Ggen_\C$ is such that $gs \in \Lambda_\C$, then $g \in \Bgen_\C$.
	\end{enumerate}
\end{definition}

The Slice Lemma then states that rational invariants of the action of $\Ggen$ 
on $\VV$ are in one-to-one correspondence with rational invariants of the 
smaller group $\Bgen \subset \Ggen$ on the slice $\Lambda \subset \VV$:

\begin{thm}[Slice Lemma] \label{thm:SliceLemma}
  Let $\Lambda$ be a slice of a linear action of an algebraic group $\Ggen$ on a finite-dimensional $\R$-vector space $\VV$, and let $\Bgen$ be its stabilizer. Then there is a field isomorphism over $\R$ between rational invariants
  	\[\varrho \colon \R(\VV)^{\Ggen} \xrightarrow{\cong} \R(\Lambda)^{\Bgen}, \ p \mapsto \restr{p}{\Lambda},\]
	which restricts a rational invariant $p \colon \VV \dashrightarrow \R$ to $\restr{p}{\Lambda} \colon \Lambda \dashrightarrow \R$.
\end{thm}

This observation goes back to \cite{Ses61}, and we refer to 
\cite[Theorem~3.1]{CTS07} for a proof. The above version of the Slice Lemma is 
weaker than the formulation in \cite[Theorem~3.1]{CTS07}, but it is sufficient 
in our case.

Explicitly, the inverse to $\varrho 
\colon \R(\VV)^{\Ggen} \xrightarrow{\cong} \R(\Lambda)^{\Bgen}$
in Theorem~\ref{thm:SliceLemma} is given by
\begin{align*}
	\varrho^{-1} \colon \R(\Lambda)^{\Bgen} &\to \R(\VV)^{\Ggen}, \\
	q &\mapsto \big(\VV \dashrightarrow \R, \quad v \mapsto q(gv), \text{ where $g \in \Ggen$ is such that $gv \in \Lambda$}\big).
\end{align*}
The assumption that property~(ii) in Definition~\ref{def:Slice} holds over the 
complex numbers is needed to show that the map on the right is indeed a 
rational function.

We will apply Theorem~\ref{thm:SliceLemma} for $\Ggen = \G$, $\VV = \VV_{2d}$ and a suitable choice for the slice $\Lambda$. The consequence of Theorem~\ref{thm:SliceLemma} for the construction of a generating set of rational invariants is the following.

\begin{cor} \label{cor:SliceLemmaGen}
  Let $\Lambda$ be a slice of a linear action of a real algebraic group $\Ggen$ 
  on a finite-dimensional $\R$-vector space $\VV$, and let $\Bgen$ be its 
  stabilizer. If $\setP[] = \{p_1, \dots, p_m\}$ is a generating set of 
  rational invariants for the action of $\Bgen$ on $\Lambda$, then $\setPT[] := 
  \{\varrho^{-1}(p_1), \dots, \varrho^{-1}(p_m)\}$ is a generating set of 
  rational invariants for the action of $\Ggen$ on $\VV$ (where $\varrho$ is 
  given as above).
\end{cor}

\begin{proof}
 Let $p_0 \in \R(\VV)^{\Ggen}$. By assumption, $\varrho(p_0) \in \R(\Lambda)^{\Bgen}$ can be written as a rational expression in the generators $p_1, \dots, p_m$. Since $\varrho$ is a field isomorphism, $p_0 = \varrho^{-1}(\varrho(p_0))$ is the same rational expression in $\varrho^{-1}(p_1), \dots, \varrho^{-1}(p_m)$.
\end{proof}

A corresponding statement for polynomial invariants requires  much stronger hypotheses on the slice. In particular, even if the generating set for $\R(\Lambda)^\Bgen$ consists of \emph{polynomial expressions} $p_1, \ldots, p_m \in \R[\Lambda]^{\Bgen}$, the construction described above typically introduces denominators, so that $\varrho^{-1}(p_1), \dots, \varrho^{-1}(p_m) \in \R(\VV)^{\Ggen}$ become \emph{rational} expressions.

  \subsection{A slice for \texorpdfstring{$\VV_{2d}$}{V\_\{2d\}}}
  
We now describe a slice $\Lambda_{2d} \subset \VV_{2d}$ for the action of $\G$ on $\VV_{2d}$ for any $d \geq 1$. %

We recall from Section~\ref{ssec:motivation} the description of the action of $\G$ on $\VV_{2}$: Elements of $\VV_2$ are ternary quadratic forms and they can be identified with symmetric $3 \times 3$-matrices as in \eqref{eq:GramianMatrix}. If the associated symmetric matrix of $v \in \VV_2$ is $A$, then for any $g \in \G \subset \R^{3 \times 3}$, the associated symmetric matrix of $gv \in \VV_2$ is the matrix product $g A g^T$.
\begin{definition}
  Let $\Lambda_2 \subset \VV_2$ denote the subspace of quadratic forms whose associated symmetric matrix is diagonal. Explicitly,
    \[\Lambda_2 = \{\lambda_1 x^2 + \lambda_2 y^2 + \lambda_3 z^2 \in \VV_2 \mid \lambda_1, \lambda_2, \lambda_3 \in \R\}.\]
  For $d \geq 2$ we consider the Harmonic Decomposition of $\VV_{2d}$ from Theorem~\ref{thm:HarmonicDecomp} and define $\Lambda_{2d} \subset \VV_{2d}$ to be the subspace
    \[\Lambda_{2d} := \H_{2d} \oplus \VQ \H_{2d-2} \oplus \dots \oplus \VQ^{d-2} \H_4 \oplus \VQ^{d-1}\Lambda_2.\]
\end{definition}

In other words, elements of the subspace $\Lambda_{2d}$ are those $v \in \VV_{2d}$ that can be written as
  \[v = h_{2d} + \VQ h_{2d-2} + \VQ^2 h_{2d-4} + \ldots + \VQ^{d-2} h_4 + \VQ^{d-1}v'\]
with $h_{2k} \in \H_{2k}$ and $v' \in \Lambda_2$ a quadratic form whose associated symmetric matrix is diagonal. The main observation is now the following:

\begin{prop} \label{prop:MainSlice}
  Let $d \geq 1$. The subspace $\Lambda_{2d} \subset \VV_{2d}$ is a slice for the action of $\G$ on $\VV_{2d}$ and its stabilizer is the group $\B \subset \G$ of signed permutation matrices. In particular, there is a one-to-one correspondence between rational invariants
    \[\varrho \colon \R(\VV_{2d})^{\G} \xrightarrow{\cong} \R(\Lambda_{2d})^{\B}\]
  given by the restriction of rational functions.
\end{prop}

We recall that a \emph{signed permutation matrix} is a matrix for which each row and each column contain only one non-zero entry and this entry is either $1$ or $-1$. Below we will remark on further structural descriptions of the group $\B$.

\begin{proof}[Proof of Proposition~\ref{prop:MainSlice}]
  The second statement is a consequence of the first statement by Theorem~\ref{thm:SliceLemma}.
  
  Let
    \[v = h_{2d} + \VQ h_{2d-2} + \VQ^2 h_{2d-4} + \ldots + \VQ^{d-2} h_4 + \VQ^{d-1}v'\]
  be the Harmonic Decomposition of an element $v \in V_{2d}$ and let $A \in \R^{3 \times 3}$ be the symmetric matrix associated to the quadratic form $v' \in \VV_2$. By the Spectral Theorem for symmetric matrices, there exists an orthogonal matrix $g \in \G \subset \R^{3 \times 3}$ such that $gAg^T$ is a diagonal matrix (whose diagonal entries are the eigenvalues of $A$). Since $gAg^T$ is the associated symmetric matrix of $gv' \in \VV_2$, this means that
    \[gv = (gh_{2d}) + \VQ (gh_{2d-2}) + \VQ^2 (gh_{2d-4}) + \ldots + \VQ^{d-2} (gh_4) + \VQ^{d-1}(gv')\]
  is contained in $\Lambda_{2d}$. This verifies property~(i) of Definition~\ref{def:Slice} for $\Lambda_{2d} \subset \VV_{2d}$.
  
  A complex orthogonal matrix $g \in \G(\C)$ lies in the 
  stabilizer 
  $\Bgen_\C := \{g \in \G(\C) \mid gv \in (\Lambda_{2d})_\C \ \forall v \in 
  (\Lambda_{2d})_\C\}$ 
  if and only if the matrix product $g \diag(\lambda_1, \lambda_2, \lambda_3) 
  g^T$ is again a diagonal matrix for all values of $\lambda_1, \lambda_2, 
  \lambda_3 \in \C$. This means that the rows of the matrix $g$ are orthonormal 
  eigenvectors of $\diag(\lambda_1, \lambda_2, \lambda_3)$ for all values of 
  $\lambda_1, \lambda_2, \lambda_3 \in \C$ (where orthogonality and length of 
  vectors refers to the 
  standard 
  \emph{bilinear} form $\C^3 \times \C^3 \to \C$, $(v,w) \mapsto v^T w$).  In 
  particular, this holds when 
  $\lambda_1, \lambda_2, \lambda_3 \in \C$ are \emph{distinct}, in which case 
  the only unit eigenvectors are $\pm e_i$ (where $e_1, e_2, e_3$ is the 
  standard basis of $\C^3$). Hence, $g$ is a signed permutation matrix. This 
  establishes the description of the stabilizer $\Bgen = \Bgen_\C = \B$.
  
  Similarly, to verify property~(ii) of Definition~\ref{def:Slice}, we only 
  need to show: If $g \in \G(\C)$ is such that for a \emph{general point} 
  $(\lambda_1, \lambda_2, \lambda_3) \in \C^3$, the matrix product $g 
  \diag(\lambda_1, \lambda_2, \lambda_3) g^T$ is a diagonal matrix, then $g \in 
  \B$. Indeed, we have just seen that this property holds whenever $\lambda_1, 
  \lambda_2, \lambda_3 \in \C$ are distinct values, i.e.\ wherever the 
  polynomial function
	\[\C^3 \to \C, \qquad (\lambda_1,\lambda_2,\lambda_3) \mapsto 
	(\lambda_1-\lambda_2)(\lambda_2-\lambda_3)(\lambda_3-\lambda_1)\]
	does not vanish.
\end{proof}

We note that in the proof we verified property~(i) of Definition~\ref{def:Slice} for \emph{all} $v \in \VV_{2d}$, while for property~(ii) we had to make use of the notion of a \emph{general point}.

Above, we introduced the group $\B \subset \G$ of signed permutation matrices. 
This is a \emph{finite} group with $2^3 \cdot 3! = 48$ elements and is known as 
the \emph{octahedral group} or as the wreath product $\GS[2] \wr \GS[3]$, 
i.e.\ as the semidirect product of the symmetric group $\GS$ with the abelian 
group $\GA$.

Note that $-\id \in \B$ acts trivially on $\VV_{2d}$, so we in fact have an 
action of the quotient $\B/\langle -\id \rangle$ on $\Lambda_{2d}$. The group 
$\B/\langle -\id \rangle$ is isomorphic to the symmetric group 
$\mathfrak S_4$. However, the group action on $\Lambda_{2d}$ is more 
straightforwardly 
formulated in terms of $\B$ (and more naturally generalizes to the cases of odd 
degree or more than three variables). For working explicitly with $\B$, we 
introduce the following notation.

\begin{notation}
For $\sigma \in \GS$ we denote by $g_\sigma \in \G$ the permutation matrix such that
  \[g_\sigma \begin{pmatrix} \alpha_1 \\ \alpha_2 \\ \alpha_3 \end{pmatrix} = \begin{pmatrix} \alpha_{\sigma^{-1}(1)} \\ \alpha_{\sigma^{-1}(2)} \\ \alpha_{\sigma^{-1}(3)} \end{pmatrix}\]
For $\tau=(\tau_1, \tau_2,\tau_3)\in \GA$ we write $g_\tau := \diag(\tau_1, \tau_2,\tau_3) \in \G$ and we will call these matrices \defstyle{sign-change matrices}.
\end{notation}

Then $\B$ is the smallest subgroup of $\G$ containing all permutation matrices and all sign-change matrices: Each signed permutation matrix $g \in \B$ can  be uniquely written as $g = g_\tau g_\sigma$ with $\tau \in \GA$, $\sigma \in \GS$. Indeed, $\tau_i$ must be $1$ or $-1$ corresponding to the sign of the unique non-zero entry in the $i$-th row of the matrix $g$. For this $\tau$, the matrix $g_\sigma := g_\tau^{-1} g$ is a permutation matrix. Analogously, we can also write each $g \in \B$ uniquely as $g = g_\sigma g_{\tau'}$ with $\tau' \in \GA$, $\sigma \in \GS$. However, $\tau' \neq \tau$ in general.

  \subsection{Illustration on quadratic forms} \label{ssec:QuadrInv}

We illustrate the Slice Method on the well-known case of quadratic forms: The slice $\Lambda_2 \subset \VV_2$ is given as
  \[\Lambda_2 = \{\lambda_1 x^2 + \lambda_2 y^2 + \lambda_3 z^2 \mid \lambda_1, 
  \lambda_2, \lambda_3 \in \R\}.\]
By Theorem~\ref{thm:SliceLemma}, the field of rational $\G$-invariants on $\VV_2$ is isomorphic to the field of ratio{}nal $\B$-invariants on $\Lambda_2$. 

We observe that sign-change matrices in $\B$ act trivially on $\Lambda_2$, so
  \[\R(\Lambda_2)^{\B} = \R(\Lambda_2)^{\GS} = \R(\lambda_1, \lambda_2, \lambda_3)^{\GS},\]
where the symmetric group $\GS \subset \B$ acts by permuting $\lambda_1, \lambda_2, \lambda_3$ (which we view as coordinates on $\Lambda_2$).
By the Fundamental Theorem of symmetric polynomials, we can thus choose the set of elementary symmetric polynomials
  \[\setP[] = \{\lambda_1 + \lambda_2 + \lambda_3,\; \lambda_1 \lambda_2 + \lambda_2 \lambda_3 + \lambda_3 \lambda_1,\; \lambda_1 \lambda_2 \lambda_3\}\]
  as a generating set of rational invariants for $\R(\Lambda_2)^{\B}$. 
By Corollary~\ref{cor:SliceLemmaGen}, there exist unique rational $\G$-invariants $\setPT[] \subset \R(\VV_2)^{\G}$ on $\VV_2$ restricting to these invariants on the subspace $\Lambda_2$, and they form a generating set for $\R(\VV_2)^{\G}$. By uniqueness, the generating set $\setPT[]$ consists up to scalars precisely of the polynomial invariants $e_1, e_2, e_3 \in \R[\VV_2]^{\G}$ described in \eqref{eq:quadrInvariants}, arising as the coefficients of the characteristic polynomial of the symmetric matrix $A$ corresponding to a quadratic form.
Alternatively one could consider 
\[\setP[] = \{\lambda_1 + \lambda_2 + \lambda_3,\; \lambda_1^2+\lambda_2^2+\lambda_3^2,\; \lambda_1^3+\lambda_2^3+\lambda_3^3\},\] the set of Newton sums, 
and 
\[\setPT[]=\{ \hbox{Tr}(A),\;\hbox{Tr}(A^2),\;\hbox{Tr}(A^3) \}.\] 
We shall actually opt for this choice when extending the result to ternary quartics in 
Theorem~\ref{thm:TernaryQuarticInvariantsNonMinimal}.

In this case, one can in fact show that $\setPT[]$ even generates the ring of \emph{polynomial} invariants $\R[\VV_2]^{\G}$. This will however no longer be true for our construction of generating rational invariants in higher degree. Furthermore, the construction of $\B$-invariants on the slice will be more involved, especially because the sign-change matrices in $\B$ no longer act trivially on $\Lambda_{2d}$ for $d \geq 2$. 

\section{Invariants of ternary quartics} \label{sec:Quartics}
   
In this section we implicitly describe a generating set of rational invariants of minimal cardinality for ternary quartics under the action of $\GO$, i.e.\ for the case $2d = 4$. %
Following the approach of Section~\ref{ssec:SliceLemma} this set of rational invariants is uniquely determined 
by a set of rational invariants for the action of $\GB$ on the slice $\Lambda_4$. 
The first step is to construct an appropriate basis of $\Lambda_4$  that is equivariant.
In this basis, a minimal generating set of rational $\GB$-invariants takes a 
particularly compact form, 
and can be chosen to consist of polynomial invariants. 

We provide an additional (near minimal) generating set of $\R(\Lambda_4)^{\GB}$ that extends the invariants for quadratic forms. For both choices of generating invariants, we make explicit how to write any other invariant in terms of these, following \cite{HK08}.

The construction of this section serves as a model for invariants of ternary forms in higher degree $2d > 4$ 
which are treated in Section~\ref{sec:B3Harmonics}.

   \subsection{A \texorpdfstring{$\B$}{B\_3}-equivariant basis for harmonic quartics} \label{ssec:HarmonicBasisQuartics}
   
In order to construct $\B$-invariants on the vector space $\Lambda_4 = \H_4 \oplus \VQ \Lambda_2$, we introduce a basis of $\H_4$ that exhibits certain symmetries with respect to $\GB$. In Section~\ref{sec:B3Harmonics} we  show how to construct bases of $\H_{2d}$ with analogous symmetry properties for arbitrary $d \geq 2$.

\begin{prop} \label{prop:HarmonicBasisQuartic}
  The following nine ternary quartics form a basis for the $\R$-vector space $\H_4$:
  \begin{gather*}
  \Vr{1} := y^4 - 6y^2z^2 + z^4, \quad \Vs{1} := y^3z - yz^3, \quad \Vt{1} := 6x^2yz-y^3z-yz^3, \\
  \Vr{2} := z^4 - 6z^2x^2 + x^4, \quad \Vs{2} := z^3x - zx^3, \quad \Vt{2} := 6y^2zx-z^3x-zx^3, \\
  \Vr{3} := x^4 - 6x^2y^2 + y^4, \quad \Vs{3} := x^3y - xy^3, \quad \Vt{3} := 6z^2xy-x^3y-xy^3.
  \end{gather*}
  The group $\GB \subset \GO$ acts on $\H_4$ with respect to this basis as follows: For $\sigma \in \GS$ and $\tau \in \GA$, the corresponding permutation and sign-change matrices $g_\sigma$ and $g_\tau$ act by
	\begin{align*}
	  g_{\sigma} \Vr{i} &= \Vr{\sigma(i)}, & \qquad g_\sigma \Vs{i} &= \sgn(\sigma) \Vs{\sigma(i)}, & \qquad g_{\sigma}\Vt{i} &= \Vt{\sigma(i)}, \\
		g_\tau \Vr{i} &= \Vr{i}, & \qquad g_\tau \Vs{i} &= \frac{\tau_1 \tau_2 \tau_3}{\tau_i} \Vs{i}, & \qquad g_\tau\Vt{i} &= \frac{\tau_1 \tau_2 \tau_3}{\tau_i} \Vt{i}.
	\end{align*}
\end{prop}

Figure~\ref{fig:quartics} illustrates this basis of $\H_4$ as explained in Section~\ref{ssec:motivation}, similar to Figure~\ref{fig:rotateSurfaces}. Here, different colors within one picture correspond to different signs of the harmonic polynomial at corresponding points.

\begin{figure}
\begin{tabular}{m{0.02\textwidth}m{0.25\textwidth}m{0.02\textwidth}m{0.25\textwidth}m{0.02\textwidth}m{0.25\textwidth}}
$\Vr{1}$ & \includegraphics[width=0.2\textwidth]{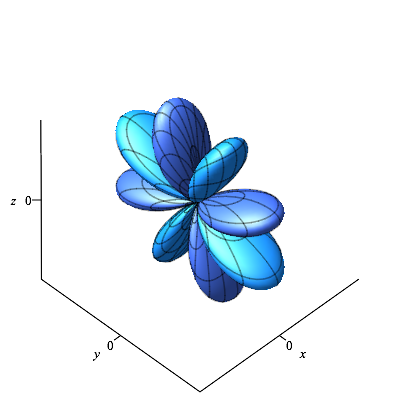} &
$\Vs{1}$ & \includegraphics[width=0.2\textwidth]{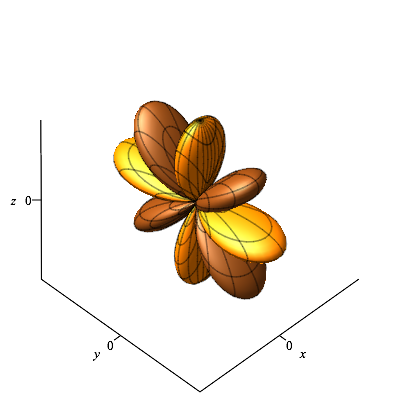} &
$\Vt{1}$ & \includegraphics[width=0.2\textwidth]{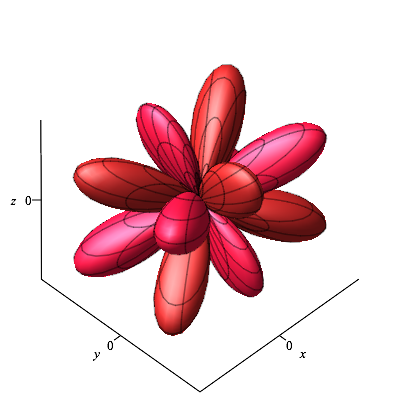} \\
$\Vr{2}$ & \includegraphics[width=0.2\textwidth]{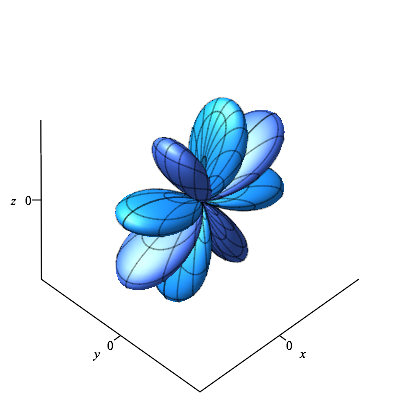} &
$\Vs{2}$ & \includegraphics[width=0.2\textwidth]{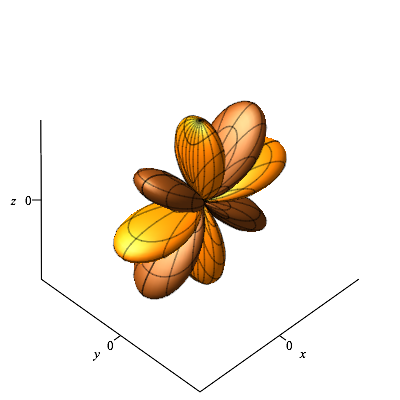} &
$\Vt{2}$ & \includegraphics[width=0.2\textwidth]{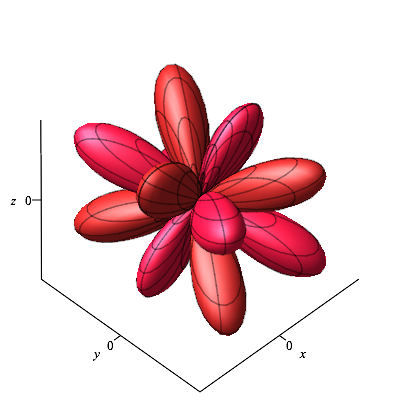} \\
$\Vr{3}$ & \includegraphics[width=0.2\textwidth]{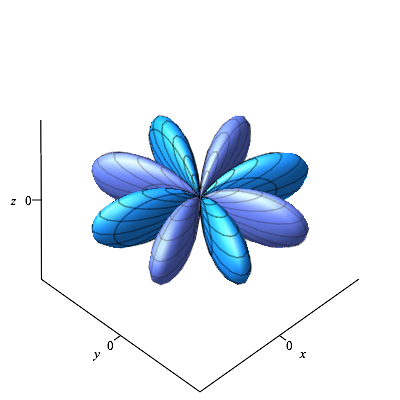} &
$\Vs{3}$ & \includegraphics[width=0.2\textwidth]{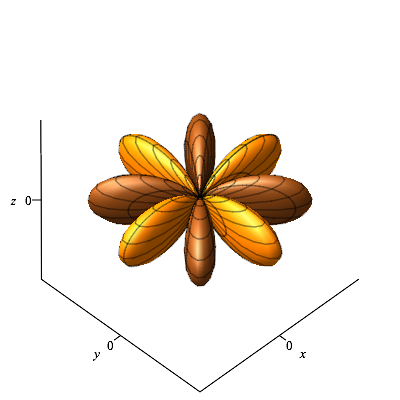} &
$\Vt{3}$ & \includegraphics[width=0.2\textwidth]{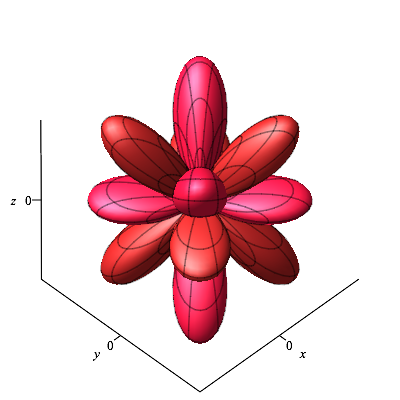}
\end{tabular}
\caption{Harmonic quartics} \label{fig:quartics}
\end{figure}

\begin{proof} 
  One can first check that the polynomials $\Vr{i}, \Vs{i}, \Vt{i}$ ($1 \leq i \leq 3$) are in the kernel of the Laplacian operator $\Delta=\frac{\partial^2\hbox{ }}{\partial \xx^2}+\frac{\partial^2\hbox{ }}{\partial \yy^2}+\frac{\partial^2\hbox{ }}{\partial \zz^2}$. As $\H_4$ is a 9-dimensional $\R$-vector space, it suffices to check that all $\Vr{i}, \Vs{i}, \Vt{i}$ ($1 \leq i \leq 3$) are linearly independent. This easily follows from examining the monomials occurring in their expressions.

  Note that a permutation matrix in $\GB$ acts by applying the corresponding permutation to the variables $x$, $y$, $z$, and a sign-change matrix acts by replacing some of the variables $x$, $y$, $z$ by their negatives. The claim about the action of $\GB$ is then read off the formulas for the basis elements.
\end{proof}

A polynomial $v$ in $\Lambda_4 = \VQ \Lambda_2 \oplus \H_4$ can thus be written as
\begin{equation} \label{coordsPoly}
  v = (\xx^2+\yy^2+\zz^2) \sum_{i=1}^3 \Cl[i] x_i^2 +\sum_{i=1}^3 \left(\Cr[i]  \Vr{i}+\Cs[i] \Vs{i}+\Ct[i] \Vt{i} \right).
\end{equation}

We introduce the vector of coefficients $\Cl=\begin{pmatrix} \Cl[1] & \Cl[2] &\Cl[3]\end{pmatrix}^T$ and similarly $\Cr$, $\Cs$ and $\Ct$. Then Proposition~\ref{prop:HarmonicBasisQuartic} shows: For $\sigma\in\GS$, if $g_\sigma v$ has coefficients $\left(\tilde\Cl,\tilde\Cr,\tilde\Cs,\tilde\Ct\right)$,
then
  \[\tilde\Cl = g_\sigma \Cl, \quad  \tilde\Cr = g_\sigma \Cr, \quad \tilde \Cs = \sgn(\sigma) g_\sigma \Cs, \quad \tilde \Ct = g_\sigma \Ct.\]
For $\tau\in\GA$,  if $g_\tau v$ is determined by $\left(\bar\Cl,\bar\Cr,\bar\Cs,\bar\Ct\right)$,
then 
  \[\bar\Cl = \Cl, \quad \bar \Cr = \Cr, \quad \bar \Cs = \det(g_\tau) g_\tau \Cs, \quad \bar \Ct = \det(g_\tau) g_\tau \Ct.\]

   \subsection{\texorpdfstring{$\B$}{B\_3}-invariants on the slice} \label{ssec:QuarticInvariants}

The rational invariants of $\GB$ on $\Lambda_4$ can be obtained computationally 
by applying the general algorithm for generating sets of rational invariants 
presented in \cite{HK08}. The results obtained with this approach suggest 
nice structures when the basis of Proposition~\ref{prop:HarmonicBasisQuartic} 
is used. We accordingly present generating sets of invariants as the results of 
the composition of equivariant maps. 
We first present a minimal generating set, which consists of $12$ algebraically independent polynomials, 
and then  a generating set that consist of $13$ rational invariants, including the invariants for $\Lambda_2$. 

\begin{lemma} \label{lem:EqMap}
	The maps $\Eqm[1], \Eqm[2] \colon \Lambda_2 \to \R^{3\times 3}$ and $\Eqm[3], \Eqm[4] \colon \Lambda_4 \to \R^3$ whose values at $v$ as in (\ref{coordsPoly}) are respectively given by 
	\[\Eqm[1](v) = \begin{pmatrix} 
								 1 & \Cl[1] & \Cl[1]^2 \\
								 1 & \Cl[2] & \Cl[2]^2 \\
								 1 & \Cl[3] & \Cl[3]^2 \\
               \end{pmatrix},   
  \quad %
  \Eqm[2](v) =
         \begin{pmatrix} 
           1 & \Ct[1]^2 & \Ct[1]^4 \\
           1 & \Ct[2]^2 & \Ct[2]^4 \\
           1 & \Ct[3]^2 & \Ct[3]^4 \\
         \end{pmatrix},
  \quad  %
  \Eqm[3](v) = \begin{pmatrix}
								 \Cs[1]\Ct[1] [\Cl] \\
								 \Cs[2]\Ct[2] [\Cl] \\
								 \Cs[3]\Ct[3] [\Cl]
							 \end{pmatrix}
  \quad \hbox{ and }\quad
  \Eqm[4](v) = \begin{pmatrix}
	               \Cs[1]\Ct[1] [\Ct^2] \\
								 \Cs[2]\Ct[2] [\Ct^2] \\
							 	 \Cs[3]\Ct[3] [\Ct^2]
							\end{pmatrix},\]
	where
    \[
			[\Cl] := (\Cl[1]-\Cl[2])(\Cl[2]-\Cl[3])(\Cl[3]-\Cl[1])=\det \Eqm[1](v)
			\quad \hbox{ and }\quad
			[\Ct^2] := 
			(\Ct[1]^2-\Ct[2]^2)(\Ct[2]^2-\Ct[3]^2)(\Ct[3]^2-\Ct[1]^2)=\det\Eqm[2](v),
		\]
 are equivariant in the sense that $\Eqm[i] (g_{\sigma}v) = g_{\sigma} 
 \Eqm[i](v)$ and $\Eqm[i](g_{\tau}v) =  \Eqm[i](v)$.
\end{lemma}

\begin{proof}
  This is a direct consequence of the discussion in Section~\ref{ssec:HarmonicBasisQuartics} 
  of the action of $\B$ on the basis of Proposition~\ref{prop:HarmonicBasisQuartic}.
\end{proof}

It follows that the entries of $\Eqm[i]^T\Eqm[j]$, for $i,j\in\{1,2,3,4\}$, and of  $\Eqm[i]^{-1}\Eqm[j]$, for $i\in\{1,2\}$ and $j\in\{1,2,3,4\}$, are rational invariants on $\Lambda_4$. With this mechanism, we provide a generating set of invariants of minimal cardinality.

\begin{thm} \label{thm:TernaryQuarticInvariants}
  A generating set of rational invariants for $\R(\Lambda_4)^{\B}$ is given by the $12$ polynomial functions $\setP[4] = \{\Il[i],\Ir[i],\Is[i],\It[i] \mid 1 \leq i \leq 3\}$, whose values at $v$ as in (\ref{coordsPoly}) are given by 
  \begin{align*}
  \It[1] &:= \Ct[1]^2+\Ct[2]^2+\Ct[3]^2, \\
  \It[2] &:= \Ct[1] \Ct[2]\Ct[3],\\
  \It[3] &:= \Ct[1]^4+\Ct[2]^4+\Ct[3]^4, 
  \end{align*}
  and the entries of the $3\times 3$ matrix 
  \[\begin{pmatrix}
  \Il[1] & \Ir[1] & \Is[1] \\
  \Il[2] & \Ir[2] & \Is[2] \\
  \Il[3] & \Ir[3] & \Is[3]
  \end{pmatrix}
  := \begin{pmatrix}
    1 & 1 & 1 \\
    \Ct[1]^2 & \Ct[2]^2 & \Ct[3]^2 \\
    \Ct[1]^4 & \Ct[2]^4 & \Ct[3]^4
    \end{pmatrix}
  \begin{pmatrix}
    \Cl[1] & \Cr[1] & \Cs[1]\Ct[1] [\Ct^2] \\
    \Cl[2] & \Cr[2] & \Cs[2]\Ct[2] [\Ct^2]  \\
    \Cl[3] & \Cr[3] & \Cs[3]\Ct[3] [\Ct^2] 
    \end{pmatrix} 
    \]
where $[\Ct^2] := (\Ct[1]^2-\Ct[2]^2)(\Ct[2]^2-\Ct[3]^2)(\Ct[3]^2-\Ct[1]^2)$.
\end{thm}

\begin{proof}
We first observe that the polynomials $\Il[i],\Ir[i],\Is[i] \in \R[\Lambda_4]$ are the entries of the matrix $\Eqm[2]^T(v) \Eqm(v)$, where $\Eqm[2]$ was introduced in Lemma~\ref{lem:EqMap} and $\Eqm \colon \Lambda_2 \to \R^{3\times 3}$ is given by 
\[\Eqm[](v) := \begin{pmatrix}
    \Cl[1] & \Cr[1] & \Cs[1]\Ct[1] [\Ct^2] \\
    \Cl[2] & \Cr[2] & \Cs[2]\Ct[2] [\Ct^2]  \\
    \Cl[3] & \Cr[3] & \Cs[3]\Ct[3] [\Ct^2] 
  \end{pmatrix}.\]
With Lemma~\ref{lem:EqMap} one easily checks that $\Eqm[]$ is equivariant in the sense that $\Eqm[] ({\sigma}v) = g_{\sigma} \Eqm[](v)$ and $\Eqm[]({\tau}v) =  \Eqm[](v)$. Hence the entries of $\Eqm[2]^T(v) \Eqm(v)$ are invariants. It is also straightforward to check that $\It[1], \It[2], \It[3] \in \R[\Lambda_4]$ are polynomial invariants.

We now show in three steps that any rational invariant can be written as a rational expression in terms of $\It[i],\Il[i],\Ir[i],\Is[i]$, for $i\in\{1,2,3\}$.

	\emph{Step 1: Reducing the problem to rewriting invariant polynomials.}

	For a finite group, the field of rational invariants is the quotient field of the ring of invariants \cite[Theorem 3.3]{PV94}. That means that for any rational invariant $p \in \R(\Lambda_4)^{\GB}$ there exist $p_1,p_0\in \R[\Lambda_4]^{\GB}$ such that $p=\frac{p_1}{p_0}$. 
We are thus left to show that any polynomial invariant
 can be written as a rational expression of $\Il[i],\Ir[i],\Is[i],\It[i]$, for $i\in\{1,2,3\}$. 

  \emph{Step 2: If $p \in \R[\Lambda_4]^{\GB}$ is an invariant whose value at $v$ as in (\ref{coordsPoly}) is given by a polynomial expression in $\Ct[1], \Ct[2], \Ct[3]$ only, then $p$ can be written polynomially in terms of $\It[1], \It[2], \It[3]$.}
  
  First, we consider the monomials $\Ct[1]^i \Ct[2]^j \Ct[3]^k$ of the polynomial expression $p$. By Proposition~\ref{prop:HarmonicBasisQuartic}, a sign-change matrix $g_\tau = \diag(\tau_1, \tau_2, \tau_3) \in \GB$ acts on $q$ by replacing $\Ct[1]^i \Ct[2]^j \Ct[3]^k$ by
    \[\tau_1^{j+k} \tau_2^{i+k} \tau_3^{i+j}  \Ct[1]^i \Ct[2]^j \Ct[3]^k,\]
  so $p$ can only be invariant with respect to all sign-change matrices if for all its monomials $\Ct[1]^i \Ct[2]^j \Ct[3]^k$, the numbers $i+j$, $i+k$ and $j+k$ are even numbers, i.e.\ $i \equiv j \equiv k \pmod 2$. In particular, we can write $p$ as a polynomial in
    \[\It[2] = \Ct[1] \Ct[2] \Ct[3], \quad \delta_1 := \Ct[1]^2, \quad \delta_2 := \Ct[2]^2, \quad \delta_3 := \Ct[3]^2.\]
  
  If $g_\sigma\in \GB$ is a permutation matrix, then $g_\sigma$ acts according to Proposition~\ref{prop:HarmonicBasisQuartic} on $p$ by replacing $\delta_i$ by $\delta_{\sigma_i}$. Therefore, $p$ is a symmetric polynomial expression in the three variables $\delta_1, \delta_2, \delta_3$. By the Fundamental Theorem of symmetric functions,
  $p$ can therefore be written as a polynomial expression in the three symmetric power sum polynomials
  \begin{align*}
    &\delta_1 + \delta_2 + \delta_3 = \It[1], \\
    &\delta_1^2 + \delta_2^2 + \delta_3^2 = \It[3] \text{ and} \\
    &\delta_1^3 + \delta_2^3 + \delta_3^3 = \frac{3}{2} \It[1] \It[3] - \frac{1}{2} \It[1]^3 + 3\It[2]^2.
  \end{align*}
  With this, we have expressed $p$ as a polynomial expression in terms of $\It[1], \It[2], \It[3]$.

  \emph{Step 3: If $p \in \R[\Lambda_4]^{\GB}$ is a polynomial expression in $\Cl,\Cr,\Cs,\Ct$, then $p$ can be written as a rational expression in the invariants $\Il,\Ir,\Is,\It$.}

  We have
  \[
  \begin{pmatrix}
    \Cl[1] & \Cr[1] & \Cs[1]\Ct[1] [\Ct^2] \\
    \Cl[2] & \Cr[2] & \Cs[2]\Ct[2] [\Ct^2]  \\
    \Cl[3] & \Cr[3] & \Cs[3]\Ct[3] [\Ct^2] 
    \end{pmatrix} 
  =  \Eqm[2]^{-T}(v)
    \begin{pmatrix}
  \Il[1] & \Ir[1] & \Is[1] \\
  \Il[2] & \Ir[2] & \Is[2] \\
  \Il[3] & \Ir[3] & \Is[3]
  \end{pmatrix},
    \]
  where $\Eqm[2](v)$ is a matrix that only involves the variables $\Ct[1], \Ct[2], \Ct[3]$. With this, we can replace each occurrence of $\Cl[i],\Cr[i],\Cs[i]$ for $i\in \{1,2,3\}$, by a linear combination of  $\Il[i],\Ir[i],\Is[i]$ with coefficients that are rational expressions in $\Ct[1], \Ct[2], \Ct[3]$. Hence, $p$ is written as a polynomial in  $\Il[i],\Ir[i],\Is[i]$ with coefficients
  that are rational expressions in $\Ct[1], \Ct[2], \Ct[3]$. Since the $\Il[i],\Ir[i],\Is[i]$ are algebraically independent, these rational expressions of  $\Ct[1], \Ct[2], \Ct[3]$ must be invariant. 
  By Step~1 and 2 they can be written as rational expressions of $\It[1], \It[2], \It[3]$.
\end{proof}

Except maybe for \emph{Step~1}, the proof above  shows how to rewrite any 
rational invariants in terms of the minimal generating set. One could 
alternatively rely on a computational proof following \cite{HK08}; that way 
shows that rewriting a rational invariant in terms of this minimal generating 
set of invariants can be done by applying the following rewrite rules to both 
the numerator and denominator:
\[\begin{pmatrix}\Cl[1] \\ \Cl[2] \\ \Cl[3] \end{pmatrix}
          \longrightarrow 
          \frac{[\Ct^2]}{\It[0]}\begin{pmatrix} 
(\Ct[3]^2-\Ct[2]^2)\Ct[2]^2\Ct[3]^2 & \Ct[2]^4-\Ct[3]^4 & \Ct[3]^2-\Ct[2]^2\\
(\Ct[1]^2-\Ct[3]^2)\Ct[3]^2\Ct[1]^2 & \Ct[3]^4-\Ct[1]^4 & \Ct[1]^2-\Ct[3]^2\\
(\Ct[2]^2-\Ct[1]^2)\Ct[1]^2\Ct[2]^2 & \Ct[1]^4-\Ct[2]^4 & \Ct[2]^2-\Ct[1]^2\\
\end{pmatrix} \begin{pmatrix} \Il[1] \\ \Il[2]\\ \Il[3]\end{pmatrix},
\]
\[\begin{pmatrix}\Cr[1] \\ \Cr[2] \\ \Cr[3] \end{pmatrix}
          \longrightarrow 
        \frac{[\Ct^2]}{\It[0]}\begin{pmatrix} 
(\Ct[3]^2-\Ct[2]^2)\Ct[2]^2\Ct[3]^2 & \Ct[2]^4-\Ct[3]^4 & \Ct[3]^2-\Ct[2]^2\\
(\Ct[1]^2-\Ct[3]^2)\Ct[3]^2\Ct[1]^2 & \Ct[3]^4-\Ct[1]^4 & \Ct[1]^2-\Ct[3]^2\\
(\Ct[2]^2-\Ct[1]^2)\Ct[1]^2\Ct[2]^2 & \Ct[1]^4-\Ct[2]^4 & \Ct[2]^2-\Ct[1]^2\\
\end{pmatrix}
  \begin{pmatrix} \Ir[1] \\ \Ir[2]\\ \Ir[3]\end{pmatrix},
\]
\[\begin{pmatrix}\Cs[1]\\ \Cs[2] \\ \Cs[3] \end{pmatrix}  
\longrightarrow  \frac{1}{\It[0]\It[2]}
          \begin{pmatrix}
            \Ct[2]\Ct[3] & 0 & 0 \\ 0& \Ct[3]\Ct[1] & 0  \\  0 & 0 &\Ct[1]\Ct[2]  
           \end{pmatrix} 
           \begin{pmatrix} 
(\Ct[3]^2-\Ct[2]^2)\Ct[2]^2\Ct[3]^2 & \Ct[2]^4-\Ct[3]^4 & \Ct[3]^2-\Ct[2]^2\\
(\Ct[1]^2-\Ct[3]^2)\Ct[3]^2\Ct[1]^2 & \Ct[3]^4-\Ct[1]^4 & \Ct[1]^2-\Ct[3]^2\\
(\Ct[2]^2-\Ct[1]^2)\Ct[1]^2\Ct[2]^2 & \Ct[1]^4-\Ct[2]^4 & \Ct[2]^2-\Ct[1]^2\\
\end{pmatrix}
  \begin{pmatrix} \Is[1] \\ \Is[2]\\ \Is[3]  \end{pmatrix},  
 \]
 \[\Ct[1] \longrightarrow  \frac{1}{\It[2]}(\It[1]\Ct[2]\Ct[3] -\Ct[2]^3\Ct[3]-\Ct[2]\Ct[3]^3),
  \; 
  \Ct[2]^4 \longrightarrow \It[1]\Ct[3]^2-\Ct[2]^2\Ct[3]^2+\It[1]\Ct[3]^2 -\Ct[3]^4+\frac{1}{2}(\It[3]-\It[1]^2),
  \;
  \Ct[3]^6 \longrightarrow  \frac{1}{2}(\It[3]-\It[1]^2)\Ct[3]^2+\It[1]\Ct[3]^4 +\It[2]^2,\]
where
  \[\It[0] := [\Ct^2]^2 = 
  (\Ct[1]^2-\Ct[2]^2)^2(\Ct[2]^2-\Ct[3]^2)^2(\Ct[3]^2-\Ct[1]^2)^2 = 
  \frac{1}{2}\It[3]^3-\frac{1}{4}\It[1]^6-27\It[2]^4+\It[1]^4\It[3] - 
  \frac{5}{4}\It[1]^2\It[3]^2 - 9\It[2]^2\It[1]\It[3] + 5\It[2]^2\It[1]^3.\]
These rewrite rules reflect a (non-reduced) Gröbner basis of the ideal 
of the generic orbit of the action.

The field of rational functions $\R(\Lambda_4)$ has transcendence degree $\dim \Lambda_4 = 12$ over $\R$, and since $\B$ is a finite group, the same is true for the field of invariants $\R(\Lambda_4)^{\B}$. Therefore, the generating set $\setP[4] \subset \R(\Lambda_4)^{\B}$ of rational invariants specified in Theorem~\ref{thm:TernaryQuarticInvariants} is of minimal cardinality. By Corollary~\ref{cor:SliceLemmaGen} and Proposition~\ref{prop:MainSlice}, there are 12 unique rational invariants in $\R(\VV_4)^{\G}$ restricting to the $\B$-invariants $\setP[4]$ on the subspace $\Lambda_4$, and they form a generating set for $\R(\VV_4)^{\G}$. In particular, the cardinality of this generating set attains the lower bound given in Theorem~\ref{thm:NumberInvariants}:

\begin{cor} \label{cor:QuarticFullInvariants}
  There exists a generating set of $\dim \VV_4 - \dim \G = 12$ rational invariants $\setPT[4] \subset \R(\VV_4)^{\G}$ restricting to the invariants $\setP[4] \subset \R[\Lambda_4]^{\B}$ on the subspace $\Lambda_4 \subset \VV_4$ specified in Theorem~\ref{thm:TernaryQuarticInvariants}.
\end{cor}

While the above generating set $\setPT[4] \subset \R(\VV_4)^{\G}$ (respectively, $\setP[4] \subset \R(\Lambda_4)^{\B}$) is minimal, we observe that it does not directly contain the three invariants for quadratic forms from Section~\ref{ssec:QuadrInv}, which can be considered as invariants on $\VV_4$ via the decomposition $\VV_4 = \H_4 \oplus \VQ \VV_2$ (respectively on $\Lambda_4$ via $\Lambda_4 = \H_4 \oplus \VQ \Lambda_2$).
We therefore now introduce an alternative, non-minimal, generating set of 
rational $\B$-invariants on $\Lambda_4$, such that a generating set of 
rational  
invariants of $\GB$ on $\Lambda_2$ is obtained by restriction. We also make 
explicit how to rewrite any other invariants in terms of these.

\begin{thm} \label{thm:TernaryQuarticInvariantsNonMinimal}
  A generating  set of rational invariants for $\R(\Lambda_4)^{\GB}$ is given by the $13$ rational functions $\Ilw[i],\Irw[i],\Isw[i],\Itw[i]$, for $i \in \{1,2,3\}$, and $\Itw[0]$ 
  whose values at $v$ as in (\ref{coordsPoly}) are given by 
    \[\Itw[0] =\Ct[1]\Ct[2]\Ct[3],\]
    \[\Ilw[1] = \Cl[1]+\Cl[2]+\Cl[3], \quad  \Ilw[2] = 
       \Cl[1]^2+\Cl[2]^2+\Cl[3]^2, \quad  \Ilw[3] =\Cl[1]^3+\Cl[2]^3+\Cl[3]^3,\]
  and the entries of the $3\times 3$ matrix
   \[\begin{pmatrix}
       \Irw[1] & \Isw[1]  & \Itw[1] \\
       \Irw[2] & \Isw[2]  & \Itw[2]\\
       \Irw[3] & \Isw[3]  & \Itw[3]
     \end{pmatrix}
  := \begin{pmatrix} 
           1 & \Cl[1] & \Cl[1]^2 \\
           1 & \Cl[2] & \Cl[2]^2 \\
           1 & \Cl[3] & \Cl[3]^2
          \end{pmatrix}^{-1}   
    \begin{pmatrix}
          \Cr[1] & \Cs[1]\Ct[1] [\Cl] & \Ct[1]^2  \\
          \Cr[2] & \Cs[2]\Ct[2] [\Cl] & \Ct[2]^2   \\
          \Cr[3] & \Cs[3]\Ct[3] [\Cl] & \Ct[3]^2  
    \end{pmatrix}.
    \]
  where $[\Cl] := (\Cl[1]-\Cl[2])(\Cl[2]-\Cl[3])(\Cl[3]-\Cl[1])$.

  Let 
    \[\Ilw[0] := (\Cl[1]-\Cl[2])^2(\Cl[2]-\Cl[3])^2(\Cl[3]-\Cl[1])^2=\frac{3}{2}\Ilw[2]\Ilw[1]^4-\frac{1}{6}\Ilw[1]^6 + 6\Ilw[3]\Ilw[2]\Ilw[1]-\frac{4}{3}\Ilw[3]\Ilw[1]^3-\frac{7}{2}\Ilw[2]^2\Ilw[1]^2-3\Ilw[3]^2+\frac{1}{2}\Ilw[2]^3.\]
  Then any rational invariant $p = p_0/p_1 \in \R(\Lambda_4)^{\B}$ can be written in terms of the above generating set by applying the following rewrite rules to both the numerator $p_0$ and the denominator $p_1$:
    \[\begin{pmatrix} \Cr[1] \\ \Cr[2] \\ \Cr[3] \end{pmatrix}
      \longrightarrow 
      \begin{pmatrix} 
           1 & \Cl[1] & \Cl[1]^2 \\
           1 & \Cl[2] & \Cl[2]^2 \\
           1 & \Cl[3] & \Cl[3]^2 \\
      \end{pmatrix}  
      \begin{pmatrix} \Irw[1] \\ \Irw[2]\\ \Irw[3]\end{pmatrix},
    \]
  \[\begin{pmatrix} \Cs[1]\\ \Cs[2] \\ \Cs[3] \end{pmatrix}
    \longrightarrow \frac{[\lambda]}{\Itw[0]\Ilw[0]}
    \begin{pmatrix}
       \Ct[2]\Ct[3] & 0 & 0 \\ 0& \Ct[3]\Ct[1] & 0  \\  0 & 0 &\Ct[1]\Ct[2]  
    \end{pmatrix}
    \begin{pmatrix} 
       1 & \Cl[1] & \Cl[1]^2 \\
       1 & \Cl[2] & \Cl[2]^2 \\
       1 & \Cl[3] & \Cl[3]^2 \\
    \end{pmatrix}  
    \begin{pmatrix} \Isw[1] \\ \Isw[2]\\ \Isw[3]  \end{pmatrix},  
  \]
  \[\Ct[1]\Ct[2]\Ct[3] \longrightarrow \Itw[0],\]
  \[\begin{pmatrix}\Ct[1]^2 \\ \Ct[2]^2 \\ \Ct[3]^2  \end{pmatrix}
     \longrightarrow 
          \begin{pmatrix} 
           1 & \Cl[1] & \Cl[1]^2 \\
           1 & \Cl[2] & \Cl[2]^2 \\
           1 & \Cl[3] & \Cl[3]^2 \\
          \end{pmatrix}  
  \begin{pmatrix} \Itw[1] \\ \Itw[2]\\ \Itw[3] \end{pmatrix},
  \qquad 
  \begin{pmatrix} 
    \Ct[2]\Ct[3] \\ \Ct[1]\Ct[3] \\ \Ct[1]\Ct[2] 
  \end{pmatrix}
  \longrightarrow \frac{1}{\Itw[0]}
  \begin{pmatrix} 
    \Ct[1] (\Itw[1]+\Cl[2]\Itw[2]+\Cl[2]^2\Itw[3])(\Itw[1]+\Cl[3]\Itw[2]+\Cl[3]^2\Itw[3])  \\
    \Ct[2] (\Itw[1]+\Cl[1]\Itw[2]+\Cl[1]^2\Itw[3])(\Itw[1]+\Cl[3]\Itw[2]+\Cl[3]^2\Itw[3])  \\
    \Ct[3] (\Itw[1]+\Cl[1]\Itw[2]+\Cl[1]^2\Itw[3])(\Itw[1]+\Cl[2]\Itw[2]+\Cl[2]^2\Itw[3])
  \end{pmatrix},
  \]
  and finally
 \[
  \Cl[1] \longrightarrow \Il[1]-\Cl[3]-\Cl[2], 
  \Cl[2]^2 \longrightarrow \Il[1]\Cl[2]+\Il[1]\Cl[3]-\Cl[3]\Cl[2]-\Cl[3]^2+\frac{1}{2}\Il[2]-\frac{1}{2}\Il[1]^2, 
  \Cl[3]^3 \longrightarrow \frac{1}{3}\Il[3]+\frac{1}{6}\Il[1]^3-\frac{1}{2}\Il[2]\Il[1]
            +\frac{1}{2}(\Il[2]-\Il[1]^2)\Ct[3]+\Il[1]\Ct[3]^2.
  \]
\end{thm}

One observes that with this generating set, the rewriting of a polynomial 
invariant only introduces powers of $\Itw[0]$ and $\Ilw[0]$ as the denominator. 
This result was first obtained by applying  the construction in \cite[Theorem 
2.16]{HK08}. This latter theorem shows that the coefficients of the reduced 
Gröbner basis of the generic orbit ideal form a generating set and how to 
rewrite any other rational invariants in terms of these. The rewrite rules 
given above can be recognized as a Gröbner basis of the generic orbit ideal. We 
do not present the reduced Gröbner basis as it is rather cumbersome.

\section{Harmonic bases with \texorpdfstring{$\B$}{B\_3}-symmetries and invariants for higher degree} \label{sec:B3Harmonics}

In this section, we extend the description of generating rational invariants for ternary quartics given in Section~\ref{sec:Quartics} to ternary forms of arbitrary even degree. %
According to Section~\ref{sec:SliceMethod}, the $\G$-invariants on $\VV_{2d}$ are uniquely determined by the $\GB$-invariants on $\Lambda_{2d}$.

A crucial step in order to describe invariants of the $\B$-action on $\Lambda_{2d}$ consists in constructing a basis of the vector space $\H_{2d}$ exhibiting certain symmetries with respect to the group $\B$ of signed permutations. We make precise what we mean by this in Section~\ref{ssec:B3HarmonicsFormulation}, and subsequently we give an explicit construction of such a $\B$-equivariant basis.%

Because of the wide-range of applications that involve harmonic functions, this $\B$-equivariant basis is of independent interest. We therefore devote Section~\ref{ssec:HarmonicsComparison} to an illustration of the constructed harmonic basis and briefly discuss connections to other bases for $\H_{2d}$. We mention at this point that our construction of $\B$-equivariant bases for harmonic polynomials would work analogously for the case of odd degree $\H_{2d+1}$, but with a view toward rational $\B$-invariants on $\Lambda_{2d}$, we restrict our treatment to the case of even degree $\H_{2d}$.

Finally, in Section~\ref{ssec:HigherInvariants}, we deduce from the $\B$-equivariant basis for $\H_{2d}$ a generating set of rational $\GB$-invariants on $\Lambda_{2d}$, as a natural extension of Theorem~\ref{thm:TernaryQuarticInvariants}.   

  \subsection{\texorpdfstring{$\B$}{B\_3}-equivariant bases} \label{ssec:B3HarmonicsFormulation}

We construct a basis $\H_{2d}$ ($d \geq 2$) that essentially splits into subsets of three polynomials that are obtained from one another by permutation of the variables. Each of these subsets spans a $\B$-invariant subspace and the action of $\B$ on this subspace is given by signed permutations on the subset. We introduce the following notation to be more precise.

\begin{definition} \label{def:B3Equivariance}
  Let $k \in \N$ and consider two maps $\zeta, \xi \colon \{0,\ldots, k-1\} \to \{0,1\}$. An indexed set 
    \[\{u_{i,j} \mid 1 \leq i \leq 3, 0 \leq j < k\} \subset \VV_{2d}\]
  of $3k$ elements of $\VV_{2d}$ is called a \defstyle{$\B$-equivariant subset of $\VV_{2d}$ with respect to $\zeta, \xi$} (or, a \defstyle{$(\zeta, \xi)$-equivariant subset}) if the action of $\B$ on this set is given as follows: For $\sigma \in \GS$ and $\tau \in \GA$, the permutation and sign-change matrices $g_\sigma$ and $g_\tau$ act by
    \[\forall i,j: \qquad g_\sigma u_{i,j} = \sgn(\sigma)^{\zeta(j)}\: u_{\sigma(i),j} \qquad \text{and} \qquad g_\tau u_{i,j} = \left(\frac{\tau_1 \tau_2 \tau_3}{\tau_i}\right)^{\xi(j)} u_{i,j}.\]
\end{definition}

Our aim  is to find a $\B$-equivariant basis for the vector space $\H_{2d}$ for arbitrary $d \geq 2$, similar to Proposition~\ref{prop:HarmonicBasisQuartic} in the case of $\H_4$. For $d \not\equiv 2 \pmod 3$, however, we need to slightly relax this aim, and instead of a basis, we shall determine a $\B$-equivariant \emph{spanning set} of the vector space $\H_{2d}$ together with the linear relationships satisfied. %

\begin{thm} \label{thm:HarmonicBasisHigher}
  Let $d \geq 2$ and let $k := \left\lceil \frac{4d+1}{3} \right\rceil$. Then there exist $3k$ harmonic polynomials $u_{i,j} \in \H_{2d}$ for $i \in \{1,2,3\}, j \in \{0, \ldots, k-1\}$ forming a $(\zeta, \xi)$-equivariant set for some maps $\zeta, \xi \colon \{0,\ldots, k-1\} \to \{0,1\}$, such that:
  \begin{enumerate}[(i)]
    \item \label{it:HBasisGen} The harmonic polynomials $u_{i,j} \in \H_{2d}$ for $i \in \{1,2,3\}$, $j \in \{0, \ldots, k-1\}$ span the vector space $\H_{2d}$.
    \item \label{it:HBasisRel} \begin{enumerate}[(a)]
            \item \label{it:HbasisRel2} If $d \equiv 2 \pmod 3$, then the $u_{i,j}$ form a basis of $\H_{2d}$.
            \item \label{it:HbasisRel1} If $d \equiv 1 \pmod 3$, then $\dim \H_{2d} = 3k-1$ and the $u_{i,j}$ satisfy the linear relation
                    \[u_{1,0} + u_{2,0} + u_{3,0} = 0.\]
            \item \label{it:HbasisRel0} If $d \equiv 0 \pmod 3$, then $\dim \H_{2d} = 3k-2$ and the $u_{i,j}$ satisfy the relations 
                    \[u_{1,0} = u_{2,0} = u_{3,0}.\]
           \end{enumerate}
  \end{enumerate}
\end{thm}

We observe that Theorem~\ref{thm:HarmonicBasisHigher} generalizes Proposition~\ref{prop:HarmonicBasisQuartic}: For $2d = 4$, we have $k = 3$ and the harmonic polynomials $u_{i,0} := \Vr{i}$, $u_{i,1} := \Vs{i}$ and $u_{i,2} := \Vt{i}$ ($i \in \{1,2,3\}$) form a basis of $\H_{2d}$ satisfying the assertions of Theorem~\ref{thm:HarmonicBasisHigher} for
  \[\zeta(j) = \begin{cases} 1 &\text{if } j = 1, \\ 0 &\text{otherwise}\end{cases} \qquad \text{and} \qquad \xi(j) = \begin{cases} 1 &\text{if } j \in \{1,2\}, \\ 0 &\text{otherwise}.\end{cases}\]

We formulated Theorem~\ref{thm:HarmonicBasisHigher} as an existence result, but the crucial part for applications is the constructive proof given below. We provide explicit closed formulas for the elements $u_{i,j}$, in addition to algorithmic constructions that appear as proofs.

  \subsection{Construction of \texorpdfstring{$\B$}{B\_3}-equivariant harmonic bases} \label{ssec:B3HarmonicsConstruction}
  
We prove Theorem~\ref{thm:HarmonicBasisHigher} by constructing a spanning set with the asserted properties. %

For nonnegative integers $i,j,k$ and $r = i+j+k$, we denote the \emph{multinomial coefficient} 
  \[\binom{r}{i,j,k} := \frac{r!}{i!j!k!}.\] 

The following result is a simple characterization of harmonic polynomials.

\begin{lemma} \label{lem:TernaryHarmonicsCriterion}
  Let $d \geq 2$. A ternary form
    \[v = \sum_{i+j+k = 2d} \binom{2d}{i,j,k} \beta_{i,j,k} \; x^i y^j z^k \in V_{2d}\]
  lies in $\H_{2d}$ if and only if $\beta_{i+2,j,k} + \beta_{i,j+2,k} + \beta_{i,j,k+2} = 0$ holds for all $i,j,k \in \N$ such that $i+j+k = 2d-2$.
\end{lemma}

\begin{proof}
  Recall that $\H_{2d}$ is the orthogonal complement of $\VQ \VV_{2d-2}$ in $\VV_{2d}$ with respect to the apolar product, where $\VQ = x^2+y^2+z^2$. Hence, $v \in \H_{2d}$ if and only if $\langle v, \VQ x^i y^j z^k \rangle = 0$ for all $i,j,k$ such that $i+j+k = 2d-2$. From
    \[\langle x^{i'} y^{j'} z^{k'}, x^i y^j z^k  \rangle = \begin{cases} i!j!k! & \text{if } i=i', j=j', k=k' \\ 0 & \text{otherwise},\end{cases}\]
  we see that
    \[\langle v, \VQ x^i y^j z^k \rangle = (2d)! (\beta_{i+2,j,k} + \beta_{i,j+2,k} + \beta_{i,j,k+2}),\]
  so the claim follows.
\end{proof}

\begin{lemma} \label{lem:evenHarmonics}
  Let $d \geq 2$ and $0 \leq \ell \leq d$. There exists a harmonic polynomial 
  $\Vg{\ell} \in \H_{2d}$, unique up to scaling, with the following four 
  properties:
  \begin{enumerate}[(1)]
    \item \label{it:evenMon} Each monomial of $\Vg{\ell}$ is of even degree in each of the variables $x$, $y$ and $z$.
    \item \label{it:evenDeg} The highest degree in which the variable $x$ occurs in $\Vg{\ell}$ is $2\ell$.
    \item \label{it:evenZeroes} The monomial $y^{2j} z^{2k}$ (for $j+k = d$) does \emph{not} occur in the expression $\Vg{\ell}$ if $|j-k| < \ell$.
    \item \label{it:evenSym} Interchanging the variables $y$ and $z$ in the 
    expression $\Vg{\ell}$ gives $(-1)^{d-\ell} \Vg{\ell}$.
  \end{enumerate}
  In this unique polynomial $\Vg{\ell}$, all monomials in $\VV_{2d}$ of degree 
  $2\ell$ in $x$ and even degree in each $y$ and $z$ occur with a non-zero 
  coefficient.
\end{lemma}

\begin{proof}
  By Lemma~\ref{lem:TernaryHarmonicsCriterion}, determining the expression 
  \[\Vg{\ell} = \sum_{i+j+k = d} \binom{2d}{2i,2j,2k} \beta_{i,j,k} \; x^{2i} y^{2j} z^{2k}\]
  amounts to finding $\beta_{i,j,k} \in \R$ for $i,j,k \geq 0$, $i+j+k = d$ such that
  \begin{enumerate}[(i)]
    \item $\beta_{i+1,j,k} + \beta_{i,j+1,k} + \beta_{i,j,k+1} = 0$ for all $i,j,k \geq 0, i+j+k = d-1$, \label{itP:harm} \label{itP:first}
    \item $\beta_{i,j,k} = 0$ if $i > \ell$, \label{itP:deg}
    \item $\beta_{\ell,j,k} \neq 0$ for some $j, k$, \label{itP:degNZ}
    \item $\beta_{i,j,k} = (-1)^{d-\ell} \beta_{i,k,j}$ for all $i+j+k = d$, \label{itP:sym}
    \item $\beta_{0,j,k} = 0$ if $|j-k| < \ell$. \label{itP:zero} \label{itP:last}
  \end{enumerate}
  From this, the coefficients $\beta_{i,j,k}$ can be determined iteratively. Combining (\ref*{itP:zero}) with (\ref*{itP:harm}), one shows by induction on $i$ that
    \begin{equation} \label{eq:zeroBetas}
      \beta_{i,j,k} = 0 \quad \text{ if } i + |j-k| < \ell.
    \end{equation}
  
  Moreover, (\ref*{itP:deg}) imposes $\beta_{i,j,k} = 0$ if $i > \ell$. Then (\ref*{itP:harm}) implies that $\beta_{\ell,j+1,k} = -\beta_{\ell, j, k+1}$ for all $j,k \geq 0$ with $\ell+j+k=d-1$. Hence, $\beta_{\ell, j, k} = (-1)^j \beta_{\ell, 0, d-\ell}$. Because of (\ref*{itP:degNZ}), we may assume (after rescaling all $\beta_{i,j,k}$) that $\beta_{\ell, j, k} = (-1)^j$ for all $j+k = d-\ell$.
  
  If $d-\ell$ is even and $s := \frac{d-\ell}{2}$, we additionally have $\beta_{\ell-1, s, s+1} + \beta_{\ell-1, s+1, s} + \beta_{\ell, s, s} = 0$ by (\ref*{itP:harm}). From $\beta_{\ell,s,s} = (-1)^s$, we conclude with (\ref*{itP:sym}) that
    \begin{equation} \label{eq:excBeta}
      \beta_{\ell-1,s,s+1} = \beta_{\ell-1,s+1,s} = \frac{(-1)^{s+1}}{2} \qquad (\neq 0).
    \end{equation}
  
  At this stage we have determined all $\beta_{i,j,k}$ for $i+j+k = d$ such that $i \geq \ell$ or $|j-k| < \max\{\ell-i, 2\}$ and they are compatible with (\ref*{itP:first})--(\ref*{itP:last}) in the sense that all statements from (\ref*{itP:first})--(\ref*{itP:last}) only involving those $\beta_{i,j,k}$ are satisfied. We now proceed to determine $\beta_{\ell-K,j,k}$ for $K \in \N$ increasing from $1$ to $\ell$, assuming that the values $\beta_{i,j,k}$ for $i > \ell-K$ are already known and compatible with (\ref*{itP:first})--(\ref*{itP:last}).
  
  Recall that from \eqref{eq:zeroBetas} and \eqref{eq:excBeta} we already know the values $\beta_{\ell-K,j,k}$ whenever $|j-k| < \max\{K, 2\}$. For increasing values of $|j-k|$, we then iteratively obtain the values $\beta_{\ell-K,j,k}$ for $j > k$ from the recursion
    \[\beta_{\ell-K,j+1,k} = - \beta_{\ell-(K-1),j,k} - \beta_{\ell-K,j,k+1}\]
  and for $k > j$ from the recursion
    \[\beta_{\ell-K,j,k+1} = - \beta_{\ell-(K-1),j,k} - \beta_{\ell-K,j+1,k}\]
  (both resulting from (\ref*{itP:harm})). This determines all $\beta_{\ell-K, j, k}$. Note that $\beta_{\ell-K,j,k} = (-1)^{d-\ell} \beta_{\ell-K,k,j}$ follows iteratively from the assumption that property~(\ref*{itP:sym}) holds for all previously constructed $\beta_{i,j,k}$. Hence, we determined the unique values for $\beta_{\ell-K,j,k}$ such that the properties (\ref*{itP:first})--(\ref*{itP:last}) are still preserved. This concludes the construction.
\end{proof}

\begin{remark} \label{rem:ExplicitFormulae}
  The proof above is actually an iterative construction of the coefficients $\beta_{i,j,k}$ defining $\Vg{\ell}$. We now give an explicit formula: If $d - \ell$ is odd and $s := \frac{d-\ell+1}{2}$, then the values
    \[\beta_{i,j,k} = (-1)^j \binom{j-s}{\ell-i}\]
  are easily checked to satisfy (\ref*{itP:first})--(\ref*{itP:last}). 
  Here, we employ for $a, b \in \Z$ the convention %
    \[\binom{a}{b} := \begin{cases} \frac{a(a-1) \ldots (a-b+1)}{b!} &\text{if } b \geq 0, \\ 0 & \text{otherwise} \end{cases}\]
  and use the identities $\binom{a+1}{b+1} = \binom{a}{b} + \binom{a}{b+1}$ and $\binom{a}{b} = (-1)^b\binom{b-a-1}{b}$.
  On the other hand, if $d - \ell$ is even and $s := \frac{d-\ell}{2}$, then
    \[\beta_{i,j,k} = (-1)^j \binom{j-s}{\ell-i} + (-1)^j \binom{j-s-1}{\ell-i}\]
  satisfy the desired properties. This gives an explicit closed formula for the harmonic polynomials $\Vg{\ell}$.
\end{remark}

Similar to Lemma~\ref{lem:evenHarmonics}, we also obtain the following result.

\begin{lemma} \label{lem:oddHarmonics}
  Let $d \geq 2$ and $0 \leq \ell \leq d-1$. Then there exists a harmonic 
  polynomial $\Vf{\ell} \in \H_{2d}$, unique up to scaling, with the following 
  four properties:
  \begin{enumerate}[(1)]
    \item \label{it:oddMon} Each monomial of $\Vf{\ell}$ is of even degree in $x$ and of odd degree in each of the variables $y$ and $z$.
    \item \label{it:oddDeg} The highest degree in which the variable $x$ occurs in $\Vf{\ell}$ is $2\ell$.
    \item \label{it:oddZeroes} The monomial $y^{2j+1} z^{2k+1}$ (for $j+k = d-1$) does \emph{not} occur in the expression $\Vf{\ell}$ if $|j-k| < \ell$.
    \item \label{it:oddSym} Interchanging the variables $y$ and $z$ in the 
    expression $\Vf{\ell}$ gives $(-1)^{d-\ell+1} \Vf{\ell}$.
  \end{enumerate}
  In this unique polynomial $\Vf{\ell}$, all monomials in $\VV_{2d}$ of degree 
  $2\ell$ in $x$ and odd degree in each $y$ and $z$ occur with a non-zero 
  coefficient.
\end{lemma}

\begin{proof}
  Writing 
    \[\Vf{\ell} = \sum_{i+j+k = d-1} \binom{2d}{2i,2j+1,2k+1} \beta_{i,j,k} \; x^{2i} y^{2j+1} z^{2k+1},\]
  the coefficients $\beta_{i,j,k}$ for $i+j+k = d-1$ have to meet the same conditions as before (replacing $d$ by $d-1$), so the construction is the same as before (and the corresponding explicit formulas from Remark~\ref{rem:ExplicitFormulae} apply).
\end{proof}

While the harmonic polynomials $\Vg{\ell}$, $\Vf{\ell}$ are (anti-)symmetric with respect to the variables $y$ and $z$, the variable $x$ plays a special role. We therefore now consider the expressions that arise from $\Vg{\ell}, \Vf{\ell}$ by cyclically permuting the variables $x$, $y$ and $z$. For this, we consider the cycle $\Vc = (123) \in \GS$ and its associated permutation matrix
  \[g_\Vc = \begin{pmatrix} 0 & 0 & 1 \\ 1 & 0 & 0 \\ 0 & 1 & 0 \end{pmatrix} \in \B.\]

\begin{lemma} \label{lem:linearIndep}
  Let $d \geq 2$ and consider $\Vf{\ell}, \Vg{\ell} \in \H_{2d}$ as in Lemmas~\ref{lem:oddHarmonics} and \ref{lem:evenHarmonics}.
  Then the harmonic polynomials
    \begin{align*}
      &\Vf{0}, \ g_\Vc \Vf{0}, \ g_\Vc^2 \Vf{0}, \quad \Vf{1}, \ g_\Vc \Vf{1}, \ g_\Vc^2 \Vf{1}, \quad \dots, \quad \Vf{d-1}, \ g_\Vc \Vf{d-1}, \ g_\Vc^2 \Vf{d-1}, \\
      &\Vg{0}, \ g_\Vc \Vg{0}, \ g_\Vc^2 \Vg{0}, \quad \Vg{1}, \ g_\Vc \Vg{1}, \ g_\Vc^2 \Vg{1}, \quad \dots, \quad \Vg{r}, \ g_\Vc \Vg{r}, \ g_\Vc^2 \Vg{r}
    \end{align*}
  are linearly independent in $\H_{2d}$ for $r := \left\lfloor \frac{d-2}{3} \right\rfloor$.
\end{lemma}

\begin{proof}
  The monomials in $\Vg{\ell}, g_\Vc \Vg{\ell}, g_\Vc^2 \Vg{\ell}$ are of the form $x^{2i} y^{2j} z^{2k}$, in $\Vf{\ell}$ of the form $x^{2i} y^{2j+1} z^{2k+1}$, in $g_\Vc \Vf{\ell}$ of the form $x^{2i+1} y^{2j} z^{2k+1}$, and in $g_\Vc^2 \Vf{\ell}$ of the form $x^{2i+1} y^{2j+1} z^{2k}$. Therefore, we can show the linear independence of those families of harmonic polynomials separately.
  
  The linear independence of $\Vf{0}, \Vf{1}, \ldots, \Vf{d-1} \in \H_{2d}$ is a consequence of property (\ref*{it:oddDeg}) in Lemma~\ref{lem:oddHarmonics}. The same for $g_\Vc \Vf{0}, \ldots, g_\Vc \Vf{d-1}$ and for $g_\Vc^2 \Vf{0}, \ldots, g_\Vc^2 \Vf{d-1}$ follows directly.
  
  To see that $\{g_\Vc^{i-1} \Vg{\ell} \mid 1 \leq i \leq 3, 0 \leq \ell \leq r\}$ is a linearly independent set, suppose that $\sum_{i = 1}^3 \sum_{\ell=0}^r \alpha_{i,\ell} g_\Vc^{i-1} \Vg{\ell} = 0$ for some $\alpha_{i,\ell} \in \R$, not all equal to zero. Let $s \in \{0,\ldots, r\}$ be maximal such that $\alpha_{i,s} \neq 0$ for some $i \in \{1,2,3\}$. We may assume that $\alpha_{1,s} \neq 0$. Note that $s \leq r = \left\lfloor \frac{d-2}{3} \right\rfloor$ implies that there exist integers $j, k > s$ such that $j+k+s=d$. It follows from property~(\ref*{it:evenDeg}) in Lemma~\ref{lem:evenHarmonics} that the coefficient of $x^{2s} y^{2j} z^{2k}$ in the expression $\sum_{i = 1}^3 \sum_{\ell=0}^s \alpha_{i,\ell} g_\Vc^{i-1} \Vg{\ell}$ is $\alpha_{1,\ell} \lambda_1$, where $\lambda_1 \neq 0$ is the coefficient of $x^{2s} y^{2j} z^{2k}$ in $\Vg{s}$. This implies $\sum_{i = 1}^3 \sum_{\ell=0}^s \alpha_{i,\ell} g_\Vc^{i-1} \Vg{\ell} \neq 0$, contrary to the assumption.
\end{proof}

\begin{lemma} \label{lem:ZeroSum}
  Let $d = 3k_0+1$ for some $k_0 \geq 1$ and let $g_\Vc \in \B$ and $\Vg{\ell}$ as before. Then 
    \[\Vg{k_0} + g_\Vc \Vg{k_0} + g_\Vc^2 \Vg{k_0} = 0.\]
\end{lemma}

\begin{proof}
  From Lemma~\ref{lem:linearIndep} we know that the set $\{g_\Vc^{i-1}\Vg{\ell} \mid 1 \leq i \leq 3, 0 \leq \ell \leq k_0-1\}$ forms a basis for a $3k_0$-dimensional subspace $W \subset \H_{2d}$. By properties (\ref*{it:evenDeg}) and (\ref*{it:evenMon}) in Lemma~\ref{lem:evenHarmonics}, all monomials occurring in any $g_\Vc^{i-1}\Vg{\ell}$ for $1 \leq i \leq 3, 0 \leq \ell \leq k_0-1$ must be of the form $x^{2i} y^{2j} z^{2k}$ with $\min\{i,j,k\} \leq k_0-1$, so $W$ is contained in the subspace $U$ of $V_{2d}$ spanned by the monomials of this form. Lemma~\ref{lem:TernaryHarmonicsCriterion} implies that $\dim U \cap \H_{2d} = d-1 = 3k_0$, so in fact $W = U \cap \H_{2d}$.
  
  By properties (\ref*{it:evenDeg}) and (\ref*{it:evenSym}) in Lemma~\ref{lem:evenHarmonics} we have $\Vg{k_0} = \lambda x^{2k_0} y^{2k_0+2} z^{2k_0} - \lambda x^{2k_0} y^{2k_0+2} z^{2k_0} + u$ for some $u \in U$ and $\lambda \in \R$. Therefore,
    \[\Vg{k_0} + g_\Vc \Vg{k_0} + g_\Vc^2 \Vg{k_0} = u + g_\Vc u + g_\Vc^2 u =: u' \in U \cap \H_{2d}\]
  and we can write $u' = \sum_{i=1}^3 \sum_{\ell=0}^{k_0-1} \alpha_{i, \ell} g_\Vc^{i-1} \Vg{\ell}$. Suppose for contradiction that not all $\alpha_{i,\ell}$ are zero and let $s \in \{0, \ldots, k_0-1\}$ be minimal such that $\alpha_{i,s} \neq 0$ for some $i \in \{1,2,3\}$. By symmetry, we may assume $\alpha_{1,s} \neq 0$.
  
  If $d+s$ is even, then the monomial $y^{d+s} z^{d-s}$ occurs in $\Vg{s}$ with a non-zero coefficient $\mu \in \R$, but not in $\Vg{\ell}$ for $\ell > s$ by property (\ref*{it:evenZeroes}) in Lemma~\ref{lem:evenHarmonics}. Observe also that the $y$-degree and the $z$-degree of this monomial are both larger than $2k_0$, so $y^{d+s} z^{d-s}$ is not contained as a monomial in any of the expressions $g_\Vc \Vg{\ell}$ or $g_\Vc^2 \Vg{\ell}$ for $\ell \leq k_0$. Therefore, the monomial $y^{d+s} z^{d-s}$ occurs in $u'$ with the non-zero coefficient $\alpha_{1, s} \mu$, but not in $\Vg{k_0}$ or $g_\Vc \Vg{k_0}$ or $g_\Vc^2 \Vg{k_0}$, a contradiction.
  
  If $d+s$ is odd, then $s \leq k_0-2$ and we consider the monomial $y^{d+s+1} z^{d-s-1}$. As before, this monomial does not occur in $\Vg{\ell}$ for $\ell > s+1$ or in any of the expressions $g_\Vc \Vg{\ell}$ or $g_\Vc^2 \Vg{\ell}$ for $0 \leq \ell \leq k_0$. Hence, its coefficient in $u'$ is $\alpha_{1,s} \mu + \alpha_{1,s+1}\mu' = 0$, where $\mu, \mu' \neq 0$ are the coefficients of $y^{d+s+1} z^{d-s-1}$ in $\Vg{s}$ and $\Vg{s+1}$, respectively. In the same way we can consider the monomial $y^{d-s-1} z^{d+s+1}$, which then gives $-\alpha_{1,s} \mu + \alpha_{1,s+1}\mu' = 0$ by the (anti-)symmetry of $\Vg{s}$ and $\Vg{s+1}$ due to property (\ref*{it:evenSym}) in Lemma~\ref{lem:evenHarmonics}. From both equations we deduce $\alpha_{1,s} = 0$, a contradiction.
\end{proof}

\begin{proof}[Proof of Theorem~\ref{thm:HarmonicBasisHigher}]
  We use the notations of the previous Lemmas and denote 
    \[k_0 := \left\lceil \frac{d-2}{3} \right\rceil \quad \text{ and } \quad k := \left\lceil \frac{4d+1}{3} \right\rceil = k_0 + d+1.\]
  For $i \in \{1, 2, 3\}$ and $j \in \{1, \ldots, k-1\}$, we define
   \[u_{i,j} := \begin{cases}
                  g_\Vc^{i-1} \Vg{k_0-j}, &\text{if } 1 \leq j \leq k_0, \\
                  g_\Vc^{i-1} \Vf{j-k_0-1}, &\text{if } k_0+1 \leq j < k. \\
                \end{cases}\]
  Note that we have only defined $u_{i,j}$ for $j \geq 1$. Additionally, let
  \[u_{i,0} := \begin{cases}
                 g_\Vc^{i-1} \Vg{k_0}, &\text{if } d \equiv 1 \text{ or } 2 \pmod 3, \\
                 \Vg{k_0} + g_\Vc \Vg{k_0} + g_\Vc^2 \Vg{k_0}, &\text{if } d \equiv 0 \pmod 3.
               \end{cases}\]
  In the following, we show that these harmonic polynomials $u_{i,j} \in \H_{2d}$ for $i \in \{1,2,3\}$, $j \in \{0, \ldots, k\}$ satisfy the properties stated in Theorem~\ref{thm:HarmonicBasisHigher} if we define $\zeta, \xi \colon \{0, \ldots, k-1\} \to \{0,1\}$ as 
    \[\zeta(j) := \begin{cases} 0 &\text{if $d+k_0+j$ even}, \\ 1 &\text{otherwise},\end{cases} \qquad \text{and} \qquad \xi(j) = \begin{cases} 0 &\text{if } 0 \leq j \leq k_0, \\ 1 &\text{if } k_0 + 1 \leq j < k.\end{cases}\]
  
  First, we check the $(\zeta, \xi)$-equivariance of the harmonic polynomials $u_{i,j}$. The symmetric group $\GS$ is generated by the cycle $\Vc = (123)$ together with the transposition $\Vtr = (23)$, so in order to show $g_\sigma u_{i,j} = \sgn(\sigma)^{\zeta(j)} u_{\sigma(i),j}$ for all $\sigma \in \GS$, it suffices to check this for $\sigma = \Vc$ and $\sigma = \Vtr$.
  By definition of $u_{i,j}$ we have 
    \[g_\Vc u_{1,j} = u_{2,j}, \quad g_\Vc u_{2,j} = u_{3,j}, \quad g_\Vc u_{3,j} = u_{1,j},\]
  i.e.\ $g_\Vc u_{i,j} = u_{\Vc(i), j}$. Note that $\sgn(\Vc) = 1$. Since $g_\Vtr$ acts by interchanging the variables $y$ and $z$, it follows from property~(\ref*{it:oddSym}) in Lemmas~\ref{lem:evenHarmonics} and \ref{lem:oddHarmonics} that $g_\Vtr u_{1,j} = (-1)^{\zeta(j)} u_{1,j}$. Indeed, for $j \geq 1$ or $d \not\equiv 0 \pmod 3$, this is immediate from the definition of $u_{1,j}$, while for $j = 0$ and $d \equiv 0 \pmod 3$, it follows after using the identities $g_\Vtr g_\Vc = g_\Vc^2 g_\Vtr$ and $g_\Vtr g_\Vc^2 = g_\Vc g_\Vtr$ and observing that $\zeta(0) = 0$ in this case. After that, we deduce
    \[g_\Vtr u_{2,j}  = g_\Vtr (g_\Vc u_{1,j}) = g_\Vc^2 (g_\Vtr u_{1,j}) = (-1)^{\zeta(j)} g_\Vc^2(u_{1,j}) = (-1)^{\zeta(j)} u_{3,j},\]
  (where we used again $g_\Vtr g_\Vc = g_\Vc^2 g_\Vtr$) and therefore also $g_\Vtr u_{3,j} = (-1)^{\zeta(j)} u_{2,j}$. All together, we have $g_\Vtr u_{i,j} = (-1)^{\zeta(j)} u_{\Vtr(i), j}$. Note that $\sgn \Vtr = -1$. With this, we have verified verified the equivariance for permutation matrices.
  
  Now let $\tau \in \GA$. Note that property~(\ref*{it:evenMon}) of Lemmata~\ref{lem:evenHarmonics} and \ref{lem:oddHarmonics} ensures that $g_\tau u_{1,j} = (\tau_2 \tau_3)^{\xi(j)} u_{1,j}$. Using 
    \[g_\tau g_\Vc = \diag(\tau_2, \tau_3, \tau_1) \quad \text{and} \quad g_\tau g_\Vc^2 = \diag(\tau_3, \tau_1, \tau_2),\]
  we also see that $g_\tau u_{2,j} = (\tau_3 \tau_1)^{\xi(j)} u_{1,j}$ and $g_\tau u_{3,j} = (\tau_1 \tau_2)^{\xi(j)} u_{1,j}$. This concludes the proof of the $(\zeta, \xi)$-equivariance.
  
  If $d \equiv 2 \pmod 3$, then $k = \frac{4d+1}{3}$, so $\dim \H_{2d} = 3k$. The $3k$ harmonic polynomials $u_{i,j}$ are linearly independent by Lemma~\ref{lem:linearIndep}, hence they form a basis of $\H_{2d}$. This verifies properties (\ref*{it:HBasisGen}) and (\ref*{it:HBasisRel}) for $d \equiv 2 \pmod 3$.
  
  If $d \equiv 0 \pmod 3$, then $\dim \H_{2d} = 4d+1 = 3k-2$ and the identities $u_{1,0} = u_{2,0} = u_{3,0}$ hold by definition of $u_{i,0}$. This shows property (\ref*{it:HBasisRel}). The $3k-3$ harmonic polynomials $u_{i,j}$ for $i \in \{1,2,3\}$, $j \in \{1, \dots, k-1\}$ are linearly independent by Lemma~\ref{lem:linearIndep}, and the subspace of $\H_{2d}$ spanned by them does not contain the harmonic polynomial $u_{1,0} = u_{2,0} = u_{3,0}$, because the latter contains the monomial $x^{\frac{2d}{3}} y^{\frac{2d}{3}} z^{\frac{2d}{3}}$, which does not occur in $u_{i,j}$ for $j \geq 1$. This also establishes (\ref*{it:HBasisGen}) for $d \equiv 0 \pmod 3$.
  
  We are left with the case $d \equiv 1 \pmod 3$. Note that then $\dim \H_{2d} = 4d+1 = 3k-1$. By Lemma~\ref{lem:linearIndep}, the $3(k-1)$ polynomials $u_{i,j}$ for $i \in \{1,2,3\}$, $j \in \{1, \ldots, k-1\}$ span a $3(k-1)$-dimensional subspace of $\H_{2d}$. From considering the (non-)occurrence of the monomials
    \[x^{\frac{2d+2}{3}} y^{\frac{2d-1}{3}} z^{\frac{2d-1}{3}}, \qquad x^{\frac{2d-1}{3}} y^{\frac{2d+2}{3}} z^{\frac{2d-1}{3}} \qquad \text{and} \qquad x^{\frac{2d-1}{3}} y^{\frac{2d-1}{3}} z^{\frac{2d+2}{3}}\]
  in the expressions $u_{i,j}$, it follows that the subspace of $\H_{2d}$ spanned by all the $3k$ harmonic polynomials $u_{i,j}$ must be at least of dimension $3k-1$ and must hence coincide with $\H_{2d}$. Together with Lemma~\ref{lem:ZeroSum}, this verifies (\ref*{it:HBasisGen}) and (\ref*{it:HBasisRel}) in this last case, concluding the proof.
\end{proof}   

  \subsection{Illustrations and examples} \label{ssec:HarmonicsComparison}
  
For degree~4, the construction above reproduces the basis for $\H_4$ given in Proposition~\ref{prop:HarmonicBasisQuartic} (up to scalars). For arbitrary $d$, the proof of Theorem~\ref{thm:HarmonicBasisHigher} together with Remark~\ref{rem:ExplicitFormulae} give explicit closed formulas for elements $u_{i,j}$ of a $\B$-equivariant spanning set of $\H_{2d}$, by means of several case distinctions.

For example, tracing back the definitions gives: If $1 \leq j \leq k_0 = \left\lceil \frac{d-2}{3} \right\rceil$ and $d+j-k_0$ is odd, then
 \[u_{1,j} = \sum_{i_1+i_2+i_3 = d} (-1)^{i_2} \binom{i_2-s}{k_0-j-i_1} \;x^{2i_1} y^{2i_2} z^{2i_3},\]
where $s := \frac{1}{2}(d-k_0+j+1)$; and $u_{2,j}$, $u_{3,j}$ are obtained from this by cyclically permuting $\xyz$. For the other elements $u_{i,j}$, similar formulas can be written out.

We give examples for $\B$-equivariant spanning sets of $\H_{2d}$ for $d=3$ and $4$. For this, we will denote the constructed harmonic polynomials $u_{i,j} \in \H_{2d}$ from Theorem~\ref{thm:HarmonicBasisHigher} by $u_{i,j}^\supind{2d}$ to avoid confusion between different values of $d$.

In degree~6, the following 13 harmonic polynomials form the $\GB$-equivariant basis for $\H_6$ constructed in Theorem~\ref{thm:HarmonicBasisHigher}:

\begingroup \scriptsize
\[ \begin{array}{ll}
  u_{1, 0}^\supind{6} = u_{2, 0}^\supind{6} = u_{3, 0}^\supind{6} & = -2x^6-2z^6-2y^6+15 x^4 y^2+15 y^4 z^2+15 z^4 x^2+15 x^4 z^2+15 y^4 x^2+15 z^4 y^2-180 x^2 y^2 z^2,
  \end{array} \]
\[ \begin{array}{ll}
  u_{1, 1}^{(6)} &= -y^6+15 y^4 z^2-15 y^2 z^4+z^6, \\
  u_{2, 1}^{(6)} &= -z^6+15 z^4 x^2-15 z^2 x^4+x^6, \\
  u_{3, 1}^{(6)} &= -x^6+15 x^4 y^2-15 x^2 y^4+y^6, 
  \end{array} \qquad \begin{array}{ll}
  u_{1, 2}^{(6)} &= 3 y^5 z-10 y^3 z^3+3 y z^5, \\
  u_{2, 2}^{(6)} &= 3 z^5 x-10 z^3 x^3+3 z x^5, \\
  u_{3, 2}^{(6)} &= 3 x^5 y-10 x^3 y^3+3 x y^5, 
  \end{array}
\]
 \[ \begin{array}{ll}
  u_{1, 3}^{(6)} &= -10 x^2 y^3 z+10 x^2 y z^3+y^5 z-y z^5, \\
  u_{2, 3}^{(6)} &= -10 y^2 z^3 x+10 y^2 z x^3+z^5 x-z x^5, \\
  u_{3, 3}^{(6)} &= -10 z^2 x^3 y+10 z^2 x y^3+x^5 y-x y^5, 
   \end{array} \qquad \begin{array}{ll}
  u_{1, 4}^{(6)} &= 10 x^4 y z-10 x^2 y^3 z-10 x^2 y z^3+y^5 z+y z^5, \\
  u_{2, 4}^{(6)} &= 10 y^4 z x-10 y^2 z^3 x-10 y^2 z x^3+z^5 x+z x^5, \\
  u_{3, 4}^{(6)} &= 10 z^4 x y-10 z^2 x^3 y-10 z^2 x y^3+x^5 y+x y^5.
  \end{array}
\] 
\endgroup

The five harmonic polynomials $u_{1,j}^\supind{6}$ for $j \in \{0,\ldots, 4\}$ are illustrated in Figure~\ref{fig:Sextics}, presented in the same way as previously for quartics in Figure \ref{fig:quartics}. The remaining basis elements $u_{2,j}^\supind{6}$ and $u_{3,j}^\supind{6}$ not depicted arise from these by permuting the coordinates.

\begin{figure} 
  \includegraphics[width=0.19\textwidth]{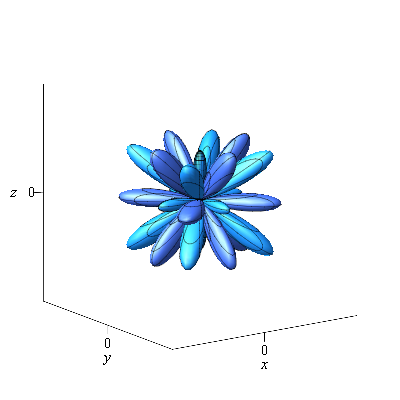}
  \includegraphics[width=0.19\textwidth]{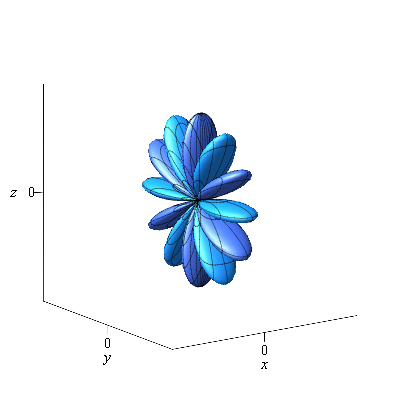}
  \includegraphics[width=0.19\textwidth]{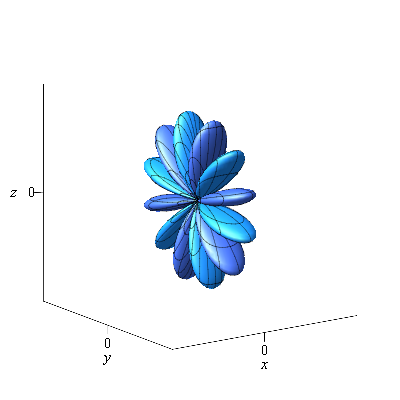}
  \includegraphics[width=0.19\textwidth]{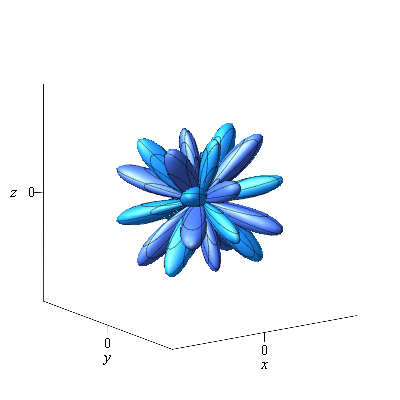}
  \includegraphics[width=0.19\textwidth]{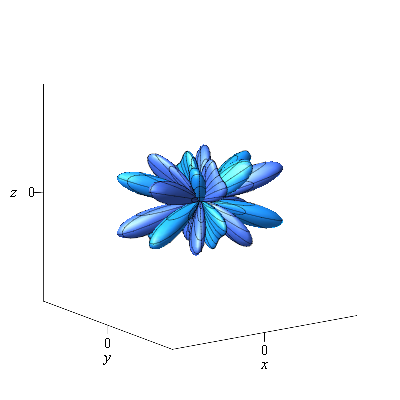}
  \caption{Illustration of the elements $u_{1,j}^\supind{6}$ in the 
  $\B$-symmetric basis of $\H_6$} \label{fig:Sextics}
\end{figure}

For degree~8, the following polynomials form the $\GB$-equivariant spanning set for $\H_8$:
\begingroup \scriptsize
\[ \begin{array}{ll}
u_{1,0}^{(8)} &= -14 x^2 y^6+210 z^2 x^2 y^4-210 y^2 z^4 x^2+14 z^6 x^2+y^8-14 y^6 z^2+14 y^2 z^6-z^8, \\
u_{2,0}^{(8)} &= -14 y^2 z^6+210 y^2 z^4 x^2-210 z^2 x^4 y^2+14 x^6 y^2+z^8-14 z^6 x^2+14 z^2 x^6-x^8, \\
u_{3,0}^{(8)} &= -14 z^2 x^6+210 z^2 x^4 y^2-210 z^2 x^2 y^4+14 y^6 z^2+x^8-14 x^6 y^2+14 x^2 y^6-y^8, \\
\end{array}\]
\[ \begin{array}{ll}
u_{1,1}^{(8)} &= y^8-28 y^6 z^2+70 y^4 z^4-28 y^2 z^6+z^8, \\
u_{2,1}^{(8)} &= z^8-28 z^6 x^2+70 z^4 x^4-28 z^2 x^6+x^8, \\
u_{3,1}^{(8)} &= x^8-28 x^6 y^2+70 x^4 y^4-28 x^2 y^6+y^8, \\
\end{array} \qquad \begin{array}{ll}
u_{1,2}^{(8)} &= -y^7 z+7 y^5 z^3-7 y^3 z^5+y z^7, \\
u_{2,2}^{(8)} &= -z^7 x+7 z^5 x^3-7 z^3 x^5+z x^7, \\
u_{3,2}^{(8)} &= -x^7 y+7 x^5 y^3-7 x^3 y^5+x y^7, \\
\end{array}\]
\[ \begin{array}{ll}
u_{1,3}^{(8)} &= 42 x^2 y^5 z-140 x^2 y^3 z^3+42 x^2 y z^5-3 y^7 z+7 y^5 z^3+7 y^3 z^5-3 y z^7, \\
u_{2,3}^{(8)} &= 42 y^2 z^5 x-140 y^2 z^3 x^3+42 y^2 z x^5-3 z^7 x+7 z^5 x^3+7 z^3 x^5-3 z x^7, \\
u_{3,3}^{(8)} &= 42 z^2 x^5 y-140 z^2 x^3 y^3+42 z^2 x y^5-3 x^7 y+7 x^5 y^3+7 x^3 y^5-3 x y^7, \\
\end{array}\]
\[ \begin{array}{ll}
u_{1,4}^{(8)} &= -35 x^4 y^3 z+35 x^4 y z^3+21 x^2 y^5 z-21 x^2 y z^5-y^7 z+y z^7, \\
u_{2,4}^{(8)} &= -35 y^4 z^3 x+35 y^4 z x^3+21 y^2 z^5 x-21 y^2 z x^5-z^7 x+z x^7, \\
u_{3,4}^{(8)} &= -35 z^4 x^3 y+35 z^4 x y^3+21 z^2 x^5 y-21 z^2 x y^5-x^7 y+x y^7, \\
\end{array}\]
\[ \begin{array}{ll}
u_{1,5}^{(8)} &= 14 x^6 y z-35 x^4 y^3 z-35 x^4 y z^3+21 x^2 y^5 z+21 x^2 y z^5-y^7 z-y z^7, \\
u_{2,5}^{(8)} &= 14 y^6 z x-35 y^4 z^3 x-35 y^4 z x^3+21 y^2 z^5 x+21 y^2 z x^5-z^7 x-z x^7, \\
u_{3,5}^{(8)} &= 14 z^6 x y-35 z^4 x^3 y-35 z^4 x y^3+21 z^2 x^5 y+21 z^2 x y^5-x^7 y-x y^7.
\end{array}\]
\endgroup 
They satisfy the linear relation
  \[u_{1,0}^\supind{8} + u_{2,0}^\supind{8} + u_{3,0}^\supind{8} = 0\]
and are otherwise linearly independent. This spanning set is illustrated in Figure~\ref{fig:Octics}.

\begin{figure}
  \hspace{-0.05\textwidth}\includegraphics[width=0.2\textwidth]{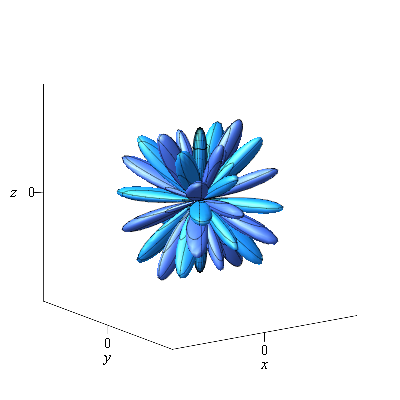}
  \hspace{-0.04\textwidth}\includegraphics[width=0.2\textwidth]{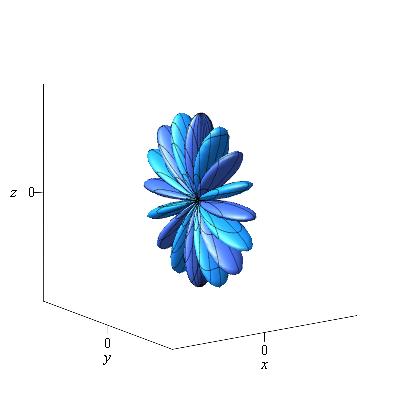}
  \hspace{-0.04\textwidth}\includegraphics[width=0.2\textwidth]{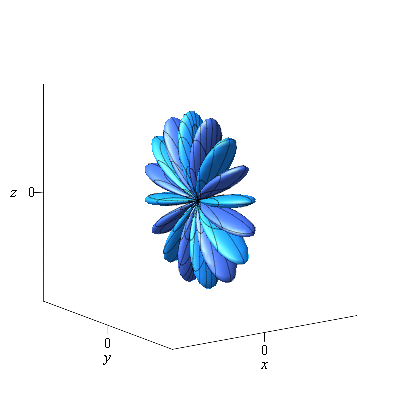}
  \hspace{-0.04\textwidth}\includegraphics[width=0.2\textwidth]{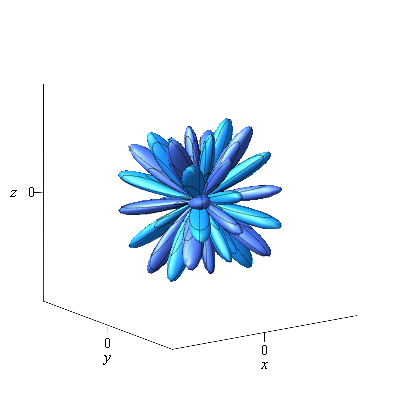}
  \hspace{-0.04\textwidth}\includegraphics[width=0.2\textwidth]{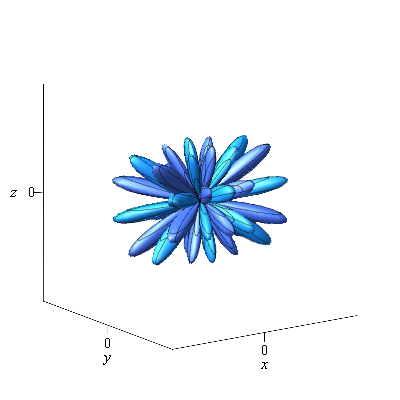}
  \hspace{-0.04\textwidth}\includegraphics[width=0.2\textwidth]{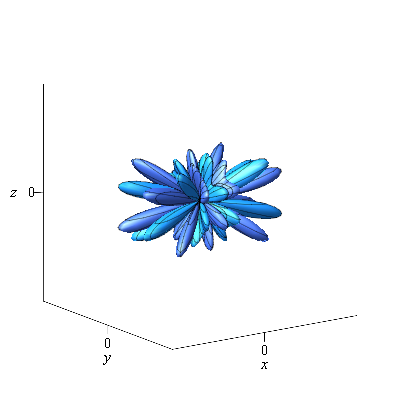}
  \caption{Illustration of the elements $u_{1,j}^\supind{8}$ in the 
  $\B$-symmetric spanning set of $\H_8$}  \label{fig:Octics}
\end{figure}

Generating these expressions is straightforward, as the definition of $u_{i,j}^\supind{2d}$ in the proof of Theorem~\ref{thm:HarmonicBasisHigher} and the formulas for $\Vg{\ell}$ and $\Vf{\ell}$ in Remark~\ref{rem:ExplicitFormulae} give explicit formulas for the coefficients in $u_{i,j}^\supind{2d}$. A straightforward implementation in \texttt{Maple}\footnote{Maple 2016. Maplesoft, a division of Waterloo Maple Inc., Waterloo, Ontario.} produces the $\B$-equivariant basis of $\H_{100}$ in a few seconds.

\subsubsection*{Representation-theoretic viewpoint and cubic harmonics}

While \emph{spherical harmonics} form the most commonly used basis of harmonic polynomials, another basis is given by the \emph{cubic harmonics} \cite{Muggli72,Fox77}, which correspond to the decomposition of the space of harmonic polynomials into irreducible $\B$-representation spaces. 
With the following representation theoretic viewpoint on Theorem~\ref{thm:HarmonicBasisHigher}, one can see how
our explicitly constructed $\B$-equivariant bases are related to these cubic harmonics. 
Indeed, the construction of the spanning set $u_{i,j}^\supind{2d}$ of $\H_{2d}$ 
can be seen as decomposing $\H_{2d}$ into $\B$-invariant subspaces endowed with 
one out of six possible representations of dimension $3$, $2$ or $1$.

For $\zeta, \xi \in \{0,1\}$, there is a three-dimensional representation space $W_{\zeta, \xi}$ of $\B$ given by
  \[\B \to \GL(\R^3) \subset \R^{3 \times 3}, \qquad g = g_\tau g_\sigma \mapsto (\sgn \sigma)^\zeta (\det g_\tau)^\xi g_\tau^\xi g_\sigma,\]
where we (uniquely) write $g \in \B$ as a product of a sign-change matrix $g_\tau$ and a permutation matrix $g_\sigma$.
$W_{0,1}$ and $W_{1,1}$ are irreducible representations of $\B$, but 
  \[W_{0,0} \cong W_{\text{triv}} \oplus W_{\text{std}}, \qquad W_{1,0} = W_{\text{sgn}} \oplus W_{\text{std}},\]
where 
$W_{\text{triv}}$ is the one-dimensional trivial representation,
$W_{\text{std}}$ the two-dimensional irreducible representation  of the symmetric group $\GS[3]$,
and 
$W_{\text{sgn}}$ the one-dimensional alternating representation of the symmetric group $\GS[3]$.
As representations of $\B$
 $W_{\text{triv}}$, $W_{\text{sgn}}$, $W_{\text{std}}$, $W_{0,1}$, $W_{1,1}$ 
are  exactly the irreducible $\B$-representations such that 
$-\id \in \B$ lies in the kernel of the corresponding group homomorphism $\B\to 
\GL(\R^k)$ \cite{GRW09}. These are thus the five irreducible 
representations of $\B/\langle -\id \rangle \cong \GS[4]$.

For $d \equiv 2 \pmod 3$, Theorem~\ref{thm:HarmonicBasisHigher} then corresponds to a decomposition of %
$\H_{2d}$ into the four types of representations $W_{0,0}$, $W_{1,0}$, $W_{0,1}$ and $W_{1,1}$ as
  \[\H_{2d} \cong \bigoplus_{j = 0}^{k-1} W_{\zeta(j), \xi(j)},\]
and we have given an explicit construction of a basis $u_{i,j}$ of $\H_{2d}$ corresponding to this decomposition. 
Indeed, for each $j \in \{0, \ldots, k-1\}$, the three-dimensional subspace of $\H_{2d}$ spanned by $u_{1,j}, u_{2,j}, u_{3,j}$ corresponds to the representation $W_{\zeta(j), \xi(j)}$. 

If $d \not\equiv 2 \pmod 3$, then only for $j \geq 1$ is it true that the subspace spanned by $u_{1, j}, u_{2, j}, u_{3, j}$ is the three-dimensional representation $W_{\zeta(j), \xi(j)}$. 
For $j = 0$, we get the one-dimensional trivial representation $W_{\text{triv}}$ if $d \equiv 0 \pmod 3$, or the two-dimensional irreducible 
representation $W_{\text{std}}$ of the symmetric group if $d \equiv 1 \pmod 3$. Precise counting gives the following decompositions of $\H_{2d}$ into the mentioned $\B$-representations~:
\[\H_{2d} \cong \begin{cases}
                  \left\lceil\frac{d+1}{6}\right\rceil \, W_{0,0} \,\oplus\, \left\lfloor\frac{d+1}{6}\right\rfloor\, W_{1,0} \,\oplus\, \left\lceil\frac{d}{2}\right\rceil \, W_{0,1} \,\oplus\, \left\lfloor\frac{d}{2}\right\rfloor\,  W_{1,1}
                  &\text{if }d \equiv 2 \pmod 3, \vspace{0.1cm}\\
                  W_{\text{triv}} \,\oplus\, \left\lfloor\frac{d}{6}\right\rfloor \, W_{0,0} \,\oplus\, \left\lceil\frac{d}{6}\right\rceil \, W_{1,0} \,\oplus\, \left\lceil\frac{d}{2}\right\rceil \, W_{0,1} \,\oplus\, \left\lfloor\frac{d}{2}\right\rfloor \, W_{1,1}
                  &\text{if }d \equiv 1 \pmod 3, \vspace{0.1cm}\\
                  W_{\text{std}} \,\oplus\, \left\lceil\frac{d-1}{6}\right\rceil \, W_{0,0} \,\oplus\, \left\lfloor\frac{d-1}{6}\right\rfloor \, W_{1,0} \,\oplus\, \left\lfloor\frac{d}{2}\right\rfloor \, W_{0,1}\,\oplus\, \left\lceil\frac{d}{2}\right\rceil \, W_{1,1}
                  &\text{if }d \equiv 0 \pmod 3.
                \end{cases}\]

  \subsection{\texorpdfstring{$\B$}{B\_3}-equivariant bases for \texorpdfstring{$\Lambda_{2d}$}{the slice}}

In order to describe rational $\B$-invariants on $\Lambda_{2d}$, we now describe a basis for $\Lambda_{2d} = \H_{2d} \oplus \VQ \H_{2d-2} \oplus \dots \oplus \VQ^{d-1}\Lambda_2$ with $\B$-symmetries. The main observation is that  we can combine the linear spanning sets of the spaces $\H_{2d+2}$ and $\H_{2d}$, when $d \equiv 0 \pmod 3$, from Theorem~\ref{thm:HarmonicBasisHigher}, so as to remove the linear dependencies. 
Specifically, we observe:
\begin{lemma} \label{lem:combineTwoHarmonicSpaces}
  Let $m \geq 1$ and $k := 8m+2$. The $3k$-dimensional subspace $\H_{6m+2} \oplus \VQ \H_{6m} \subset V_{6m+2}$ has a $\B$-equivariant basis $\tilde u_{i,j}$ ($1 \leq i \leq 3, 0 \leq j < k$) with respect to some maps $\zeta, \xi \colon \{0,\ldots, k-1\} \to \{0,1\}$.
\end{lemma}

\begin{proof}
  A $\B$-equivariant basis is given by
    \[\{\VQ u_{i,0}^\supind{6m} + u_{i,0}^\supind{6m+2} \mid 1 \leq i \leq 3\} \; \cup \; \{\VQ u_{i,j}^\supind{6m} \mid i,j \geq 1\} \; \cup \; \{u_{i,j}^\supind{6m+2} \mid i,j \geq 1\}\]
  where $\{u_{i,j}^\supind{6m}\}$ and $\{u_{i,j}^\supind{6m+2}\}$ are respectively the spanning sets of $\H_{6m}$ and $\H_{6m+2}$ constructed in Theorem~\ref{thm:HarmonicBasisHigher} and the maps $\zeta$ and $\xi$  are defined accordingly.
\end{proof}

\begin{thm} \label{thm:SliceBasisHigher}
  Let $d \geq 1$ and let $k := \left\lfloor\frac{(d+1)(2d+1)}{3}\right\rfloor-1$. Then the vector space $\Lambda_{2d}$ contains $3k$ linearly independent elements $w_{i,j} \in \Lambda_{2d}$ for $1 \leq i \leq 3, 0 \leq j < k$ forming a $\B$-equivariant set with respect to some maps $\zeta, \xi \colon \{0,\ldots, k-1\} \to \{0,1\}$. \\
  If $d$ is not divisible by~3, then these $w_{i,j}$ form a basis of $\Lambda_{2d}$. If $d \equiv 0 \pmod 3$, the set $\{w_{i,j}\}$ can be extended to a basis of $\H_{2d}$ by adding an element $w_\infty \in \Lambda_{2d}$ which is fixed under the action of $\B$.
\end{thm}

\begin{proof}
  We recall that $\Lambda_{2d} = \H_{2d} \oplus \VQ \H_{2d-2} \oplus \dots \oplus \VQ^{d-2} \H_4 \oplus \VQ^{d-1}\Lambda_2$, where $\VQ^{d-1} \Lambda_2 \subset V_{2d}$ is a three-dimensional subspace with basis $w_{1,0} := \VQ^{d-1} x^2, w_{2,0} := \VQ^{d-1} y^2, w_{3,0} := \VQ^{d-1} z^2$. For $i \in \{1,2,3\}$ we observe that $g w_{i,0} = w_{\sigma(i),0}$ if $g \in \B$ is a $3 \times 3$ permutation matrix with corresponding permutation $\sigma$ and $g w_{i,0} = w_{i,0}$ if $g \in \B$ is a $3 \times 3$ sign-change matrix.
  
  From Theorem~\ref{thm:HarmonicBasisHigher} and Lemma~\ref{lem:combineTwoHarmonicSpaces} we obtain bases $\{u_{i,j}^\supind{6m-2}\} \subset \H_{6m-2}$ and $\{\tilde u_{i,j}^\supind{6m+2}\} \subset \H_{6m+2} \oplus \VQ \H_{6m}$ with the desired $\B$-equivariance property. By multiplying with appropriate powers of $\VQ$, we obtain the desired elements $w_{i,j} \in \Lambda_{2d}$. Explicitly,
    \[\{w_{i,j}\} = \left\{\VQ^{d-3m-1} \tilde u_{i,j}^\supind{6m+2} \mid 1 \leq m \leq \frac{d-2}{6}\right\} \cup \left\{\VQ^{d-3m+1} u_{i,j}^\supind{6m-2} \mid 1 \leq m \leq \frac{d+1}{3}\right\} \cup \left\{u_{i,j}^\supind{6m} \mid m = \frac{d}{3}, j \geq 1\right\},\]
  additionally to the elements $w_{i,0}$ from above. Note that the last set in this union is empty if $d \not\equiv 0 \pmod 3$, and otherwise consists of those elements $u_{i,j}^\supind{2d} \in \H_{2d}$ that are linearly independent.
  
  From Theorem~\ref{thm:HarmonicBasisHigher} and Lemma~\ref{lem:combineTwoHarmonicSpaces} it is clear that this set $\{w_{i,j}\}$ is $\B$-equivariant (adequately inheriting the definition of the maps $\zeta, \xi$ from Theorem~\ref{thm:HarmonicBasisHigher} and Lemma~\ref{lem:combineTwoHarmonicSpaces}). In the case $d \equiv 0 \pmod 3$, we do not yet have a basis of $\Lambda_{2d}$, because we left out the element
    \[w_\infty := u_{1,0}^\supind{2d} = u_{2,0}^\supind{2d} = u_{3,0}^\supind{2d}\]
  from the linear spanning set $u_{i,j}^\supind{2d}$ of the subspace $\H_{2d} \subset \Lambda_{2d}$ described in Theorem~\ref{thm:HarmonicBasisHigher}. Note that the relation $u_{1,0}^\supind{2d} = u_{2,0}^\supind{2d} = u_{3,0}^\supind{2d}$ and the description of the $\B$-action on $u_{1,0}^\supind{2d}, u_{2,0}^\supind{2d}, u_{3,0}^\supind{2d}$ in Theorem~\ref{thm:HarmonicBasisHigher} imply that $g w_\infty = w_\infty$ for all $g \in \B$. This concludes the proof.
\end{proof}

Again it is important for applications to not consider Theorem~\ref{thm:SliceBasisHigher} a pure existence result, but to observe how to immediately obtain the basis 
$\{w_{i,j}\}$ (resp.\ $\{w_{i,j}, w_\infty\}$) of $\Lambda_{2d}$ from harmonic polynomials $u_{i,j}^\supind{2k} \in \H_{2k}$ as in Theorem~\ref{thm:HarmonicBasisHigher}. In particular, one can easily write out such a basis. Slightly deviating from Definition~\ref{def:B3Equivariance}, we will call the constructed basis \emph{$\B$-equivariant}, even in the case $d \equiv 0 \pmod 3$, when the basis includes an element $w_\infty$ fixed under $\B$.

\begin{remark} \label{rem:ZetaXiZero}
  We observe that for $d \geq 2$ there exists an index $j_0 \in \{0, \ldots, k-1\}$ such that $\zeta(j_0) = 0$ and $\xi(j_0) = 1$ (and may assume after permuting the basis elements that this holds for $j=0$). Indeed, the basis contains -- up to multiplication with a power of $\VQ$ -- the elements $u_{i,j}^\supind{4}$ forming a basis of $\H_4$.
  This basis of $\H_4$, given in Proposition~\ref{prop:HarmonicBasisQuartic}, satisfies the desired property, as remarked after Theorem~\ref{thm:HarmonicBasisHigher}.
\end{remark}

  \subsection{Rational invariants for ternary forms of arbitrary even degree} \label{ssec:HigherInvariants}

With the construction of the $\B$-equivariant basis for the slice $\Lambda_{2d} \subset \VV_{2d}$ established for arbitrary $d$, we can now turn to the construction of a generating set of rational invariants for the action of $\B$ on $\Lambda_{2d}$, generalizing the case $2d=4$ described in Section~\ref{sec:Quartics}.

Let $d \geq 2$ and consider a basis $\{w_{i,j} \mid 1 \leq i \leq 3, 0 \leq j < k\} \; [\cup \; \{w_\infty\}]$ of $\Lambda_{2d}$ and corresponding maps $\zeta, \xi \colon \{0, \ldots, k-1\} \to \{0,1\}$ as described in Theorem~\ref{thm:SliceBasisHigher}. Here, and from now on, we  distinguish between the case $d \not\equiv 0 \pmod 3$ and the case  $d \equiv 0 \pmod 3$, where there is an additional basis element $w_\infty \in \Lambda_{2d}$, with square brackets, whenever we are confident that no confusion arises from this.

As observed in Remark~\ref{rem:ZetaXiZero}, we can assume $\zeta(0) = 0, \xi(0)=1$. With this chosen basis, an element $v \in \Lambda_{2d}$ can be uniquely expressed as
\begin{equation} \label{eq:coordInSlice}
  v = \sum_{i=1}^3 \Ct[i] w_{i,0} + \sum_{i=1}^3 \sum_{j=1}^{k-1} \alpha_{i,j} w_{i,j} \; [+ \; \alpha_\infty w_\infty]
\end{equation}
for $\Ct[i], \alpha_{i,j} \; [, \alpha_\infty] \in \R$. In this way, we identify $\R(\Lambda_{2d})$ with the field of rational expressions in variables $\Ct[i], \alpha_{i,j} \:[, \alpha_\infty]$, i.e.\ 
  \[\R(\Lambda_{2d}) = \R(\Ct[i], \alpha_{i,j} \; [, \alpha_\infty] \; \mid 1 \leq i \leq 3, 0 \leq j < k).\]

\begin{thm} \label{thm:InvariantsHigherDeg}
  With the notations as above, a minimal generating set of rational invariants 
  for $\R(\Lambda_{2d})^{\B}$ is given by the $2d^2+3d-2 = \dim \Lambda_{2d}$ 
  polynomial functions $\setP = \{p_{i,j}\} \; [\cup \{p_\infty\}]$, whose 
  values at $v$ as in \eqref{eq:coordInSlice} are given as follows:
  \begin{itemize}
    \item Three invariants are given by
          \begin{align*}
            p_{1,0} &:= \Ct[1]^2+\Ct[2]^2+\Ct[3]^2, \\
            p_{2,0} &:= \Ct[1] \Ct[2]\Ct[3], \\
            p_{3,0} &:= \Ct[1]^4+\Ct[2]^4+\Ct[3]^4.
          \end{align*}
    \item In the case $d \equiv 0 \pmod 3$, one invariant is given by
            \[p_\infty := \alpha_\infty.\]
    \item The remaining invariants in the generating set are given as the entries of the $3\times (k-1)$-matrix 
            \[\begin{pmatrix}
            p_{1,1} & \dots & p_{1,k-1} \\
            p_{2,1} & \dots & p_{2,k-1} \\
            p_{3,1} & \dots & p_{3,k-1}
            \end{pmatrix}
            := \begin{pmatrix}
            1 & 1 & 1 \\
            \Ct[1]^2 & \Ct[2]^2 & \Ct[3]^2 \\
            \Ct[1]^4 & \Ct[2]^4 & \Ct[3]^4
            \end{pmatrix}  M(\alpha),\]
            where $k=\lfloor \frac{2d^2+3d-5}{3}\rfloor $ 
            and $M(\alpha)$ is the $3\times (k-1)$-matrix whose $(i,j)$-th entry is
            \[M(\alpha)_{i,j} = \Ct[i]^{\xi(j)} \left((\Ct[1]^2-\Ct[2]^2)(\Ct[2]^2-\Ct[3]^2)(\Ct[3]^2-\Ct[1]^2)\right)^{\zeta(j)}\; \alpha_{i,j}.\]
  \end{itemize}
\end{thm}

\begin{proof}
  With Theorem~\ref{thm:SliceBasisHigher} established, the proof is analogous to the one for Theorem~\ref{thm:TernaryQuarticInvariants}: One checks that $\setP$ consists of $\B$-invariants (for example, by the use of equivariant maps as before, or directly from the definitions). In order to express a given rational invariant $\hat p \in \R(\Lambda_{2d})^{\B}$ as a rational combination in terms of $\setP$, we may first replace each occurrence of the variables $\alpha_{i,j}$ in $\hat p$ by a rational expression in terms of invariants in $\setP$ and $\Ct[1], \Ct[2], \Ct[3]$ by using that $M(\alpha)_{i,j}$ is the $(i,j)$-th entry of the matrix product
    \[\begin{pmatrix}
    1 & 1 & 1 \\
    \Ct[1]^2 & \Ct[2]^2 & \Ct[3]^2 \\
    \Ct[1]^4 & \Ct[2]^4 & \Ct[3]^4
    \end{pmatrix}^{-1} 
    \begin{pmatrix}
    p_{1,1} & \dots & p_{1,k-1} \\
    p_{2,1} & \dots & p_{2,k-1} \\
    p_{3,1} & \dots & p_{3,k-1}
    \end{pmatrix}.\]
  If $d \equiv 0 \pmod 3$, we also replace each occurrence of $\alpha_\infty$ by $p_\infty$. As in the proof of Theorem~\ref{thm:TernaryQuarticInvariants}, we express the remaining $\B$-invariant rational expression in $\Ct[1], \Ct[2], \Ct[3]$ in terms of $p_{1,0}, p_{2,0}, p_{3,0}$ -- note that $\Ct[1], \Ct[2], \Ct[3]$ play the exact same role here as in Theorem~\ref{thm:TernaryQuarticInvariants}.
\end{proof}

Explicitly, as in Section~\ref{ssec:QuarticInvariants}, we again have a routine for expressing rational $\B$-invariants on $\Lambda_{2d}$ in terms of $\setP$, given by the following rewrite rules:
\[\begin{pmatrix} \alpha_{1,j} \\ \alpha_{2,j} \\ \alpha_{3,j} \end{pmatrix}
      \longrightarrow 
           \frac{\delta^{1-\zeta(j)}}{\It[0] p_{2,0}^{\xi(j)}}
           \begin{pmatrix}
            \Ct[2]\Ct[3] & 0 & 0 \\ 0& \Ct[3]\Ct[1] & 0  \\  0 & 0 &\Ct[1]\Ct[2]  
           \end{pmatrix}^{\xi(j)}
           \begin{pmatrix} 
						(\Ct[3]^2-\Ct[2]^2)\Ct[2]^2\Ct[3]^2 & \Ct[2]^4-\Ct[3]^4 & \Ct[3]^2-\Ct[2]^2\\
						(\Ct[1]^2-\Ct[3]^2)\Ct[3]^2\Ct[1]^2 & \Ct[3]^4-\Ct[1]^4 & \Ct[1]^2-\Ct[3]^2\\
						(\Ct[2]^2-\Ct[1]^2)\Ct[1]^2\Ct[2]^2 & \Ct[1]^4-\Ct[2]^4 & \Ct[2]^2-\Ct[1]^2\\
					 \end{pmatrix}
					 \begin{pmatrix} p_{1,j} \\ p_{2,j} \\ p_{3,j}\end{pmatrix} \qquad \forall j \geq 1,
\]
  \[\alpha_\infty \longrightarrow p_\infty, \qquad \Ct[1] \longrightarrow  \frac{1}{p_{3,0}}(p_{1,0}\Ct[2]\Ct[3] -\Ct[2]^3\Ct[3]-\Ct[2]\Ct[3]^3),\]
  \[\Ct[2]^4 \longrightarrow p_{1,0}\Ct[2]^2-\Ct[2]^2\Ct[3]^2+p_{1,0}\Ct[3]^2 -\Ct[3]^4+\frac{1}{2}(p_{2,0}-p_{1,0}^2), \qquad \Ct[3]^6 \longrightarrow  \frac{1}{2}(p_{2,0}-p_{1,0}^2)\Ct[3]^2+p_{1,0}\Ct[3]^4 +p_{3,0}^2,\] 
where
  \[\delta := (\Ct[1]^2-\Ct[2]^2)(\Ct[2]^2-\Ct[3]^2)(\Ct[3]^2-\Ct[1]^2) \qquad \text{and}\]
  \[\It[0] := \delta^2 = \frac{1}{2}p_{2,0}^3-\frac{1}{4}p_{1,0}^6-27p_{3,0}^4+p_{1,0}^4p_{2,0} - \frac{5}{4}p_{1,0}^2p_{2,0}^2 - 9p_{3,0}^2p_{1,0}p_{2,0} + 5p_{3,0}^2p_{1,0}^3.\]

As in Section~\ref{ssec:QuarticInvariants}, we notice that the generating set 
$\setP \subset \R(\Lambda_{2d})^{\B}$ is a transcendence basis for 
$\R(\Lambda_{2d})^{\B}$ as a field extension over $\R$, since $|\setP| = \dim 
\Lambda_{2d}$. In particular, this generating set is of minimal cardinality. By 
Corollary~\ref{cor:SliceLemmaGen}, there uniquely exists a corresponding 
generating set $\setPT$ of rational $\G$-invariants on $\VV_{2d}$ with 
$|\setPT| = |\setP| = \dim \Lambda_{2d}$, attaining the lower bound from 
Theorem~\ref{thm:NumberInvariants} because of $\dim \VV_{2d} = \Lambda_{2d} + 
\dim \G$.

\begin{cor} \label{cor:HigherFullInvariants}
  There exists a generating set of $\dim \VV_{2d} - \dim \G = 2d^2+3d-2$ rational invariants $\setPT \subset \R(\VV_{2d})^{\G}$ uniquely determined by their restrictions to  $\Lambda_{2d} \subset \VV_{2d}$ that are given by the invariants $\setP \subset \R[\Lambda_{2d}]^{\B}$ specified in Theorem~\ref{thm:InvariantsHigherDeg}.
\end{cor}

\section{Solving the main algorithmic challenges} \label{sec:AlgorithmicSolutions}

In Sections~\ref{sec:Quartics} and~\ref{sec:B3Harmonics}, in Corollaries~\ref{cor:QuarticFullInvariants} and \ref{cor:HigherFullInvariants},
we identified, for any $d \geq 2$, a finite set of rational invariants $\setPT \subset \R(\VV_{2d})^{\G}$ generating $\R(\VV_{2d})^{\G}$. 
We denote the number of these rational invariants as $m := |\setPT|$. 
Recall that $m$ is minimal, that is 
  \[m = \dim \VV_{2d} - \dim \G = \binom{2d+2}{2} - 3.\]

Instead of being given by a closed form formula,  each invariant $\Vpt \in 
\setPT$ is uniquely determined, in virtue of Theorem~\ref{thm:SliceLemma}, by 
the restriction $\Vp = \restr{\Vpt}{\Lambda_{2d}} \in \R(\Lambda_{2d})^{\B}$ of 
the rational map $\Vpt \colon \VV_{2d} \dashrightarrow \R$ to the subspace 
$\Lambda_{2d} \subset \VV_{2d}$. 
We shall in this section examine the practical implications of Theorem~\ref{thm:SliceLemma} in addressing  the algorithmic challenges  formulated in Section~\ref{ssec:AlgoFormulations}.
For these, we provide algorithms that rely only on the explicit knowledge of 
the restricted invariants $\setP = \{\restr{\Vpt}{\Lambda_{2d}} \mid \Vpt \in 
\setPT\}$.

\subsection{The Evaluation Problem} \label{ssec:EvaluationProblem}
 
Having identified a finite generating set of rational invariants $\setPT \subset \R(\VV_{2d})^{\G}$, the most basic algorithmic question is: \emph{How can we evaluate each of the generating invariants $\Vpt \in \setPT$ at a given point $v \in \VV_{2d}$?}

Let $\Vpt \in \setPT$ and consider the restriction $\Vp := \restr{\Vpt}{\Lambda_{2d}} \in \setP$ of $\Vpt \colon \VV_{2d} \dashrightarrow \R$ to the subspace $\Lambda_{2d} \subset \VV_{2d}$. By construction of the invariants, we know explicit expressions for $\Vp \in \R(\Lambda_{2d})^{\B}$, which allow us to evaluate $\Vpt \in \setPT$ at any point of the slice $u \in \Lambda_{2d}$ by computing $\Vpt(u) = \Vp(u) \in \R$.

When we want to evaluate $\Vpt(v)$ for arbitrary $v \in \VV_{2d}$, we observe that $\Vpt(gv) = \Vpt(v)$ for all $g \in \G$, since $\Vpt$ is an invariant for the action of $\G$. By Proposition~\ref{prop:MainSlice}, we know that for general $v \in \VV_{2d}$ there exists an orthogonal transformation $g \in \G$ such that $gv \in \Lambda_{2d}$. For such $g \in \G$, we may then compute
  \[\Vpt(v) = \Vpt(gv) = \Vp(gv).\]

Recalling the definition of the slice $\Lambda_{2d}$, this idea leads to 
Algorithm~\ref{algo:Evaluation}. In the following, we comment on its validity 
and its computational realization. 

\begin{algorithm}[h]
	  \DontPrintSemicolon
		\SetKwInOut{Input}{Input}\SetKwInOut{Output}{Output}
		\Input{$v \in \VV_{2d} = \R[\xyz]_{2d}$, $d \geq 2$}
		\Output{$\Vpt(v) \in \R$ for each $\Vpt \in \setPT$} %
		\BlankLine
		Compute $v' \in \VV_2$ in the Harmonic Decomposition
		  \[v = h_{2d} + \VQ h_{2d-2} + \VQ^2 h_{2d-4} + \ldots + \VQ^{d-2} h_4 + \VQ^{d-1} v', \; \hbox{ where } h_{2k} \in \H_{2k}.\]
		 \label{algoStep:HarmDecomp}\;
		Let $A \in \R^{3 \times 3}$ be the symmetric matrix corresponding to the quadratic form $v' \in \VV_2$ as in \eqref{eq:GramianMatrix}. \label{algoStep:Gramian}\;
		Determine $g \in \G$ such that the matrix product $g A g^T$ is diagonal:
		  \[g A g^T = \diag(\lambda_1, \lambda_2, \lambda_3) \quad \text{ for some } \lambda_1, \lambda_2, \lambda_3 \in \R.\] \label{algoStep:Diag} \vspace{-1em}\;
		If $\lambda_i = \lambda_j$ for some $i \neq j$, output \emph{``undefined at $v$''} and stop. \label{algoStep:Undef}\;
		Compute $u := gv \in \Lambda_{2d}$. \label{algoStep:ApplyG}\;
		Compute and output the values $\Vp(u)$ for each $\Vp \in \setP$. \label{algoStep:EvRestr}\;
		\caption{Evaluation Algorithm} \label{algo:Evaluation}
\end{algorithm}

\subsubsection*{Validity of Algorithm~\ref{algo:Evaluation}}

First, we describe how the formulation in Algorithm~\ref{algo:Evaluation} corresponds to the idea described above of evaluating $\Vpt(v)$ as $\Vp(gv)$ where $g\in \G$ is some orthogonal transformation such that $gv \in \Lambda_{2d}$. We recall that if
  \[v = h_{2d} + \VQ h_{2d-2} + \VQ^2 h_{2d-4} + \ldots + \VQ^{d-2} h_4 + \VQ^{d-1} v'\]
is the Harmonic Decomposition of $v$ computed in Step~\ref{algoStep:HarmDecomp}, then the Harmonic Decomposition of $gv$ is given as
  \[gv = gh_{2d} + \VQ(gh_{2d-2}) + \VQ^2(gh_{2d-4}) + \ldots + \VQ^{d-2} (gh_4) + \VQ^{d-1}(gv')\]
by Proposition~\ref{prop:InvariantSubspaces}. Then the definition of $\Lambda_{2d}$ gives: $gv$ is contained in $\Lambda_{2d}$ if and only if $gv' \in \VV_2$ lies in $\Lambda_2$, i.e.\ the symmetric matrix associated to the quadratic form $gv'$ is diagonal. This matrix is $gAg^T$ (where $A$ is the matrix of the quadratic form $v' \in \VV_2$).

To validate Step~\ref{algoStep:Undef} in Algorithm~\ref{algo:Evaluation}, we recall that even when the restricted invariant $\Vp \in \R(\Lambda_{2d})^{\B}$ is known to be a polynomial invariant, the corresponding invariant $\Vpt \in \R(\VV_{2d})^{\G}$ is typically only a rational function, defined at a general point only. In fact, going back to the proof of Proposition~\ref{prop:MainSlice}, we see that the rational function $\Vpt \colon \VV_{2d} \dashrightarrow \R$ which we obtain from $\Vp \in \R(\Lambda_{2d})^{\B}$ is only defined at those points $v \in \VV_{2d}$ whose quadratic part $v' \in \VV_2$ (in the Harmonic Decomposition) does not have repeated eigenvalues in the matrix representation \eqref{eq:GramianMatrix}. This precisely corresponds to Step~\ref{algoStep:Undef}. We observe that leaving out Step~\ref{algoStep:Undef}, Algorithm~\ref{algo:Evaluation} would still output a (meaningless) value for the non-defined cases, but that value would depend on the choice of $g$ in Step~\ref{algoStep:Diag}.

\subsubsection*{Computational realization of Algorithm~\ref{algo:Evaluation}} First, we discuss the implementation of Step~\ref{algoStep:HarmDecomp}:
In order to compute $v' \in \VV_2$ in the Harmonic Decomposition of $v \in \VV_{2d}$, 
we can use the explicit projection operators on $\H_{2d},\ldots,\H_{4}$ given in \cite[Theorem 5.21]{ABR01}.
Alternatively we express the given element $v = \sum_{ijk} a_{ijk} \xx^i \yy^j \zz^k$ in terms of a basis of $\VV_{2d}$ that reflects the Harmonic Decomposition
  \[\VV_{2d} = \H_{2d} \oplus \ldots \oplus \VQ^{d-2}\H_{4} \oplus \VQ^{d-1}\VV_2.\]
  
Say, $\hBasis{2k} \subset \H_{2k}$ is a basis of the space $\H_{2k}$ of harmonic polynomials of degree $2k$ (for $2 \leq k \leq d$). Then
  \[\vBasis{2d} := \bigcup_{k=2}^d \VQ^{d-k}\hBasis{2k}\; \cup \; \{\VQ^{d-1}x^2, \VQ^{d-1}y^2, \VQ^{d-1}z^2, \VQ^{d-1}xy, \VQ^{d-1}yz, \VQ^{d-1}zx\}\]
is a basis of $\VV_{2d}$ and we can uniquely write 
  \begingroup \small
  \begin{equation} \label{eq:exprInBasis}
    v = \sum_{k=2}^d \sum_{\hbEl \in \hBasis{2k}} \alpha_{\hbEl} \VQ^{d-k}\hbEl \; + \; \beta_{200} \VQ^{d-1}x^2 + \beta_{020} \VQ^{d-1}y^2 + \beta_{002} \VQ^{d-1}z^2 + \beta_{110} \VQ^{d-1}xy + \beta_{011} \VQ^{d-1}yz + \beta_{101} \VQ^{d-1}zx
  \end{equation}
  \endgroup
with $\alpha_{\hbEl}, \beta_{ijk} \in \R$. Then the Harmonic Decomposition of $v$ as in Step~\ref{algoStep:HarmDecomp} is given by
  \[h_{2k} = \sum_{\hbEl \in \hBasis{2k}} \alpha_{\hbEl} \hbEl \in \H_{2k}, \qquad v' = \beta_{200} x^2 + \beta_{020} y^2 + \beta_{002} z^2 + \beta_{110} xy + \beta_{011} yz + \beta_{101} zx \in \VV_2.\]

For $\hBasis{2k} \subset \H_{2k}$, we might as well use the elements $u_{i,j}^{(2k)}$ explicitly constructed in Theorem~\ref{thm:HarmonicBasisHigher} (by selecting a linearly independent subset of these in the cases $k \not\equiv 2 \pmod 3$), though it is not essential at this stage.
For chosen bases $\hBasis{2k} \subset \H_{2k}$, computing the expression \eqref{eq:exprInBasis} for a given $v = \sum_{ijk} a_{ijk} \xx^i \yy^j \zz^k$ corresponds to converting from the monomial basis $\{x^i y^j z^k \mid i+j+k = 2d\} \subset \VV_{2d}$ to the basis $\vBasis{2d} \subset \VV_{2d}$. The corresponding base-change matrix is the inverse of the matrix whose entries are given by the coefficients of the expressions $\bEl \in \vBasis{2d}$ in the variables $\xyz$. This base-change matrix may be precomputed for fixed $d$. If we use bases $\hBasis{2k} \subset \H_{2k}$ built from elements $u_{i,j}^{(2k)} \in \H_{2k}$ as constructed in Theorem~\ref{thm:HarmonicBasisHigher}, the base change transformation could in fact be described explicitly by a careful analysis of the proofs in Section~\ref{ssec:B3HarmonicsConstruction}.

Step~\ref{algoStep:Diag} is the problem of computing the \emph{eigendecomposition} of the symmetric matrix $A$ and is equivalent to finding the eigenvalues and eigenvectors of $A$. \emph{Exact (symbolic)} algorithmic solutions to this problem would involve introducing (nested) square roots and cubic roots, but such expressions are typically very undesirable from a practical standpoint. However, there are well-established numerical methods for computing the eigendecomposition of a symmetric matrix -- with the additional benefit that the numerical stability of these methods is well-studied \cite{Kre05}, \cite[Chapter 8]{GVL13}. 

Realizing Step~\ref{algoStep:ApplyG} is straightforward by the definition of the action of $g = (g_{ij}) \in \G$ on $\VV_{2d}$: Applying the substitutions
  \[
    \xx \mapsto g_{11} \xx + g_{21} \yy + g_{31} \zz, \qquad
    \yy \mapsto g_{12} \xx + g_{22} \yy + g_{32} \zz, \qquad 
    \zz \mapsto g_{13} \xx + g_{23} \yy + g_{33} \zz,
  \]
to $v \in \VV_{2d} = \R[\xyz]_{2d}$ and expanding the resulting expression gives $u = gv \in \Lambda_{2d}$. This expansion may also be precomputed symbolically for fixed $d$ such that it is only necessary to evaluate with the entries $g_{ij}$.

Finally, for Step~\ref{algoStep:EvRestr}, we want to evaluate at $u \in \Lambda_{2d}$ the expressions for the rational invariants $\setP \subset \Lambda_{2d}$ described in Theorem~\ref{thm:InvariantsHigherDeg}. For that, we express $u = \sum_{i=1}^3 \Ct[i] w_{i,0} + \sum_{i=1}^3 \sum_{j=1}^k \alpha_{i,j} w_{i,j} \; [+ \; \alpha_\infty w_\infty]$ , where $\{w_{i,j} \mid 1 \leq i \leq 3, 0 \leq j \leq k\} \; [\cup \; \{w_\infty\}]$ is the basis of $\Lambda_{2d}$ from Theorem~\ref{thm:SliceBasisHigher}. This can again be realized by an appropriate linear base change transformation. We then evaluate $\Vp(u)$ for each $\Vp \in \setP$ with the formulas for the invariants given in Theorem~\ref{thm:InvariantsHigherDeg}.

  \subsection{The Rewriting Problem}

In this section, we discuss how to address the Rewriting Problem specified in Section~\ref{ssec:AlgoFormulations}. 

For notational simplification, we now assume that the rational invariants in 
$\setPT \in \R(\VV_{2d})^{\G}$ are indexed as $\setPT = \{\Vpt[1], \ldots, 
\Vpt[m]\}$ and their restrictions to $\Lambda_{2d}$ are $\setP = \{\Vp[1], 
\ldots, \Vp[m]\}$. Since $\Vpt[1], \ldots, \Vpt[m] \in \R(\VV_{2d})^{\G}$ form 
a generating set of rational invariants, it is possible to express any other 
rational invariant $\Vpt[0] \in \R(\VV_{2d})^{\G}$ as a rational combination of 
$\Vpt[1], \ldots, \Vpt[m]$, i.e.\ there exists a rational expression in $m$ 
variables 
  $r(T_1, \ldots, T_m) \in \R(T_1, \ldots, T_m)$ such that
  \begin{equation} \label{eq:rewriteExpr}
    \Vpt[0] = r(\Vpt[1], \ldots, \Vpt[m]).
  \end{equation}

Determining such $r(T_1, \ldots, T_m) \in \R(T_1, \ldots, T_m)$ for any given $\Vpt[0] \in \R(\VV_{2d})^{\G}$
is a  problem that can be reduced to the corresponding problem for $\B$-invariants on the subspace $\Lambda_{2d} \subset \VV_{2d}$.

Note that restricting the equality \eqref{eq:rewriteExpr} to the subspace $\Lambda_{2d} \subset V_{2d}$ gives:
  \[\restr{\Vpt[0]}{\Lambda_{2d}} = r(\Vp[1], \ldots, \Vp[m]).\]
Hence, to determine the rational expression $r$, it is sufficient to rewrite $\Vp[0] := \restr{\Vpt[0]}{\Lambda_{2d}} \in \R(\Lambda_{2d})^{\B}$ in terms of the restricted generating rational invariants $\Vp[1], \ldots, \Vp[m] \in \R(\Lambda_{2d})^{\B}$. This leads to Algorithm~\ref{algo:Rewrite}.

\begin{algorithm}[h]
	  \DontPrintSemicolon
		\SetKwInOut{Input}{Input}\SetKwInOut{Output}{Output}
		\Input{$\Vpt[0] \in \R(\VV_{2d})^{\G}$ for some $d \geq 2$} 
		\Output{A rational expression $r \in \R(T_1, \ldots, T_m)$ such that $\Vpt[0] = r(\Vpt[1], \ldots, \Vpt[m])$}
		\BlankLine
    Let $\Vp[0] := \restr{\Vpt[0]}{\Lambda_{2d}} \in \R(\Lambda_{2d})^{\B}$.\;
    Determine $r \in \R(T_1, \ldots, T_m)$ such that $\Vp[0] = r(\Vp[1], \ldots, \Vp[m])$. \label{algoStep:rewriteOnSlice}\;
		Output $r$.\;
		\caption{Rewriting Algorithm} \label{algo:Rewrite}
\end{algorithm}

Step~\ref{algoStep:rewriteOnSlice} is the problem of rewriting a rational $\B$-invariant on $\Lambda_{2d}$ in terms of the generating rational invariants $\Vp[1], \ldots, \Vp[m] \in \R(\Lambda_{2d})^{\B}$. The realization of this is addressed in the proofs of Theorems~\ref{thm:TernaryQuarticInvariants} and \ref{thm:InvariantsHigherDeg} and made explicit by the “rewrite rules” given in Sections~\ref{ssec:QuarticInvariants} and \ref{ssec:HigherInvariants}.

  \subsection{The Reconstruction Problem}

By construction,
  \[\setPT = \{\Vpt[i,j] \mid 1 \leq i \leq 3, 0 \leq j < k\} \ [\cup \{\Vpt[\infty]\}]\]
is the set of rational $\G$-invariants on $\VV_{2d}$ such that 
  \[\restr{\Vpt[i,j]}{\Lambda_{2d}} = \Vp[i,j] \quad \forall i,j \qquad [\text{and } \qquad \restr{\Vpt[\infty]}{\Lambda_{2d}} = \Vp[\infty]]\]
are the $\B$-invariants on $\Lambda_{2d}$ from Theorem~\ref{thm:InvariantsHigherDeg}.

In Section~\ref{ssec:EvaluationProblem}, we saw how to numerically evaluate $\Vpt[i,j](v) \in \R$ [and $\Vpt[\infty](v) \in \R$] at a general point $v \in \VV_{2d}$. Now we consider the inverse algorithmic problem: Given real values $\mu_{i,j} \in \R$ for $1 \leq i \leq 3, 0 \leq j < k$ [and $\mu_\infty \in \R$], we want to compute $v \in \VV_{2d}$ such that $\Vpt[i,j](v) = \mu_{i,j}$ for all $i,j$ [and $\Vpt[\infty] = \mu_{\infty}$]. This may not be possible for all $m$-tuples $\mu = (\mu_{i,j} \; [, \mu_\infty] \; \mid 1 \leq i \leq 3, 0 \leq j < k) \in \R^m$, so we are also interested in the following question:

\emph{For which $m$-tuples $\mu = (\mu_{i,j} \; [, \mu_\infty] \; \mid 1 \leq i \leq 3, 0 \leq j < k) \in \R^m$ does there exist  $v \in \VV_{2d}$ such that $\Vpt[i,j](v) = \mu_{i,j}$ for all $i,j$ [and $\Vpt[\infty]=\mu_\infty$]?}

Note that the reconstructed $v \in \VV_{2d}$ is not uniquely determined, since for any orthogonally equivalent $w \in \VV_{2d}$ (i.e.\ $w = gv$ for some $g \in \G$), the invariants $\Vpt \in \setPT$ take the same values for $v$ and $w$. Theorem~\ref{thm:rationalOrbitSeparation} implies that generically, the reconstructed $v$ is unique up to orthogonal transformations, i.e.\ any different reconstructed $w \in \VV_{2d}$ is orthogonally equivalent to $v$. 

By Proposition~\ref{prop:MainSlice}, for a general point $v \in \VV_{2d}$ there 
exists 
$g \in \G$ such that $u := gv \in \Lambda_{2d}$. Then 
  \[\Vpt(v) = \Vpt(gv) = \Vp(u) \qquad \forall \Vpt \in \setPT, \ \Vp := \restr{\Vpt}{\Lambda_{2d}}.\]
In particular, we can always choose to reconstruct an element $v$ that lies in the subspace $\Lambda_{2d}$.

With respect to the basis $\{w_{i,j} \mid 1 \leq i \leq 3, 0 \leq j < k\} \ [\cup \{w_\infty\}]$ of $\Lambda_{2d}$ from Theorem~\ref{thm:SliceBasisHigher}, we therefore want to determine $\alpha_{i,j}, \Ct[i], \ [\alpha_\infty] \in \R$ such that
  \[v = \sum_{i=1}^3 \Ct[i] w_{i,0} + \sum_{i=1}^3 \sum_{j=1}^{k-1} \alpha_{i,j} w_{i,j} \; [+ \; \alpha_\infty w_\infty] \ \in \Lambda_{2d}\]
satisfies $\Vp[i,j](v) = \mu_{i,j} \ \forall i,j$ [and $\Vp[\infty](v) = \mu_{\infty}$]. With the explicit formulas for the invariants given in Theorem~\ref{thm:InvariantsHigherDeg}, this leads to the problem of solving the following system of polynomial equations in unknowns $\alpha_{i,j}, \Ct[i] \; [, \alpha_\infty] \in \R$:

\begin{equation} \label{eq:polSys}
   \begin{gathered}
            \Ct[1]^2+\Ct[2]^2+\Ct[3]^2 = \mu_{1,0}, \qquad
            \Ct[1] \Ct[2]\Ct[3] = \mu_{2,0}, \qquad
            \Ct[1]^4+\Ct[2]^4+\Ct[3]^4 = \mu_{3,0}, \qquad
            [\alpha_\infty = \mu_\infty], \\[0.5em]
            \begin{pmatrix}
            1 & 1 & 1 \\
            \Ct[1]^2 & \Ct[2]^2 & \Ct[3]^2 \\
            \Ct[1]^4 & \Ct[2]^4 & \Ct[3]^4
            \end{pmatrix} \cdot 
					  M(\alpha) = 
						\begin{pmatrix}
						\mu_{1,1} & \dots & \mu_{1,k-1} \\
            \mu_{2,1} & \dots & \mu_{2,k-1} \\
            \mu_{3,1} & \dots & \mu_{3,k-1}
            \end{pmatrix},
  \end{gathered}
\end{equation}
where the $(i,j)$-th entry of the matrix $M(\alpha)$ is $M(\alpha)_{i,j} = \alpha_{i,j} \Ct[i]^{\xi(j)} \left((\Ct[1]^2-\Ct[2]^2)(\Ct[2]^2-\Ct[3]^2)(\Ct[3]^2-\Ct[1]^2)\right)^{\zeta(j)}$.

The crucial part for the resolution of the polynomial system \eqref{eq:polSys} lies in solving the first three equation for $\Ct[1], \Ct[2], \Ct[3]$. Once values for $\Ct[1], \Ct[2], \Ct[3] \in \R$ are known, we are left with a system of \emph{linear} equations in the remaining variables. Therefore, the following observations are essential:

\begin{lemma} \label{lem:CubicSystem}
  Let $\mu_{1,0}, \mu_{2,0}, \mu_{3,0} \in \C$. If $\Ct[1], \Ct[2], \Ct[3] \in \C$ form a solution of the system
    \[\Ct[1]^2+\Ct[2]^2+\Ct[3]^2 = \mu_{1,0}, \qquad \Ct[1] \Ct[2]\Ct[3] = \mu_{2,0}, \qquad \Ct[1]^4+\Ct[2]^4+\Ct[3]^4 = \mu_{3,0},\]
 then the squares $\Ct[1]^2, \Ct[2]^2, \Ct[3]^2 \in \C$ are the zeroes (with multiplicities) of the cubic polynomial
    \[T^3 - \mu_{1,0}T^2 + \frac{\mu_{1,0}^2 - \mu_{3,0}}{2}T - \mu_{2,0}^2 \in \C[T].\]
\end{lemma}

\begin{proof}
  Let $\Ct[1], \Ct[2], \Ct[3] \in \C$ be solution of the system and let $f \in \C[T]$ be the polynomial whose zeroes (with multiplicities) are $\Ct[1]^2, \Ct[2]^2, \Ct[3]^2$. Then
	\begin{align*}
	  f(T) &= (T-\Ct[1]^2)(T-\Ct[2]^2)(T-\Ct[3]^2) = T^3 - (\Ct[1]^2+\Ct[2]^2+\Ct[3]^2) T^2 + (\Ct[1]^2 \Ct[2]^2+\Ct[2]^2\Ct[3]^2+\Ct[3]^2 \Ct[1]^2) T - \Ct[1]^2 \Ct[2]^2 \Ct[3]^2 \\
				 &= T^3 - \mu_{1,0}T^2 + \frac{\mu_{1,0}^2 - \mu_{3,0}}{2}T - \mu_{2,0}^2.
  \end{align*}
\end{proof}

The following classical fact then characterizes when the solution in Lemma~\ref{lem:CubicSystem} has real solutions:
\begin{lemma}
  A cubic polynomial $f(T) = T^3-aT^2+bT-c \in \R[T]$ has three distinct positive real solutions if and only if
	\[a,b,c > 0 \quad \text{ and } \quad a^2 b^2 - 4 b^3 - 4a^3 c - 27 c^2 + 18abc > 0.\]
\end{lemma}

\begin{proof}
  The expression $a^2 b^2 - 4 b^3 - 4a^3 c - 27 c^2 + 18abc$ is the discriminant of the cubic polynomial $f$, which is positive if and only if $f$ has three distinct real solutions. Descartes' rule of signs implies that $f$ has no negative solutions if and only if the signs of the coefficients of $f$ alternate, i.e.\ $a,b,c > 0$.
\end{proof}

Combining these two results, we obtain Algorithm~\ref{algo:ReconstructTernary}.
\begin{algorithm}[h]
	  \DontPrintSemicolon
		\SetKwInOut{Input}{Input}\SetKwInOut{Output}{Output}
		\Input{$\mu_{i,j} \in \R$ for $1 \leq i \leq 3, 0 \leq j < k$ [plus $\mu_\infty \in \R$]}
		\Output{A ternary form $v \in \VV_{2d}$ such that $\Vpt[i,j](v)=\mu_{i,j}$ for all $i,j$ [and $\Vpt[\infty](v) = \mu_\infty$], \emph{if possible}}
		\BlankLine
		\If{$d$ is a multiple of 3}
		{$\alpha_\infty := \mu_\infty$\;}
		$a := \mu_{1,0}$, $b := \frac{\mu_{1,0}^2 - \mu_{3,0}}{2}$, $c := \mu_{2,0}^2$\;
		\uIf{$a,b,c > 0$ and $a^2 b^2 - 4 b^3 - 4a^3 c - 27 c^2 + 18abc > 0$}
		{
		  Compute the (distinct) roots $r_1, r_2, r_3 > 0$ of the polynomial $T^3-aT^2+bT-c \in \R[T]$. \label{algoStep:findRoots}\;
			$\Ct[1] := \pm \sqrt{r_1}$, $\Ct[2] := \pm \sqrt{r_2}$, $\Ct[3] := \pm \sqrt{r_3}$, where we choose signs such that $\Ct[1] \Ct[2] \Ct[3] = \mu_{2,0}$.\;
			Calculate the matrix product
			  $M := \begin{pmatrix}
	         1 & 1 & 1 \\
					 r_1 & r_2 & r_3 \\
					 r_1^2 & r_2^2 & r_3^2
	       \end{pmatrix}^{-1}
				\begin{pmatrix}
	    \mu_{1,1} & \dots & \mu_{1,k-1}\\
			\mu_{2,1} & \dots & \mu_{2,k-1} \\
			\mu_{3,1} & \dots & \mu_{3,k-1}
	      \end{pmatrix}.$\;
		 \For{$1 \leq i \leq 3, 1 \leq j < k$}
		   {$\alpha_{i,j} := \frac{M_{i,j}}{\Ct[i]^{\xi(j)} \left((r_1-r_2)(r_2-r_3)(r_3-r_1)\right)^{\zeta(j)}}$.\;}
			Output the ternary form $v = \sum_{i=1}^3 \Ct[i] w_{i,0} + \sum_{i=1}^3 \sum_{j=1}^k \alpha_{i,j} w_{i,j} \; [+ \; \alpha_\infty w_\infty] \ \in \Lambda_{2d}$.\;
		}
		\Else{Output \emph{``The values $\mu_{i,j}$ allow no unambiguous reconstruction in $\VV_{2d}$''}. \label{algoStep:UndefTernary}\;}
		\caption{Reconstruction Algorithm.} \label{algo:ReconstructTernary}
\end{algorithm}

For Step~\ref{algoStep:UndefTernary} in 
Algorithm~\ref{algo:ReconstructTernary}, it should be observed that an 
unambiguous reconstruction requires $\Ct[1]^2, \Ct[2]^2, \Ct[3]^2$ to be 
distinct and non-zero. The validity of the algorithm follows from the 
discussion above. 
 
\bibliographystyle{alpha}

\begin{thebibliography}{LBBL{\etalchar{+}}86}

\bibitem[ABR01]{ABR01}
S.~Axler, P.~Bourdon, and W.~Ramey.
\newblock {\em Harmonic function theory}, volume 137 of {\em Graduate Texts in
  Mathematics}.
\newblock Springer-Verlag, New York, second edition, 2001.

\bibitem[AH12]{AH12}
K.~Atkinson and W.~Han.
\newblock {\em Spherical harmonics and approximations on the unit sphere: an
  introduction}, volume 2044 of {\em Lecture Notes in Mathematics}.
\newblock Springer, Heidelberg, 2012.

\bibitem[BP07]{basser-pajevic:07}
P.~Basser and S.~Pajevic.
\newblock Spectral decomposition of a 4th-order covariance tensor :
  Applications to diffusion tensor mri.
\newblock {\em Signal Processing}, 87(2):220--236, February 2007.

\bibitem[CTS07]{CTS07}
J.-L. Colliot-Th\'el\`ene and J.-J. Sansuc.
\newblock The rationality problem for fields of invariants under linear
  algebraic groups (with special regards to the {B}rauer group).
\newblock In {\em Algebraic groups and homogeneous spaces}, volume~19 of {\em
  Tata Inst. Fund. Res. Stud. Math.}, pages 113--186. Tata Inst. Fund. Res.,
  Mumbai, 2007.

\bibitem[CV15]{CV15}
E.~Caruyer and R.~Verma.
\newblock On facilitating the use of {H}ardi in population studies by creating
  rotation-invariant markers.
\newblock {\em Medical Image Analysis}, 20(1):87 -- 96, 2015.

\bibitem[DK15]{DK15}
H.~Derksen and G.~Kemper.
\newblock {\em Computational invariant theory}.
\newblock Springer-Verlag, 2 edition, 2015.

\bibitem[DVW{\etalchar{+}}09]{delmaire-vidailhet-etal:09}
C.~Delmaire, M.~Vidailhet, D.~Wassermann, M.~Descoteaux, R.~Valabregue,
  F.~Bourdain, C.~Lenglet, S.~Sangla, A.~Terrier, R.~Deriche, and
  S.~Leh{\'e}ricy.
\newblock Diffusion abnormalities in the primary sensorimotor pathways in
  writer's cramp.
\newblock {\em Archives of Neurology}, 66(4), 2009.

\bibitem[FK77]{Fox77}
K.~Fox and B.~Krohn.
\newblock Computation of cubic harmonics.
\newblock {\em J. Computational Phys.}, 25(4):386--408, 1977.

\bibitem[GD16]{Ghosh16}
A.~Ghosh and R.~Deriche.
\newblock A survey of current trends in diffusion mri for structural brain
  connectivity.
\newblock {\em J. Neural Eng.}, 13, 2016.

\bibitem[GPD12]{GPD12a}
A.~Ghosh, T.~Papadopoulo, and R.~Deriche.
\newblock Biomarkers for {Hardi}: 2nd \& 4th order tensor invariants.
\newblock In {\em {IEEE International Symposium on Biomedical Imaging: From
  Nano to Macro - 2012}}, Barcelona, Spain, May 2012.

\bibitem[GVL13]{GVL13}
G.~Golub and C.~Van~Loan.
\newblock {\em Matrix computations}.
\newblock Johns Hopkins Studies in the Mathematical Sciences. Johns Hopkins
  University Press, Baltimore, MD, fourth edition, 2013.

\bibitem[GW09]{GRW09}
R.~Goodman and N.~R. Wallach.
\newblock {\em Symmetry, representations, and invariants}, volume 255 of {\em
  Graduate Texts in Mathematics}.
\newblock Springer, Dordrecht, 2009.

\bibitem[GY10]{GY03}
J.~Grace and A.~Young.
\newblock {\em The algebra of invariants}.
\newblock Cambridge Library Collection. Cambridge University Press, Cambridge,
  2010.
\newblock Reprint of the 1903 original.

\bibitem[HK07a]{HK08}
E.~Hubert and I.~Kogan.
\newblock Rational invariants of a group action. {C}onstruction and rewriting.
\newblock {\em J. Symbolic Comput.}, 42(1-2):203--217, 2007.

\bibitem[HK07b]{Hubert07s}
E.~Hubert and I.~Kogan.
\newblock Smooth and algebraic invariants of a group action. {L}ocal and global
  constructions.
\newblock {\em Foundations of Computational Mathematics}, 7(4):355--393, 2007.

\bibitem[HL12]{Hubert12Labahn}
E.~Hubert and G.~Labahn.
\newblock Rational invariants of scalings from {H}ermite normal forms.
\newblock In {\em ISSAC 2012}, pages 219--226. ACM Press, 2012.

\bibitem[HL13]{Hubert13Labahn}
E.~Hubert and G.~Labahn.
\newblock Scaling invariants and symmetry reduction of dynamical systems.
\newblock {\em Foundations of Computational Mathematics}, 13(4):479--516, 2013.

\bibitem[HL16]{Hubert16Labahn}
E.~Hubert and G.~Labahn.
\newblock Computation of the invariants of finite abelian groups.
\newblock {\em Mathematics of Computations}, 85(302):3029--3050, 2016.

\bibitem[Hub12]{Hubert12focm}
E.~Hubert.
\newblock Algebraic and differential invariants.
\newblock In F.~Cucker, T.~Krick, A.~Pinkus, and A.~Szanto, editors, {\em
  Foundations of computational mathematics, Budapest 2011}, number 403 in
  London Mathematical Society Lecture Note Series, pages 168--190. Cambrige
  University Press, 2012.

\bibitem[Isa09]{Isa09}
I.~Isaacs.
\newblock {\em Algebra: a graduate course}, volume 100 of {\em Graduate Studies
  in Mathematics}.
\newblock American Mathematical Society, Providence, RI, 2009.

\bibitem[JBB14]{JohansenBerg14}
H.~Johansen-Berg and T.~Behrens, editors.
\newblock {\em Diffusion MRI. From Quantitative Measurement to In vivo
  Neuroanatomy.}
\newblock Academic Press, second edition. edition, 2014.

\bibitem[Jon11]{Jones11}
D.~Jones, editor.
\newblock {\em Diffusion MRI. Theory, Methods, and Applications}.
\newblock Oxford University Press, 2011.

\bibitem[Kre05]{Kre05}
D.~Kressner.
\newblock {\em Numerical methods for general and structured eigenvalue
  problems}, volume~46 of {\em Lecture Notes in Computational Science and
  Engineering}.
\newblock Springer-Verlag, Berlin, 2005.

\bibitem[LBBL{\etalchar{+}}86]{le-bihan-breton-etal:86}
D.~Le~Bihan, E.~Breton, D.~Lallemand, P.~Grenier, E.~Cabanis, and
  M.~Laval-Jeantet.
\newblock Mr imaging of intravoxel incoherent motions: Application to diffusion
  and perfusion in neurologic disorders.
\newblock {\em Radiology}, 161(2):401--407, 1986.

\bibitem[LMM{\etalchar{+}}10]{li-messe-etal:10}
X.~Li, A.~Mess{\'e}, G.~Marrelec, M.~P{\'e}l{\'e}grini-Issac, and H.~Benali.
\newblock An enhanced voxel-based morphometry method to investigate structural
  changes: application to {Alzheimer's} disease.
\newblock {\em Neuroradiology}, 52(3):203--213, 2010.

\bibitem[Mug72]{Muggli72}
J.~Muggli.
\newblock Cubic harmonics as linear combinations of spherical harmonics.
\newblock {\em Z. Angew. Math. Phys.}, 23:311--317, 1972.

\bibitem[OKA17]{AKO17}
M.~Olive, B.~Kolev, and N.~Auffray.
\newblock {A minimal integrity basis for the elasticity tensor}.
\newblock {\em {Archive for Rational Mechanics and Analysis}}, 2017.

\bibitem[Oli17]{Oli16}
M.~Olive.
\newblock About {G}ordan's algorithm for binary forms.
\newblock {\em Found. Comput. Math.}, 17(6):1407--1466, 2017.

\bibitem[PGD14]{Pap14}
T.~Papadopoulo, A.~Ghosh, and R.~Deriche.
\newblock {Complete Set of Invariants of a 4th Order Tensor: The 12 Tasks of
  HARDI from Ternary Quartics}.
\newblock In Polina Golland, Nobuhiko Hata, Christian Barillot, Joachim
  Hornegger, and Robert Howe, editors, {\em {{Medical Image Computing and
  Computer-Assisted Intervention -- MICCAI 2014}}}, volume 8675 of {\em
  {Lecture Notes in Computer Science}}, pages 233 -- 240, Boston, United
  States, September 2014.

\bibitem[Pop94]{Popov94s}
V.~Popov.
\newblock Sections in invariant theory.
\newblock In {\em The Sophus Lie Memorial Conference (Oslo, 1992)}, pages
  315--361. Scand. Univ. Press, Oslo, 1994.

\bibitem[PV94]{PV94}
V.~Popov and \`E. Vinberg.
\newblock {\em Invariant theory}, volume~55 of {\em Encyclopaedia of
  Mathematical Sciences}.
\newblock Springer-Verlag, Berlin, 1994.
\newblock A translation of {{\i}t Algebraic geometry. 4} (Russian), Akad. Nauk
  SSSR Vsesoyuz. Inst. Nauchn. i Tekhn. Inform., Moscow, 1989, Translation
  edited by A. N. Parshin and I. R. Shafarevich.

\bibitem[Ros56]{Rosenlicht56}
M.~Rosenlicht.
\newblock Some basic theorems on algebraic groups.
\newblock {\em American Journal of Mathematics}, 78:401--443, 1956.

\bibitem[Sch01]{Sch01}
G.~Schwarz.
\newblock Algebraic quotients of compact group actions.
\newblock {\em J. Algebra}, 244(2):365--378, 2001.

\bibitem[Ses62]{Ses61}
C.~Seshadri.
\newblock On a theorem of {W}eitzenb\"ock in invariant theory.
\newblock {\em J. Math. Kyoto Univ.}, 1:403--409, 1961/1962.

\bibitem[SFG{\etalchar{+}}13]{Schultz13}
T.~Schultz, A.~Fuster, A.~Ghosh, R.~Deriche, L.~Florack, and L.~Lek-Heng.
\newblock {Higher-Order Tensors in Diffusion Imaging}.
\newblock In B.~Burgeth, A.~Vilanova, and C.~F. Westin, editors, {\em
  {Visualization and Processing of Tensors and Higher Order Descriptors for
  Multi-Valued Data. Dagstuhl Seminar 2011}}. {Springer}, 2013.

\bibitem[Stu08]{Stu08}
B.~Sturmfels.
\newblock {\em Algorithms in invariant theory}.
\newblock Texts and Monographs in Symbolic Computation. Springer, second
  edition, 2008.

\end{thebibliography}

\newcommand{\etalchar}[1]{$^{#1}$}

\end{document}